\documentclass[11pt]{article}
\usepackage[utf8]{inputenc}
\usepackage{amssymb,amsmath, amsthm, mathabx,mathtools}
\usepackage{amsmath,amsfonts,amssymb,amsthm,epsfig,epstopdf,titling,url,array}
\usepackage{tikz}
\usetikzlibrary{shapes,snakes}
\usepackage{color, enumerate}
\usepackage{comment, authblk}
\usepackage{hyperref}
\usepackage{pgfplots}
\usepackage{bm}
\usepackage{stmaryrd}
\usepackage{graphicx}
\usepackage{caption}
\usepackage{rotating,subfigure}
\usepackage{titlesec}
\usepackage{epstopdf}
\usepackage[flushleft]{threeparttable}
\usepackage{multirow}
\usepackage{sectsty}
\usepackage{lipsum}
\usepackage{algorithm}
\usepackage[noend]{algpseudocode}
\usepackage{algpseudocode}
\usepackage{pifont}

\newcommand{\head}[1]{\textnormal{\textbf{#1}}}

\titleformat*{\subsection}{\large\bfseries}
\numberwithin{equation}{section}

\pgfplotsset{compat=newest}
\pgfplotsset{plot coordinates/math parser=false}
\newlength\figureheight
\newlength\figurewidth

\usepackage{enumerate}
\usepackage{mathrsfs}
\usepackage{wrapfig}

\usepackage{color}

\numberwithin{equation}{section}

\newcommand{\beq}{\begin{equation}}
\newcommand{\bEq}{\end{equation}}

\newcommand{\be}{\begin{equation}}
\newcommand{\ee}{\end{equation}}

\newcommand{\T}{\mathbb T}

\usepackage{amsmath} 
\usepackage{amssymb}
\usepackage{amsthm}

\renewcommand{\epsilon}{\varepsilon}
\renewcommand{\leq}{\leqslant}
\renewcommand{\geq}{\geqslant}

\renewcommand{\le}{\leq}

\theoremstyle{plain} 
\newtheorem{theorem}{Theorem}[section]
\newtheorem*{theorem*}{Theorem}
\newtheorem{lemma}[theorem]{Lemma}
\newtheorem{assumption}[theorem]{Assumption}
\newtheorem*{lemma*}{Lemma}
\newtheorem{corollary}[theorem]{Corollary}
\newtheorem*{corollary*}{Corollary}

\newtheorem*{proposition*}{Proposition}
\newtheorem{claim}[theorem]{Claim}

\newtheorem{definition}[theorem]{Definition}
\newtheorem*{definition*}{Definition}
\theoremstyle{remark}

\newtheorem*{example*}{Example}
\newtheorem{remark}[theorem]{Remark}

\newtheorem*{remark*}{Remark}
\newtheorem*{remarks*}{Remarks}

\date{}

\allowdisplaybreaks

\title{ Multivariate functional responses low rank regression with an application to brain imaging data}
\author[1]{Xiucai Ding \thanks{E-mail: xcading@ucdavis.edu.}}
\author[2]{Dengdeng Yu  \thanks{E-mail: dengdeng.yu@utoronto.ca}}
\author[3]{Zhengwu Zhang  \thanks{E-mail: zhengwu$_{-}$zhang@urmc.rochester.edu}}
\author[2]{Dehan Kong \thanks{Corresponding author; E-mail: dehan.kong@utoronto.ca}}

\affil[1]{Department of Statistics, University of California, Davis}

\affil[2]{Department of Statistical Sciences, University of Toronto}

\affil[3]{Department of Biostatistics and Computational Biology, University of Rochester}

\begin{document}
\maketitle

\begin{abstract}
We propose a multivariate functional responses low rank regression model with possible high dimensional functional responses and scalar covariates. By expanding the slope functions on a set of sieve basis, we reconstruct the basis coefficients as a matrix. To estimate these coefficients, we propose an efficient procedure using nuclear norm regularization. We also derive error bounds for our estimates and evaluate our method using simulations. We further apply our method to the Human Connectome Project neuroimaging data to predict cortical surface motor task-evoked functional magnetic resonance imaging signals using various clinical covariates to illustrate the usefulness of our results.
\end{abstract}

\section{Introduction}
The advancement of neuroimaging technology has produced massive imaging data observed over both time and space, including functional magnetic resonance imaging (fMRI), electroencephalography (EEG), diffusion tensor imaging (DTI), positron emission tomography (PET) and single photon emission-computed tomography (SPECT), among others. Scientists are often interested in characterizing the association between imaging data and clinical predictors. The functional regression models are widely used to achieve this goal, for instance, the functional linear regression \cite{Ramsayet1991, RamsaySilverman2005, Yaoet2005b} and functional response regression model \cite{fara, RamsaySilverman2005}. We refer the readers to \cite{wang2016functional} for a recent review. 

In this paper, we are interested in predicting the blood-oxygen-level-dependent (BOLD) signals obtained from different regions of interest (ROIs) of the brain using clinical covariates. Specifically, we propose the following multivariate functional responses regression model
\begin{equation}
\label{populationmodel}
{\bm Y}(t)=\sum_{j=1}^s X_j \bm{\beta}_j(t)+\bm{\epsilon}(t), 
\end{equation}
where $ {\bm Y}(t)=(Y_1(t),\ldots, Y_p(t))^{\T}\in \mathbb{R}^p$ represents the BOLD signals from $p$ ROIs, $$\bm{\beta}_j(t)=(\beta_{j1}(t), \ldots, \beta_{jp}(t))^{\T}\in \mathbb{R}^p$$ for $1\leq j\leq s$ represent the coefficient functions, characterizing the effect of the $j$th predictor $X_j$ ($1\leq j\leq s$) on the responses, and $\bm{\epsilon}(t)\in \mathbb{R}^p $ is the random error which is independent of $X_j$ ($1\leq j\leq s$). In the current paper, we assume that both ${\bm Y}(t)$ and $X_j$'s are centered and focus on the case when the model \eqref{populationmodel} does not have an intercept term. {Indeed,  our methodology can be extended to the case when they have nonzero known mean functions. We refer the readers to our discussion in Remark \ref{rem_nonzerosmean}.} 

A similar multivariate varying coefficient model (MVCM) \cite{zhu2011fadtts, zhu2012multivariate} has been studied for delineating the association between multiple diffusion properties along
major white matter fiber bundles with a set of covariates of interest. They assume that the error term $\bm{\epsilon}(t)$ can be decomposed into two independent terms, where the first term depicts the error correlations between two time points, and the second term depicts the individual curve variation. Under this assumption, they proposed a weighted least squares procedure based on a local polynomial kernel smoothing technique \cite{fan1996local} to estimate the coefficient functions $\bm{\beta}_j(t).$ They also employed the functional principal component analysis to delineate the structure of the variability in fiber bundle diffusion properties.  

{
There are several key differences between our proposal and the MVCM. First, the task-evoked fMRI data often has non-stationary nature \cite{jones2012non}.  Motivated by this perspective, unlike the MVCM, our error process can cover a wide range of non-stationary processes. Second, in neuroimaging studies, the dimensions of the responses and covariates can be quite large. In the present paper, we allow the dimensions to be divergent with the sample size, while the MVCM considers the case when the dimenions are fixed. Third, the MVCM uses local kernel smoothing method to estimate the coefficient function, which can be computationally slow since it needs to estimate the coefficient functions pointwisely. To overcome this computational difficulty, our method employs the state-of-art sieve regression  which utilizes the global information among all the time points. By imposing a low-rank structure of the coefficient matrix, our proposal can obtain a global fit of the coefficient curves, which significantly improves the computational efficiency.  }

The rest of the article is organized as follows. In Section \ref{sec:modelassumption}, we
introduce our multivariate functional responses low rank regression model, and propose a low-rank estimation procedure with an efficient algorithm. 
Section \ref{theories} investigates the theoretical properties of our method.
Simulations are conducted in Section \ref{simulation} to evaluate the finite sample performance of the proposed approach. Section \ref{realdata} illustrates an application of our method using data from the Human Connectome Project. We end with some discussion in Section \ref{discussion}. Technical proofs, numerical simulation and real data results are given in the supplementary file.

\section{Model setup and estimation procedure}\label{sec:modelassumption}
 Denote $ \{\bm{y}_i(t), (x_{ij}, 1\leq j\leq s), \bm{\epsilon}_i(t): i=1, \ldots, n\}$ independent and identically distributed (i.i.d.) realizations
from the population $\{{\bm Y}(t), (X_j, 1\leq j\leq s), \bm{\epsilon}(t)\}$ generated from the model (\ref{populationmodel}). Without loss of generality, we assume $ t\in [0,1]$. In practice, we cannot observe the entire trajectories of $\{\bm{y}_i(t)\}$. Instead, we can collect intermittent
measurements $\{\bm{y}_i(t_k)\}$ for $ 0\leq t_1\leq t_2\leq \ldots \leq t_T\leq 1$ for each $i$, where $T \in \mathbb{N}$ is the number of time points. In this paper, we assume each subject is observed at the same time points $ t_1, \ldots, t_T$, and this assumption is valid for our fMRI data. In light of model \eqref{populationmodel}, we can write
\begin{equation}\label{eq_model}
\bm{y}_i(t_k)=\sum_{j=1}^s x_{ij} \bm{\beta}_j(t_k)+\bm{\epsilon}_i(t_k), \ k=1,2,\cdots, T, \ i=1,2,\cdots,n. 
\end{equation}
We are interested in estimating the coefficient functions $\{\bm{\beta}_j(t), 1\leq j\leq s\}$. In the current paper, we assume that the functions $\beta_{jl}(t) (1 \leq j \leq s, 1\leq l\leq p)$ are smooth in $t$, which is a realistic assumption for fMRI data. We also allow $ s, p $ and $T$ to diverge with the sample size $n$.

To estimate $\beta_{jl}(t)$, we approximate $\beta_{jl}(t)$ using sieve expansion \cite{chensieve}. Examples of sieve basis include trigonometric series, orthogonal polynomials, and the orthogonal wavelet basis. In particular, according to Section 2.3 of \cite{chensieve}, we have
\begin{equation}\label{eq_betaapprox}
\beta_{jl}(t)=\sum_{h=1}^c M_{jl,h} b_h(t)+\sum_{h=c+1}^{\infty} M_{jl,h} b_h(t) \approx \sum_{h=1}^c M_{jl,h} b_h(t), 
\end{equation} 
where $\{b_h(t)\}_{h=1}^\infty$ is a set of pre-chosen sieve basis functions, $\{M_{jl,h}: 1\leq j\leq s, 1\leq l\leq p, 1\leq h\leq c \} $ are coefficients to be estimated, and $ c $ is the truncation number of sieve basis functions. For simplicity, we use the same $c$ for all $ 1\leq j\leq s$ and $ 1\leq l\leq p$. 

{Plugging} \eqref{eq_betaapprox} into \eqref{eq_model}, we obtain the approximation
\begin{equation}\label{eq_appro_model}
y_{il}(t_k)\approx \sum_{j=1}^s x_{ij} \sum_{h=1}^c M_{jl,h} b_h(t_k)+\epsilon_{il}(t_k),
\end{equation}
for $k=1,2,\cdots, T, \ l=1, 2, \ldots, p, \ i=1,2,\cdots,n$. 
Based on this approximation, the estimation of $\{\bm{\beta}_j(t), 1\leq j\leq s\}$ boils down to estimating $M_{jl,h}$'s. 

Let $\bm{Y}_i \in \mathbb{R}^{p \times T}$ with $lk$th entry $ y_{il}(t_k)$, and $\bm{E}_i\in \mathbb{R}^{p \times T}$ with $lk$th entry $ \epsilon_{il}(t_k)$. Let $\otimes$ be the Kronecker product. We define
\begin{equation}\label{eq_xgenerate}
\bm{X}_i=\bm{x}_i \otimes \bm{B}\in \mathbb{R}^{sc\times T}, 
\end{equation}
 where $\bm{x}_i=(x_{i1}, \cdots, x_{is})^{\T} \in \mathbb{R}^s$ and $\bm{B} \in \mathbb{R}^{c \times T}=(\bm{b}(t_1),\cdots, \bm{b}(t_T))$ with $\bm{b}(t)=(b_1(t),\ldots, b_c(t))^{\T} \in \mathbb{R}^c$. Further, denote $ \bm{M}_j \in \mathbb{R}^{p \times c}$, whose  entry satisfies $(\bm{M}_j)_{lh}=M_{jl,h},\ 1 \leq l \leq p, \  1 \leq h \leq c$, and $\mathsf{M}=(\bm{M}_1, \bm{M}_2, \ldots, \bm{M}_s)\in \mathbb{R}^{p \times sc}$. One can rewrite \eqref{eq_appro_model} as a matrix form
\begin{equation}\label{eq_appro_model_matrixform}
\bm{Y}_i\approx\mathsf{M} \bm{X}_i+\bm{E}_i, \ i=1,2,\cdots,n. 
\end{equation} 
The model \eqref{eq_appro_model_matrixform} is a multivariate response linear regression model, and the parameter of interest is the coefficient matrix $\mathsf{M}$.

The conventional approach to estimate $\mathsf{M}$ is the ordinary least squares (OLS). However, the OLS may perform suboptimally since they do not utilize the information that the entries of $\bm{Y}_i$ are related, especially when both $ p $ and $ T$ diverge with the sample size $ n$. Recently, \cite{yuan2007dimension, chen2013reduced} proposed reduced rank regression models by assuming low-rankness of $\mathsf{M}$. They introduced nuclear norm penalized regression methods to estimate $\mathsf{M}$, which can achieve parsimonious models with
enhanced interpretability. The low-rank assumption has been commonly used in neuroimaging applications, see 
\cite{zhu2014bayesian, zhou2014regularized, Kong2019, yu2020beyond, hu2020matrix, hu2020nonparametric} for example. {In the current paper, we also assume that $\mathsf{M}$ is of low rank. As we will see in the discussion of Section \ref{theories}, the low rank assumption of $\mathsf{M}$ indicates that $\bm{\beta}_{is}, 1 \leq i \leq s$ can be viewed as a finite linear combination of dynamic factors. }

In particular, we solve
\begin{equation}\label{eq_originalmodel}
\min_{\mathsf{M}}\left[\frac{1}{2 nT} \sum_{i=1}^n \text{Tr} \left\{(\bm{Y}_i-\mathsf{M}\bm{X}_i )(\bm{Y}_i-\mathsf{M}\bm{X}_i)^{\T} \right\}+\lambda \|\mathsf{M}\|_*\right],
\end{equation}
where $\|\mathsf{M}\|_*$ is the nuclear norm, defined as the summation of all the singular values of $\mathsf{M}$, and $\lambda$ is a tuning parameter. 

If we let $\mathcal{Y} = \{ \bm{Y}_1^{\T}, \ldots, \bm{Y}_n^{\T}\}^{\T}\in \mathbb{R}^{nT \times p}$, $\mathcal{X} = \{ \bm{X}_1^{\T}, \ldots, \bm{X}_n^{\T}\}^{\T} \in \mathbb{R}^{nT \times sc}$, the optimization problem (\ref{eq_originalmodel}) is equivalent to
\begin{equation}\label{optimization1}
\min_{\mathsf{M}}\left(  \frac{1}{2nT} \| \mathcal{ Y}-\mathcal{X} \mathsf{M}^{\T}
\|_F^2 +\lambda \| \mathsf{M}^{\T}  \|_*\right).
\end{equation}
Denote the solution of (\ref{optimization1}) as $\widehat{\mathsf{M}}$. It is easy to see that the estimate of the coefficient function can be written as $ \widehat{\beta}_{jl}(t)=\sum_{h=1}^c \widehat{M}_{jl,h} b_h(t)$. 

The optimization problem \eqref{optimization1} can be solved by the proximal gradient algorithm. 
For simplicity, we use $\mathbf{D} = \mathsf{M}^{\T}$. Let ${\mathcal{L}(\mathbf{D} ) }= \frac{1}{2nT}\| \mathcal{ Y}-\mathcal{X} \mathbf{D} \|_F^2$ and ${ \mathcal{P}(\mathbf{D} )} = \lambda \| \mathbf{D}   \|_*$.
The objective function {$\mathcal{Q}(\mathbf{D} )$} can be decomposed as { $\mathcal{Q}(\cdot) = \mathcal{L}(\cdot)+\mathcal{P}(\cdot)$ }. Define $\nabla { \mathcal{L}}(\mathbf{S^{(t)}}) = \nabla ||\mathcal{ Y}-\mathcal{X} \mathbf{S^{(t)}}||_F^2 = 2 \mathcal{X}^T (\mathcal{X} \mathbf{S^{(t)}} - \mathcal{Y})$. We utilize the Nestrov's gradient descent method \cite{beck2009fast,nesterov2013introductory} to solve (\ref{optimization1}). In particular, we propose the following algorithm.

\begin{algorithm}[H]\label{algo1}
1. Initialize: $\mathbf{D} ^{(0)} = \mathbf{D}^{(1)} $, $\alpha^{(0)} = 0$ and $\alpha^{(1)} = 1$, $\delta = 1/\lambda_{\rm{max}}(\mathcal{X}^{\T} \mathcal{X})$. \\
2. Repeat:\\
i. $\mathbf{S^{(t)}} = \mathbf{D}^{(t)} + \frac{\alpha^{(t-1)}-1}{\alpha^{(t)}} (\mathbf{D}^{(t)} - \mathbf{D}^{(t-1)})$;\\
ii. $\mathbf{A}_{\rm{temp}} = \mathbf{S^{(t)}} - \delta \nabla { \mathcal{L}}(\mathbf{S^{(t)}})$;\\
iii. Singular value decomposition: $\mathbf{A}_{\rm{temp}} = \mathbf{U} \operatorname{diag}(\mathbf{a}) \mathbf{V}^T$;\\
iv. $\mathbf{d} = (\mathbf{a} - \lambda \delta \cdot \mathbf{1})_{+}$;\\
v. { $\mathbf{D}^{(t+1)} = \mathbf{U} \operatorname{diag}(\mathbf{d}) \mathbf{V}^T$; }\\
vii. $\alpha^{(t+1)} = \left[1 + \sqrt{1+(2 \alpha^{(t)})^2}\right]/2$;\\
3. Until objective function ${\mathcal{Q}}(\mathbf{D}^{(t)})$ converges.\\ 
\end{algorithm}
From step ii to step v, the gradient descent is based on the first order approximation to the loss function {$\mathcal{L}$} at the current search point $\mathbf{S^{(t)}}$. Specifically, 
{
\begin{align*}
g(\mathbf{D}| \mathbf{S^{(t)}},\delta) =& \mathcal{L}(\mathbf{S^{(t)}}) + \langle \nabla  \mathcal{L}(\mathbf{S^{(t)}}),\mathbf{D} - \mathbf{S^{(t)}} \rangle + \frac{1}{2 \delta} ||\mathbf{D} - \mathbf{S^{(t)}}||^2_F + \mathcal{P}(\mathbf{D}), \\
& = \frac{1}{ 2 \delta} ||\mathbf{D} - (\mathbf{S^{(t)}}- \delta \nabla  \mathcal{L}(\mathbf{S^{(t)}} ))||_F^2 + \mathcal{P}(\mathbf{D}) + c^{(t)},
\end{align*}
}
where $c^{(t)}$ collects the term irrelevant to the optimization and the constant $\delta$ is chosen such that the relation between the surrogate and target functions always holds: $g(\mathbf{D}| \mathbf{S^{(t)}},\delta) \geq { \mathcal Q}(\mathbf{D})$.
We set $\delta$ as a Lipschitz constant for $\nabla {\mathcal{L}}(\cdot)$ with $\delta = 1/\lambda_{\rm{max}}(\mathcal{X}^{\T} \mathcal{X})$.
Solution to the surrogate optimization problem is given by  \emph{Proposition 1} of \cite{zhou2014regularized}. 
Singular value decomposition is performed on the intermediate matrix $\mathbf{A}_{\rm{temp}} = \mathbf{S^{(t)}} - \delta \nabla { \mathcal{L}}(\mathbf{S^{(t)}})$.
The next iterate  $\mathbf{D}^{(t+1)}$ shares the same singular vectors as $\mathbf{A}_{\rm{temp}}$ and its singular values $\mathbf{d}^{(t+1)}$ are determined by minimizing 
$\frac{1}{2 \delta} ||\mathbf{d} -\mathbf{a}||_2^2 + f(\mathbf{d})$, where $\mathbf{a} = \sigma(\mathbf{A}_{\rm{temp}})$.
For the nuclear norm regularization $f(\mathbf{d}) = \lambda \sum_j |b_j| $, the solution is given by soft thresholding the singular values $b_j^{(t+1)} = (a_j -\lambda \delta )_+$ as suggested by \emph{Corollary 1} of \cite{zhou2014regularized}. 

There are two tuning parameters involved in our estimation procedure, the truncation number $ c $ and the regularization parameter  $\lambda$. In this paper, we use $5$-fold cross-validation to select the optimal values based on two-dimensional grid search.

\section{Theoretical Results}\label{theories}
We begin with some notation. Define $\bm{z}=(z_1,\cdots, z_s)^{\T}$ as a \emph{subgaussian} random vector {with some parameter $\sigma>0$} if  for all $\bm{\alpha} \in \mathbb{R}^s$,
\begin{equation*}
{
\mathbb{E} \left[ \exp(\bm{\alpha}^{\T} \bm{z}) \right] \leq \exp(\| \bm{\alpha} \|^2 \sigma^2 /2). }
\end{equation*}
We next introduce the  \emph{locally stationary time series.} Consider the time series \cite{WZ1, WZ2} 
\begin{equation}\label{defn_model}
z_i=G(\frac{i}{n}, \mathcal{F}_i),
\end{equation}
where $\mathcal{F}_i=(\ldots, \eta_{i-1}, \eta_i)$ and $\eta_i, \ i  \in \mathbb{Z}$ are  i.i.d.  random variables, and $G:[0,1] \times \mathbb{R}^{\infty} \rightarrow \mathbb{R}$ is a measurable function such that $\xi_i(t):=G(t, \mathcal{F}_i)$ is a properly defined random variable for all $t \in [0,1].$ We introduce the following dependence measure to quantify the temporal dependence of (\ref{defn_model}).
\begin{definition}
\label{defn_physical} 
Let $\{\eta_i^{\prime}\}$ be an i.i.d. copy of $\{\eta_i\}.$ We assume that for some $q>2,\ \|x_i\|_q<\infty, $ where $\| \cdot \|_q=[\mathbb{E} |\cdot|^q ]^{1/q}$ is the $\mathcal{L}_q$ norm of a random variable.  For $j \geq 0,$ we define the physical dependence measure by 
\begin{equation}\label{eq_phyoriginal}
\delta(j,q):=\sup_{t \in [0,1]} \max_{i} \| G(t, \mathcal{F}_i)-G(t, \mathcal{F}_{i,j})\|_q,
\end{equation}
where $\mathcal{F}_{i,j}:=(\mathcal{F}_{i-j-1}, \eta^{\prime}_{i-j},\eta_{i-j+1}, \ldots, \eta_i).$ 
\end{definition} 
The measure $\delta(j, q)$ quantifies the changes in the system's output when the input of the system
$j$ steps {before} is changed to an i.i.d. copy. If the change is small, then we have short-range
dependence. 

With the above notation, we introduce the assumptions for the theoretical development. 

\begin{assumption}\label{assum_varaible} {We assume that $\bm{x}_i, i=1,2,\cdots,n,$ are  i.i.d. centered subgaussian random vectors independent of $\bm{\epsilon}_i(t_k), i=1,2,\cdots,n,$ for all $k=1,2, \cdots,T$.}  Moreover, we assume that for each $i=1,2,\cdots,n,$ and $l=1,2,\cdots,p,$ $\{\epsilon_{il}(t_k)\}_{k=1}^T$ is a centered locally stationary time series of the form (\ref{defn_model}). Finally, for some large $q, \gamma>0,$ there exists some universal constant $C>0$, such that
\begin{equation}\label{assum_phy}
\delta(j,q) \leq Cj^{-\gamma},  \ j \geq 1.
\end{equation} 
\end{assumption}

{We mention that the assumption (\ref{defn_model}) represents a wide class of stationary, locally stationary linear and nonlinear processes  \cite{WZ1, WZ2}. As we mentioned earlier, previous works \cite{zhu2011fadtts, zhu2012multivariate} focus on fitting the coefficients of functional regression locally. Hence, they do not need to consider the temporal relation for the underlying stochastic process. In contrast, our estimation relies on (\ref{eq_appro_model}), which utilizes the global information for all the time points. A natural assumption is the short-range temporal dependence, i.e., (\ref{assum_phy}), which needs that the temporal correlation between the process $\bm{\epsilon}_i(\cdot)$ has a polynomial decay. Moreover, as a technical byproduct, we only require the existence of second moment of $\epsilon_{il}(\cdot)$. This improves the assumption of finite fourth moment in \cite{zhu2012multivariate}. Finally, we mention that (\ref{assum_phy}) can be satisfied by many stochastic processes, for instance, the Ornstein-Uhlenbeck process and the linear process $$\epsilon_{il}(t)=\sum_{k=1}^\infty a_{k,il}(t) \upsilon_i,$$ where $\{\upsilon_i\}$ are independent standard Gaussian random variables and  $\sup_t |a_{k,il}(t)|^2 \leq C k^{-\gamma},$ for some constant $C>0.$ 
}

We also need the following assumption on the smoothness of $\beta_{jl}(\cdot)$'s. 
\begin{assumption}\label{assum_smooth} For $j=1,2,\cdots,s, l=1,2,\cdots, p, $ $\beta_{jl}(\cdot)$'s are smooth functions of time such that $\beta_{jl}(\cdot) \in C^d([0,1]),$ where $C^d([0,1])$ ) is the function space on $[0, 1]$ of continuous functions that have continuous first $d$ derivatives.
\end{assumption}

By Assumption \ref{assum_smooth}, $\beta_{jl}(t)$ can be well approximated by sieve expansion \cite{chensieve}.  Specifically, in light of (\ref{eq_betaapprox}), we have that 
\begin{equation}\label{eq_betaexpan}
\beta_{jl}(t)=\sum_{h=1}^c M_{jl,h} b_h(t)+O(c^{-d}), \ j=1,2,\cdots,s,
\end{equation} 
where the error $O(c^{-d})$ is entrywise. 
Plugging (\ref{eq_betaexpan}) into (\ref{eq_model}), with high probability, we have
\begin{equation}\label{eq_scalarmodel}
\bm{y}_i(t_k)=\sum_{j=1}^s x_{ij} \bm{M}_j \bm{b}(t_k)+\bm{\epsilon}_i(t_k)+O(sc^{-d}).
\end{equation}

When the error term $O(sc^{-d})$ is negligible, we can approximate $\bm{\beta}_j(t_k)$ using
\begin{equation}\label{eq_approximatebetaone}
\tilde{\bm{\beta}}_j(t_k)=\bm{M}_j \bm{b}(t_k).
\end{equation}
{For a rigorous justification, we refer the readers to Theorem \ref{cor_main} and Corollary \ref{coro_coro} and their proofs. }

Recall $\bm{Y}_i \in \mathbb{R}^{p \times T}$ is a matrix with $lk$-th entry $ y_{il}(t_k)$, and $\bm{E}_i\in \mathbb{R}^{p \times T}$ with $lk$-th entry $ \epsilon_{il}(t_k)$. We can write the model (\ref{eq_model}) as 
\begin{equation}\label{eq_finalmodel1}
\bm{Y}_i=\mathsf{M} \bm{X}_i+\bm{E}_i+o(1), \ i=1,2,\cdots,n,
\end{equation} 
where $\mathsf{M} $ is defined under equation (\ref{eq_xgenerate}). {Therefore, our estimation problem boils down to estimating the coefficient matrix $\mathsf{M} $. }


{
\begin{remark}
Since $\widetilde{\bm{\beta}}_j(t)$ can approximate $\bm{\beta}_j$ well, we now connect the structure of $\widetilde{\bm{\beta}}_j(t)$ with the matrix $\mathsf{M}$ to show that $\bm{\beta}_j$ will have a dynamic factor model structure when $\mathsf{M}$ is approximately low rank. Note $\mathsf{M} $ is a rectangular matrix stacking $\bm{M}_j \in \mathbb{R}^{p \times c}, j=1,2,\cdots, s.$ For each $j=1,2,\cdots,s,$ we write the singular value decomposition of $\bm{M}_j$ as
\begin{equation*}
\bm{M}_j=\sum_{l=1}^{\min \{p,c\}} \sigma_l \mathbf{u}_l \mathbf{v}_l^\top,
\end{equation*}
where $\{\sigma_l\}, \{\mathbf{u}_l\}$ and $\{\mathbf{v}_l\}$ are the singular values, left singular vectors and right singular vectors of $\bm{M}_j,$ respectively.  
Consequently, we find that 
\begin{equation*}
\widetilde{\bm{\beta}}_j(t_k)=\sum_{l=1}^{\min \{p,c\}} \sigma_l (\mathbf{v}_l^\top \mathbf{b}(t_k))\mathbf{u}_l.      
\end{equation*}
If we further denote $\alpha_l(t_k):=\sigma_l(\mathbf{v}_l^\top \mathbf{b}(t_k)),$ then $\widetilde{\bm{\beta}}_j(t_k)$ can be further written as 
\begin{equation*}
\widetilde{\bm{\beta}}_j(t_k)=\sum_{l=1}^{\min\{p,c\}}\alpha_l(t_k) \mathbf{u}_l.      
\end{equation*}

This implies that $\widetilde{\bm{\beta}}_j(t_k)$ is a time-varying linear combination of the basis $\{\mathbf{u}_l\},$ as which we can regard a dynamic factor model. In the current paper, we follow the common low-rank assumption in the literature of approximate factor models \cite{bai2002determining} and assume that only a few of $\{\mathbf{u}_l\}$ are useful for our estimation and prediction. As a result, $\bm{M}_j$ is of low-rank structure.

Suppose that the rank of $\bm{M}_j$ is $r_j,j=1,2,\cdots,s.$ Since 
\begin{equation*}
\text{rank}(\mathsf{M})  \leq \sum_{j=1}^s r_j,   
\end{equation*}
and $s$ is slowly divergent, we can assume that $\mathsf{M}$ is of approximate low-rank structure. This is formally stated in Assumption \ref{assum_lowrank}. 
\end{remark}
}

Denote 
\begin{equation}\label{eq_defnxisieve}
\xi:=\sup_{1 \leq h \leq c} \sup_{t \in [0,1]} |b_h(t)|.
\end{equation}
As mentioned in Section 4.2 of \cite{ding2019}, $\xi$ can be well controlled for the commonly used sieve basis functions. For instance, $\xi=O(1)$ for the trigonometric series and orthogonal polynomials and $\xi=O(\sqrt{c})$ for the orthogonal wavelet basis.

 \begin{assumption}\label{assum_lowrank} We assume that $\mathsf{M}$ is of approximately low-rank structure, i.e., there exists a constant $\kappa>0$ such that
$$ \sum_{i=1}^{\min\{p,sc\}} \sigma_i(\mathsf{M}) \leq \kappa,$$ where $\sigma_i(\mathsf{M}), i=1,2,\cdots, \min\{p,sc\}$ are the singular values of $\mathsf{M}.$  Moreover,  we assume that for {any arbitrarily} small constant $\tau>0,$ we have that
 \begin{equation}\label{eq_boundone}
\frac{ \sqrt{r} \xi p n^{\tau}}{c s\sqrt{T}}=o(1), \ \text{where} \ r=\text{rank}(\mathsf{M}). 
 \end{equation} 

 \end{assumption}
 
\begin{remark} 
The assumption in (\ref{eq_boundone}) is mild. Denote
\begin{equation}\label{eq_discussion}
c=O(n^{\alpha_1}), T=O(n^{\alpha_2}), p=O(n^{\alpha_3}), s=O(n^{\alpha_4}).
\end{equation}
If we choose the trigonometric series or orthogonal polynomials,  (\ref{eq_boundone}) reads as 
\begin{equation*}
r^{1/2} n^{\alpha_3-\alpha_2/2+\tau-\alpha_1-\alpha_4}=o(1).
\end{equation*}
In other words, when the true rank $r$ is finite and $p=O(sc)$, we only need to have $O(n^{2\tau})$ time points observed from the stochastic process $\bm{\epsilon}_i(\cdot)$. 
\end{remark}

To guarantee the consistency of the estimation, we need the following assumption on the parameters.  
  \begin{assumption}\label{assump_parameter} We assume that 
  \begin{equation}\label{eq_boundtwo}
 \xi p s^2 n^{2\tau} c^{-d}=o(1).
  \end{equation}
 \end{assumption} 

The assumption \ref{assump_parameter} is mild. When we use the trigonometric series or orthogonal polynomials, and assume $\bm{\beta}_j(\cdot)$ is infinitely differentiable, (\ref{eq_boundtwo}) will always hold and we can allow $ps^2$ diverging fast. {In our paper, we need our sieve bases satisfy (\ref{eq_boundone}) and (\ref{eq_boundtwo}) and belong to $C^d([0,1])$ defined in Assumption \ref{assum_smooth}. Note that the parameter $\xi$  in (\ref{eq_boundone}) and (\ref{eq_boundtwo}) is directly related to the sieve bases via (\ref{eq_defnxisieve}). Indeed, all the sieve bases listed in Section 2.3 of \cite{chensieve} satisfy these assumptions. For instance, the Fourier basis, the orthogonal polynomials, the Daubenchies orthogonal wavelets and the splines.  }
 
Finally, we introduce the following assumption to guarantee that the covariance matrix of $\bm{x}_i$ is regular.  We will see later that the following condition is a sufficient condition for the restricted strong convexity condition (c.f. Definition \ref{defn_rsc}). 

 \begin{assumption}\label{ass_pdc}  Denote $\Sigma_s$ as the covariance matrix of $\bm{x}_i.$ We assume that $\Sigma_s$ is bounded and there exists some constant $\delta>0$ such that 
 \begin{equation*}
 \lambda_{\min}(\Sigma_s) \geq \delta,
 \end{equation*}
 where $\lambda_{\min}(\Sigma_s)$ is the smallest eigenvalue of $\Sigma_s.$
 \end{assumption}
 
Armed with the above assumptions, we now present our main result. Denote $\lambda_n$ as the regularization parameter of the optimization problem (\ref{optimization1}), and $\mathsf{M}^*$ the true value of $\mathsf{M}$. Recall that $\widehat{\mathsf{M}}$ is the solution of (\ref{optimization1}), we have the following result. {It can be seen that even though our approach may be suboptimal, it can achieve consistency under mild conditions.}

\begin{theorem}\label{cor_main}
Suppose Assumptions \ref{assum_varaible}--\ref{ass_pdc} hold. For any given {arbitrarily} small constant $\tau>0$ defined in Assumption \ref{assum_lowrank}, when both $n$ and $T$ are  large enough, there exists some $C_q>0$ depending on $q$ in Assumption \ref{assum_varaible}, with probability at least $1-C_qn^{-q \tau}$, we have for some constants $C, C_1>0,$ when $\lambda_n \geq C_1 p \xi n^{\tau} T^{-1/2},$ 
\begin{equation}\label{eq_boundbound}
\|\mathsf{M}^*-\widehat{\mathsf{M}}  \|_F  \leq C \left( \frac{\sqrt{p \xi } n^{\tau/2}}{ \sqrt{c \operatorname{tr}(\Sigma_s)} T^{1/4}}+ \frac{\sqrt{r} p \xi n^{\tau}}{c \operatorname{tr}(\Sigma_s) \sqrt{T}} \right).
\end{equation}
\end{theorem}

{One thing to note here is that $T$ also diverges with the sample size $n$. Since the role of $\tau$ is to control the probability and can be arbitrary, we can obtain a consistent estimator for a reasonably large $T$. Specially, in the setting of (\ref{eq_discussion}), our estimator is consistent under Assumption \ref{assump_parameter}. }

\begin{remark}
In our theoretical development, we borrow the idea of the regularized M-estimator developed in \cite{negahban2012unified}. However, one main challenge of our proof is that we need to account for the approximation errors brought by truncation of the basis expansion in \eqref{eq_betaapprox}. 

We next provide some insights of the above results when we use either the trigonometric series or the orthogonal polynomials.  From Assumption \ref{ass_pdc}, we have $\operatorname{tr}(\Sigma_s)=O(s).$ Hence, by Assumption \ref{assum_lowrank}, the second term of the right-hand side of (\ref{eq_boundbound}) is of order $o(1).$ On one hand, if the second term of the right-hand side of (\ref{eq_boundbound}) dominates the first one,
\begin{equation}
\frac{\sqrt{r} p \xi n^{\tau}}{c \operatorname{tr}(\Sigma_s) \sqrt{T}}>\frac{\sqrt{p \xi } n^{\tau/2}}{ \sqrt{c \operatorname{tr}(\Sigma_s)} T^{1/4}},
\end{equation}
which implies $\sqrt{r} \sqrt{p} n^{\tau/2}>\sqrt{cs} T^{1/4}.$ If we further let the matrix be of exactly low rank and the functions $\bm{\beta}_j(\cdot)$ be infinitely differentiable such that $c$ can be chosen at an order of $\log T$,  then we obtain an upper bound for number of time points as 
\begin{equation*}
T \ll n^{2 \tau} \left( \frac{p}{s} \right)^2.
\end{equation*}
On the other hand, when the first term of the right-hand side dominates, we need to have that 
\begin{equation*}
\frac{\sqrt{p} n^{\tau/2}}{\sqrt{cs} T^{1/4}}=o(1), 
\end{equation*}
which basically requires that 
\begin{equation*}
T \gg n^{2 \tau} \left( \frac{p}{s} \right)^2.
\end{equation*}
In this sense, the choice of $T$ will not significantly affect the consistency of our estimators.  Since our assumptions are mild as explained in Section \ref{sec:modelassumption}, once $T$ and $n$ are reasonably large, we obtain a consistent estimator. 
\end{remark}

 Let $\widehat{\mathsf{M}}=(\widehat{\bm{M}}_1, \cdots, \widehat{\bm{M}}_s).$ Denote the estimator 
\begin{equation*}
    \widehat{\bm{\beta}}_j(t)=\widehat{\bm{M}}_j \bm{b}(t), \ 1 \leq j \leq s, t \in [0,1]. 
\end{equation*}
We now state the convergence result for the coefficient functions.
\begin{corollary}\label{coro_coro} Suppose that the assumptions of Theorem \ref{cor_main} hold. Then for $1 \leq j \leq s$ and some universal constant $C>0,$ with probability at least $1-C_q n^{-q \tau},$ we have
\begin{equation*}
\sup_t \|\bm{\beta}_j(t)-\widehat{\bm{\beta}}_j(t) \| \leq C \sqrt{c} \left( \frac{\sqrt{p \xi } n^{\tau/2}}{ \sqrt{c \operatorname{tr}(\Sigma_s)} T^{1/4}}+ \frac{\sqrt{r} p \xi n^{\tau}}{c \operatorname{tr}(\Sigma_s) \sqrt{T}} \right).   
\end{equation*}
\end{corollary}


Compared to Theorem \ref{cor_main},
we get an extra $\sqrt{c}$ factor in Corollary \ref{coro_coro} since our estimate involves the sieve basis functions. Similarly, we can obtain a consistent estimator when both $T$ and $n$ are reasonably large.

{
\begin{remark}\label{rem_nonzerosmean}
It is remarkable that in the high dimensional setting, it is not trivial to center the high-dimensional responses. However, in many applications, we can assume that there exists a time-varying mean function for each $Y_l(t), \ 1 \leq l \leq p.$ 
Specifically, we can assume that for some functions $m_i(\cdot) \in C^d([0,1])$ such that 
\begin{equation*}
\mathbb{E}(Y_l(t))=m_l(t), \ 1 \leq l \leq p. 
\end{equation*}
In this setting, we can rewrite our model (\ref{populationmodel}) as
\begin{equation*}
\bm{Y}(t)-\bm{m}(t)=\sum_{j=1}^s X_j \bm{\beta}_j(t)+\bm{\epsilon}(t),    
\end{equation*}
 where $\bm{m}(t)=(m_1(t),\cdots, m_p(t)).$ Moreover, if we set $\bm{\beta}_0(t)=\bm{m}(t)$ and $X_0=1,$ we can further write 
 \begin{equation*}
 \bm{Y}(t)=\sum_{j=0}^s X_j \bm{\beta}_j+\bm{\epsilon}(t).
 \end{equation*}
 
 Since $m_l(t) \in C^d([0,1]), l=1,2,\cdots, p,$ we can expand them on a set of basis functions  and $m_l(t) \approx \sum_{h=1}^c \kappa_{lh} b_h(t).$ Therefore, we can apply our current methodology to estimate the coefficients. 
\end{remark}

}

\section{Simulations}\label{simulation}
In the section, we perform simulation studies to evaluate our method. We consider a set of Fourier basis
 \[
    b_j(t)=\left\{
                \begin{array}{ll}
                  1, & \textrm{if } j = 1;\\
                  \sqrt{2} \sin ( \pi j t), & \textrm{if } j \textrm{ is even;}\\
                  \sqrt{2} \cos ( \pi (j-1) t), &  \textrm{Otherwise}.
                \end{array}
              \right.
  \]

The $\bm{x}_i$ is generated from a multivariate normal distribution with mean zero and covariance $\Sigma$ with $j_1j_2$th entry $\Sigma_{j_1j_2}=0.5^{|j_1-j_2|}$ for $ 1\leq j_1\leq j_2\leq s$. The $\bm{X}_i=\bm{x}_i \otimes \bm{B}$, where $\bm{B} \in \mathbb{R}^{c \times T}=(\bm{b}(t_1),\cdots, \bm{b}(t_T))$ with $\bm{b}(t)=(b_1(t),\ldots, b_c(t))^{\T} \in \mathbb{R}^c$. Here, we set $ t_k=\frac{k-1}{T}$ for $ k=1, \ldots, T$.  

The response is generated from
\begin{equation*}
\mathbf{Y}_i=\mathsf{M}\mathbf{X}_i+\nu \cdot \mathbf{E}_i, \ i=1,2,\cdots,n,
\end{equation*}
where $\nu$ is a constant, and each row of $\mathbf{E}_i \in \mathbb{R}^{p \times T}$ is a time series with an autoregressive structure. In particular, let $\mathbf{E}_{ij} \in \mathbb{R}^{T}$ be the $j$th row of $\mathbf{E}_{i}$ and $E_{ijk}=\mathbf{E}_{ij}(\frac{k}{T})$ as the $jk$th entry of $\mathbf{E}_i$ for $j = 1, \ldots, p$ and $k=1, \ldots, T$. We set $\mathbf{E}_{ij}(\frac{k}{T}) = 0.3 \mathbf{E}_{ij}(\frac{k-1}{T}) + \bm\varepsilon_{i}(\frac{k}{T})$, with $\mathbf{E}_{ij}(0) = 0$ for all $j = 1, \ldots, p$, where $\bm\varepsilon_i(\frac{k}{T})$
 is a series of i.i.d Gaussian random variables with mean $0$ and variance $1$. We consider $ (n, p, T, c, s)=(100, 32, 256, 4, 8)$. In this case, the matrix $\mathsf{M}\in \mathbb{R}^{32\times 32}$, and we consider three different shapes for $\mathsf{M}$: a square shape, a T shape and a cross shape, shown in Figure \ref{sim1snr1figure} (a) (d) (g). The true coefficient functions are generated from $\beta_{jl}(t)=\sum_{h=1}^c M_{jl,h} b_h(t)$.  
 
We define the signal to noise ratio (SNR) as 
$$\mathrm{SNR} = 
\frac{\sum_{i=1}^n \text{Tr}(\mathbf{X}_i^* \mathsf{M}^* \mathsf{M} \mathbf{X}_i)}{\sum_{i=1}^n \nu^2 \text{Tr}(\mathbf{E}_i^* \mathbf{E}_i)}.$$
We consider three cases $\mathrm{SNR}=1, 5, 10 $, where we change $ \nu $ to obtain different SNRs. 


To fit the model, we consider two sets of sieve basis. One is the same set of Fourier basis, and we evaluate the performance of our method if one can choose the basis correctly. In practice, we may not know the true basis, therefore, we also fit a different basis when applying our procedure, to reflect the scenarios where the true underlying basis does not align with the fitted basis. In the simulation, we consider the Chebyshev basis of second kind. In particular,
the basis is defined as
\[
    b_j(t)=\left\{
                \begin{array}{ll}
                  2\cdot (1-[2(t-1/2)]^2)^{1/4}/\sqrt{\pi}, & \textrm{if } j = 1;\\
                  2 t b_1(t), & \textrm{if } j = 2;\\
                  2t b_{j-1}(t) - b_{j-2}(t), &  \textrm{Otherwise}.
                \end{array}
              \right.
  \]
  
For each case, we report the mean integrated squared errors (MISEs) of the estimates of $\bm\beta_j(\cdot) \in \mathbb{R}^p$, $j = 1, \ldots, s$ defined as
\begin{equation*}
\mathrm{MISE}_j := \frac{1}{p}\sum_{l=1}^{p}\int_0^1 \left( \beta_{jl}(t) - \hat{\beta}_{jl}(t) \right)^2 dt.
\end{equation*}
We also compare with the ordinary least squares, where we set $\lambda=0 $ in \eqref{eq_originalmodel} and solve the optimization problem. All the results are based on 100 Monte Carlo run.


We include the cases (SNR = 5), where the true basis is the Fourier basis and the fitted basis is also Fourier basis in Table \ref{sim1:snr5:mise:basis22:tab}, and the true basis is the Fourier basis with the fitted basis Chebyshev of second kind in Table \ref{sim1:snr5:mise:basis21:tab}. For SNR = 1 and SNR = 10, the results are included in Tables \ref{sim1:snr1:mise:basis22:tab}  
to \ref{sim1:snr10:mise:basis21:tab} in the supplementary material. 
In particular, the MISEs for 8 functional slope estimates for our proposed methods are smaller than those for OLS methods. As expected, when the true basis is Fourier, fitting using Fourier basis results in better estimation accuracy (smaller MISEs) compared to using Chebyshev basis.
We have also plotted the estimated $\widehat{\mathsf{M}}$ from one randomly selected Monte Carlo run in Figure \ref{sim1snr1figure} for SNR = 1 with Fourier basis fit. From the results, we can see that our estimates can achieve much better estimation accuracy for those coefficient functions compared with OLS.

In addition, we also perform a simulation study where the true basis is the Chebyshev basis of second kind defined in previous paragraph. The results of fitting our method using both Chebyshev basis of second kind and Fourier basis are included in Tables 
\ref{sim1:snr1:mise:basis11:tab} 
to \ref{sim1:snr10:mise:basis12:tab} in the supplementary material. The findings are similar.

When the fitted basis and the true basis align with each other, we also report the average number of basis selected using using 5-fold cross validation in Table \ref{sim1:c:tab}. As shown from the results, if we know the true basis, cross validation can select the right number of basis for most scenarios. When SNR increases, the average number of basis selected gets closer to the truth ($c=4$).

{To investigate how the truncation number $ c $ affects the performance of the proposed method, we add a simulation study, where 4 Fourier basis is used to generate the data, but we fit our model by setting $ c=6 $, and the $ \lambda $ is still chosen by 5-fold cross validation. Compared with the result where $ (c, \lambda) $ are chosen by 5-fold cross validation, we find the rank of estimated  $\widehat{\mathsf{M}}$ becomes smaller. 
The results are as shown in Tables \ref{sim1v1c6:snr1:mise:basis22:tab}, \ref{sim1v1c6:snr5:mise:basis22:tab} and \ref{sim1v1c6:snr10:mise:basis22:tab} in the supplementary material. Taking SNR = 5 for example (Table \ref{sim1v1c6:snr5:mise:basis22:tab}), the average ranks of estimated $\widehat{\mathsf{M}}$ when $c=6$ ($1.00$, $2.00$ and $2.00$) are smaller than the average ranks of $\widehat{\mathsf{M}}$ when $c$ is determined by cross validation ($9.79$, $13.41$ and $13.56$). However, the MISEs of the estimated coefficient functions using $c=6$ are actually greater than the MISEs of the estimated functional slopes from using cross validation, which shows the importance of using cross-validation to select the truncation number $ c $. 

In addition, we perform a simulation study, where a set of equivalent basis (Chebyshev2) is used in fitting. The results are included in Table \ref{sim1v1c6:mise:basis21:tab} in the supplementary material. We find that the average ranks of estimated $\widehat{\mathsf{M}}$ are smaller than the average ranks of $\widehat{\mathsf{M}}$ when the true basis (Fourier) is used. Taking SNR = 5 (Table \ref{sim1v1c6:mise:basis21:tab}) for example, the average ranks of estimated $\widehat{\mathsf{M}}$ when the equivalent basis is used ($1.00$, $2.00$ and $2.00$) are smaller than the average ranks of $\widehat{\mathsf{M}}$ when $c$ is determined by cross validation ($9.79$, $13.41$ and $13.56$). We have found that the MISEs obtained by fitting using the Chebyshev2 basis is still reasonably small, which shows the robustness of proposed method when an equivalent basis is used. 

To mimic the case where the coefficient functions lie in an infinite dimensional space, 
we add an additional simulation study with a modified setting $ (n, p, T, c, s)=(100, 32, 256,$ $ 50, 4)$, 
where the coefficient function $\beta_{il}(t)$'s are generated from $50$ basis functions such that $\beta_{jl}(t)=\sum_{h=1}^8 M_{jl,h}  \omega_h b_h(t) + \sum_{h=9}^{50} w_h b_h(t)$, where $\omega_1 = 1, \omega_2 = 0.8, \omega_3 = 0.6, \omega_4 = 0.5$, and $\omega_h = 8(h-2)^{-4}$ for $h \geq 5$. 
We consider three cases $\mathrm{SNR} =1, 5, 10$. As shown in Table \ref{sim1v7:c:tab} of the supplementary material, the number of basis selected by cross validation is much smaller than $50$ due to the decay of $\omega_h$s. Taking SNR = 5 for example, when the Fourier basis is used, the average number of basis selected for the T shape is $4.790$ with a standard error $0.041$. We include the cases (SNR = 5), where the true basis is the Fourier basis and fit our method using the Fourier basis in Table \ref{sim1v7:snr5:mise:basis22:tab} and the Chebyshev basis of second kind in Table \ref{sim1v7:snr5:mise:basis21:tab} in the supplementary material. For SNR = 1 and SNR = 10, the results are included in Tables 
\ref{sim1v7:snr1:mise:basis21:tab},
\ref{sim1v7:snr10:mise:basis21:tab},
 \ref{sim1v7:snr1:mise:basis22:tab} 
and \ref{sim1v7:snr10:mise:basis22:tab}
 in the supplementary material. 
In particular, the MISEs for 4 functional slope estimates for our proposed methods are smaller than those for OLS methods. As expected, when the true basis is Fourier, fitting using Fourier basis results in better estimation accuracy (smaller MISEs) compared with using Chebyshev basis.

In addition, we also perform simulation studies where the true basis is the Chebyshev basis of second kind defined in previous paragraph, and we fit our method using both Chebyshev basis of second kind and Fourier basis. The results are included in Tables 
\ref{sim1v7:snr1:mise:basis11:tab} 
to \ref{sim1v7:snr10:mise:basis12:tab} in the supplementary material. The findings are similar.

\section{Real data applications}\label{realdata}
We apply our method to the cortical surface motor task related fMRI data from Human Connectome Project (HCP) Dataset (https://www.humanconnectome.org/). We use the $900$ Subjects release that includes behavioral and 3T MR imaging data from $970$ healthy adult participants collected in 2012-spring 2015. We focus on the $845$ subjects having the cortical surface motor task related fMRI data. This task was adapted from the one developed by Buckner and colleagues \cite{buckner2011organization,yeo2011organization}. 

In the motor task, participants are presented with visual cues that ask them to either tap their left or right fingers, or squeeze their left or right toes, or move their tongue to map motor areas. Each block of a movement type lasted 12 seconds (10 movements), and is preceded by a 3 second cue. In each of the two runs, there are 13 blocks, with 2 of tongue movements, 4 of hand movements (2 right and 2 left), and 4 of foot movements (2 right and 2 left). In addition, there are 3 15-second fixation blocks per run. This task contains the following events, each of which is computed against the fixation baseline. For each subject, number of frames per run of the motor task is $284$, with run duration of $3.57$ minutes \cite{wu20171200}. {For each subject, two motor task-related fMRI scans are available: one run was acquired with right-to-left phase encoding, and a second run with left-to-right phase encoding. In this paper, we use the left-to-right phase encoding scan for each subject.}

{We use the ``Desikan-Killiany'' atlas \cite{Desikan2006} to divide the brain into $68$ regions of interest (ROIs). For each subject $i$, we average the blood oxygenation level dependent (BOLD) time series of all pixels in each ROI, which results in a functional curve $ y_{il}(t) $ for $ 1\leq i\leq n $ and $1\leq l\leq p$.} For each curve $ y_{il}(t) $, we do not observe their full trajectory, but instead realization of the curve on $284$ equally space time points: $ t_k=2.16*(k-1)/283$ minutes ($1\leq k\leq 284$). 
We consider $ s=4$ motor instrument covariates measured using tests adapted from the American Thoracic Society's 6-minute walk test \cite{enright2003six}, 
the 9-hole pegboard test \cite{wang2015dexterity}, and 
the American Society of Hand Therapy's grip strength test \cite{macdermid1994interrater}. In the test adapted from American Thoracic Society's 6-minute walk test, the sub-maximal cardiovascular endurance is measured by recording the distance that the participant is able to walk on a 50-foot course in 2 minutes and the time that the participant is able to walk a 4-meter distance at their usual pace. 
In the 9-hole pegboard test, the manual dexterity is measured by the time required for the participant to accurately place and remove 9 plastic pegs into a plastic pegboard.
In the test adapted from American Society of Hand Therapy's grip strength test, participants are seated in a chair with their feet touching the ground. With the elbow bent to 90 degrees and the arm against the trunk, wrist at neutral, participants squeeze the Jamar Plus Digital dynamometer as hard as they can for a count of three. The dynamometer records a digital reading of force in pounds.
The 4 covariates we consider are ``Endurance-AgeAdj",    ``GaitSpeed-Comp",   ``Dexterity-AgeAdj",  and  ``Strength-AgeAdj", where ``GaitSpeed-Comp" is the distance walked in 2 minute, and ``Endurance-AgeAdj", ``Dexterity-AgeAdj'' and ``Strength-AgeAdj" are sub-maximal cardiovascular endurance, manual dexterity and grip strength respectively, adjusted by the participant's age.

To implement our method, we first standardize the functional responses $ y_{il}(t) $'s and centre the covariates $ x_{ij}$'s. 
We apply our method by fitting the model using Fourier basis and select the optimal regularization parameter and truncation number by five-fold cross validation. {We have selected $9$ Fourier basis functions, and the rank of $\widehat{\mathsf{M}}$ is $4$. }


{We have also obtained the OLS estimate $\widehat{\mathsf{M}}_{\mathrm{OLS}}$, i.e. setting $ \lambda=0 $ in equation \eqref{eq_originalmodel}. We plotted the first 10 singular values of the $\widehat{\mathsf{M}}$ (red solid) and $\widehat{\mathsf{M}}_{\mathrm{OLS}}$ (black dashed) in Figure \ref{dataPlotsScree}. Inspecting the figure reveals that the first $4$ singular value of $\widehat{\mathsf{M}}_{\mathrm{OLS}}$ dominate the remaining ones, which verifies the low-rank assumption in this paper. }

Previous literature \cite{martino2011intrasurgical} suggested the left and right superior frontal regions are strongly associated with motor function. Therefore, 
We plot the estimated coefficient functions $\{\hat{\beta}_j(t), 1\leq j\leq 4\}$ corresponding to the left superior frontal regions in Figure \ref{dataPlots1}, and the estimated coefficient functions $\{\hat{\beta}_j(t), 1\leq j\leq 4\}$ corresponding to the right superior frontal region in Figure \ref{dataPlots2}.  From the figures, we can see the estimated coefficient functions $\hat{\beta}_j(t)$ for right and left superior frontal regions have similar patterns for each $ 1\leq j\le 4$. This is explained by the symmetry of the brain.

{To summarize the result of the performance of all $68$ regions, we plot the standardized $\hat{\beta}_1(t)$ for all $68$ ROIs in Figure \ref{dataPlots3} (a). Here the standardized $\hat{\beta}_1(t)$ is defined as $\hat{\beta}_{j,\mathrm{stand}}(t) = \{\hat{\beta}_j(t) - \int_0^1 \hat{\beta}_j(s) ds\} / [\int_0^1 \{\hat{\beta}_j(u) - \int_0^1 \hat{\beta}_j(s)ds\}^2 du]^{1/2}$ for $1 \leq j \leq 4$.  Similar plots for standardized versions of $\hat{\beta}_2(t)$, $\hat{\beta}_3(t)$ and $\hat{\beta}_4(t)$ are included in Figure \ref{dataPlots3} (b)-(d), respectively. }

{We have also tested the non-stationary assumption of the error processes in real data application. In particular, we apply the Kwiatkowski-Phillips-Schmidt-Shin (KPSS) tests \cite{kwiatkowski1992testing} on the fitted residual time series for $845$ individuals, which yield $ 845\times 68= 57,640$ error processes. We find that $96.7\%$ of them are not (trend) stationary with significance level $0.05$. This indicates that most of the error processes in the application are not stationary.}

}

\section{Discussion}\label{discussion}
In this paper, we propose a multivariate functional responses low rank regression model with possible high dimensional functional responses and scalar covariates. To estimate the nonparametric coefficient functions, our method employs the state-of-art sieve regression. By imposing a low-rank structure of the coefficient matrix, our proposal can obtain a global fit of the coefficient estimates. We have shown that our method performs well in both simulation and the HCP fMRI data application. 

There are a number of important directions for future work. First, we assume the covariates affect the responses linearly with only main effects. Further investigation is warranted to extend the proposed approach to the case with interaction effects and/or nonlinear effects. Second, it is an interesting topic to further develop inference procedure for our approach, which can characterize the uncertainty of estimates. One may consider using either bootstrap or debiased approaches to construct simultaneous confidence bands for the coefficient curves.

\section*{Acknowledgement}
The authors would like to thank the editor, associate editor, and the referees for their constructive comments, which have substantially improved the paper.

\bibliographystyle{abbrv}
\bibliographystyle{plain}
\bibliography{refsg}

\begin{thebibliography}{10}

\bibitem{bai2002determining}
J.~Bai and S.~Ng.
\newblock Determining the number of factors in approximate factor models.
\newblock {\em Econometrica}, 70(1):191--221, 2002.

\bibitem{beck2009fast}
A.~Beck and M.~Teboulle.
\newblock A fast iterative shrinkage-thresholding algorithm for linear inverse
  problems.
\newblock {\em SIAM Journal on Imaging Sciences}, 2(1):183--202, 2009.

\bibitem{buckner2011organization}
R.~L. Buckner, F.~M. Krienen, A.~Castellanos, J.~C. Diaz, and B.~T. Yeo.
\newblock The organization of the human cerebellum estimated by intrinsic
  functional connectivity.
\newblock {\em Journal of Neurophysiology}, 2011.

\bibitem{chen2013reduced}
K.~Chen, H.~Dong, and K.-S. Chan.
\newblock Reduced rank regression via adaptive nuclear norm penalization.
\newblock {\em Biometrika}, 100(4):901--920, 2013.

\bibitem{chensieve}
X.~Chen.
\newblock Chapter 76 {L}arge sample sieve estimation of semi-nonparametric
  models.
\newblock volume~6 of {\em Handbook of Econometrics}, pages 5549 -- 5632.
  Elsevier, 2007.

\bibitem{Desikan2006}
R.~S. Desikan, F.~S{\'e}gonne, B.~Fischl, B.~T. Quinn, B.~C. Dickerson,
  D.~Blacker, R.~L. Buckner, A.~M. Dale, R.~P. Maguire, B.~T. Hyman, et~al.
\newblock An automated labeling system for subdividing the human cerebral
  cortex on {MRI} scans into gyral based regions of interest.
\newblock {\em NeuroImage}, 31(3):968--980, 2006.

\bibitem{ding2019}
X.~Ding and Z.~Zhou.
\newblock Estimation and inference for precision matrices of nonstationary time
  series.
\newblock {\em Annals of Statistics}, 48(4):2455--2477, 2020.

\bibitem{enright2003six}
P.~L. Enright.
\newblock The six-minute walk test.
\newblock {\em Respiratory Care}, 48(8):783--785, 2003.

\bibitem{fan1996local}
J.~Fan and I.~Gijbels.
\newblock {\em Local polynomial modelling and its applications}, volume~66 of
  {\em Monographs on Statistics and Applied Probability}.
\newblock Chapman \& Hall, London, 1996.

\bibitem{FLF}
Y.~Fang, K.~A. {Loparo}, and X.~Feng.
\newblock Inequalities for the trace of matrix product.
\newblock {\em IEEE Transactions on Automatic Control}, 39(12):2489--2490,
  1994.

\bibitem{fara}
J.~J. Faraway.
\newblock Regression analysis for a functional response.
\newblock {\em Technometrics}, 39(3):254--261, 1997.

\bibitem{hu2020nonparametric}
W.~Hu, T.~Pan, D.~Kong, and W.~Shen.
\newblock Nonparametric matrix response regression with application to brain
  imaging data analysis.
\newblock {\em Biometrics}, page to appear, 2020.

\bibitem{hu2020matrix}
W.~Hu, W.~Shen, H.~Zhou, and D.~Kong.
\newblock Matrix linear discriminant analysis.
\newblock {\em Technometrics}, 62(2):196--205, 2020.

\bibitem{jones2012non}
D.~T. Jones, P.~Vemuri, M.~C. Murphy, J.~L. Gunter, M.~L. Senjem, M.~M.
  Machulda, S.~A. Przybelski, B.~E. Gregg, K.~Kantarci, D.~S. Knopman, et~al.
\newblock Non-stationarity in the “resting brain’s” modular architecture.
\newblock {\em PloS One}, 7(6):e39731, 2012.

\bibitem{Kong2019}
D.~Kong, B.~An, J.~Zhang, and H.~Zhu.
\newblock {L2RM}: Low-rank linear regression models for high-dimensional matrix
  responses.
\newblock {\em Journal of the American Statistical Association},
  115(529):403--424, 2020.

\bibitem{kwiatkowski1992testing}
D.~Kwiatkowski, P.~C. Phillips, P.~Schmidt, Y.~Shin, et~al.
\newblock Testing the null hypothesis of stationarity against the alternative
  of a unit root.
\newblock {\em Journal of Econometrics}, 54(1-3):159--178, 1992.

\bibitem{macdermid1994interrater}
J.~C. MacDermid, J.~F. Kramer, M.~G. Woodbury, R.~M. McFarlane, and J.~H. Roth.
\newblock Interrater reliability of pinch and grip strength measurements in
  patients with cumulative trauma disorders.
\newblock {\em Journal of Hand Therapy}, 7(1):10--14, 1994.

\bibitem{martino2011intrasurgical}
J.~Martino, A.~Gabarr{\'o}s, J.~Deus, M.~Juncadella, J.~Acebes, A.~Torres, and
  J.~Pujol.
\newblock Intrasurgical mapping of complex motor function in the superior
  frontal gyrus.
\newblock {\em Neuroscience}, 179:131--142, 2011.

\bibitem{NIPS}
S.~Negahban, B.~Yu, M.~J. Wainwright, and P.~K. Ravikumar.
\newblock A unified framework for high-dimensional analysis of {$M$}-estimators
  with decomposable regularizers.
\newblock In {\em Advances in Neural Information Processing Systems 22}, pages
  1348--1356. 2009.

\bibitem{negahban2012unified}
S.~N. Negahban, P.~Ravikumar, M.~J. Wainwright, B.~Yu, et~al.
\newblock A unified framework for high-dimensional analysis of $ m $-estimators
  with decomposable regularizers.
\newblock {\em Statistical Science}, 27(4):538--557, 2012.

\bibitem{nesterov2013introductory}
Y.~Nesterov.
\newblock {\em Introductory lectures on convex optimization: A basic course},
  volume~87.
\newblock Springer Science \& Business Media, 2013.

\bibitem{matrixbook}
K.~B. Petersen and M.~S. Pedersen.
\newblock {\em The matrix cookbook}.
\newblock Technical report, Technical University of Denmark, 2007., 2007.

\bibitem{Qi1984}
L.~Qi.
\newblock Some simple estimates for singular values of a matrix.
\newblock {\em Linear Algebra and its Applications}, 56:105 -- 119, 1984.

\bibitem{Ramsayet1991}
J.~O. Ramsay and C.~J. Dalzell.
\newblock Some tools for functional data analysis (with discussion).
\newblock {\em Journal of the Royal Statistical Society: Series B (Statistical
  Methodology)}, 53:539--572, 1991.

\bibitem{RamsaySilverman2005}
J.~O. Ramsay and B.~W. Silverman.
\newblock {\em Functional Data Analysis (2nd ed)}.
\newblock Springer, New York, 2005.

\bibitem{HT}
H.~Tasaki.
\newblock Convergence rates of approximate sums of riemann integrals.
\newblock {\em Journal of Approximation Theory}, 161(2):477 -- 490, 2009.

\bibitem{RMTNON}
R.~Vershynin.
\newblock Introduction to the non-asymptotic analysis of random matrices.
\newblock {\em arXiv preprint arXiv 1011.3027}, 2011.

\bibitem{wang2016functional}
J.-L. Wang, J.-M. Chiou, and H.-G. M{\"u}ller.
\newblock Functional data analysis.
\newblock {\em Annual Review of Statistics and Its Application}, 3:257--295,
  2016.

\bibitem{wang2015dexterity}
Y.-C. Wang, R.~W. Bohannon, J.~Kapellusch, A.~Garg, and R.~C. Gershon.
\newblock Dexterity as measured with the 9-hole peg test (9-hpt) across the age
  span.
\newblock {\em Journal of Hand Therapy}, 28(1):53--60, 2015.

\bibitem{wu20171200}
H.~WU-Minn.
\newblock 1200 subjects data release reference manual.
\newblock {\em URL https://www. humanconnectome.org}, 2017.

\bibitem{Yaoet2005b}
F.~Yao, H.~G. M\"{u}ller, and J.~L. Wang.
\newblock Functional linear regression analysis for longitudinal data.
\newblock {\em Annals of Statistics}, 33:2873--2903, 2005.

\bibitem{yeo2011organization}
B.~T. Yeo, F.~M. Krienen, J.~Sepulcre, M.~R. Sabuncu, D.~Lashkari,
  M.~Hollinshead, J.~L. Roffman, J.~W. Smoller, L.~Z{\"o}llei, J.~R. Polimeni,
  et~al.
\newblock The organization of the human cerebral cortex estimated by intrinsic
  functional connectivity.
\newblock {\em Journal of neurophysiology}, 2011.

\bibitem{yu2020beyond}
D.~Yu, L.~Wang, D.~Kong, and H.~Zhu.
\newblock Beyond scalar treatment: A causal analysis of hippocampal atrophy on
  behavioral deficits in {A}lzheimer's studies.
\newblock {\em arXiv preprint arXiv:2007.04558}, 2020.

\bibitem{yuan2007dimension}
M.~Yuan, A.~Ekici, Z.~Lu, and R.~Monteiro.
\newblock Dimension reduction and coefficient estimation in multivariate linear
  regression.
\newblock {\em Journal of the Royal Statistical Society: Series B (Statistical
  Methodology)}, 69(3):329--346, 2007.

\bibitem{zhou2014regularized}
H.~Zhou and L.~Li.
\newblock Regularized matrix regression.
\newblock {\em Journal of the Royal Statistical Society: Series B (Statistical
  Methodology)}, 76(2):463--483, 2014.

\bibitem{WZ1}
Z.~Zhou and W.~Wu.
\newblock Local linear quantile estimation for non-stationary time series.
\newblock {\em Annals of Statistics}, 37:2696--2729, 2009.

\bibitem{WZ2}
Z.~Zhou and W.~Wu.
\newblock Simultaneous inference of linear models with time varying
  coefficents.
\newblock {\em Journal of the Royal Statistical Society: Series B (Statistical
  Methodology)}, 72:513--531, 2010.

\bibitem{zhu2014bayesian}
H.~Zhu, Z.~Khondker, Z.~Lu, and J.~G. Ibrahim.
\newblock Bayesian generalized low rank regression models for neuroimaging
  phenotypes and genetic markers.
\newblock {\em Journal of the American Statistical Association},
  109(507):977--990, 2014.

\bibitem{zhu2011fadtts}
H.~Zhu, L.~Kong, R.~Li, M.~Styner, G.~Gerig, W.~Lin, and J.~H. Gilmore.
\newblock Fadtts: functional analysis of diffusion tensor tract statistics.
\newblock {\em NeuroImage}, 56(3):1412--1425, 2011.

\bibitem{zhu2012multivariate}
H.~Zhu, R.~Li, and L.~Kong.
\newblock Multivariate varying coefficient model for functional responses.
\newblock {\em Annals of Statistics}, 40(5):2634--2666, 2012.

\end{thebibliography}

\begin{center}
{\Large Supplementary material for \\
Multivariate functional responses low rank regression with an application to brain imaging data}
\end{center}
This supplementary material contains technical proofs, numerical simulation and real data application results.

\begin{appendix}

\renewcommand{\theequation}{S.\arabic{equation}}
\renewcommand{\thetable}{S.\arabic{table}}
\renewcommand{\thefigure}{S.\arabic{figure}}
\renewcommand{\thesection}{S.\arabic{section}}
\renewcommand{\thelemma}{S.\arabic{lemma}}
\renewcommand{\thetheorem}{S.\arabic{theorem}}

\section{Some preliminaries}
We write (\ref{eq_betaexpan}) as
\begin{equation}\label{eq_mentryall}
\beta_{jl}(t)=\sum_{h=1}^c M_{jl,h} b_h(t)+\sum_{h=c+1}^{\infty} M_{jl,h} b_h(t),
\end{equation}
where $\sum_{h=c+1}^{\infty} M_{jl,h}b_h(t)$ corresponds to the error $O(c^{-d})$ in (\ref{eq_betaexpan}). Similar to the definition of $\mathsf{M},$ we denote $\mathsf{M}^{\dagger}$ as the collection of the entries corresponding to the second term of the right-hand side of (\ref{eq_mentryall}) and
$\overline{\mathsf{M}}$ as the matrix containing all the entries $M_{jl,h}, j=1,2,\cdots,s,\ h=1,2,\cdots,$ in (\ref{eq_mentryall}). Strictly speaking, $\overline{\mathsf{M}}$ is not a matrix, but for notation convenience, we denote it as a matrix of dimension $ p\times s\infty$. Therefore, we have that $\overline{\mathsf{M}}=(\mathsf{M}, \mathsf{M}^{\dagger}).$ 

In light of (\ref{eq_model}) with (\ref{eq_mentryall}), for the sequence of observation pairs $(\mathsf{X}_i, \bm{Y}_i), i=1,2,\cdots, n,$  we denote the loss function for $\overline{\mathsf{M}}$ as 
\begin{equation}\label{eq_defineloss1}
\mathcal{L}(\overline{\mathsf{M}}; (\mathsf{X}_i, \bm{Y}_i), i=1,2,\cdots,n)=\frac{1}{nT}\sum_{i=1}^n\| \bm{Y}_i-\overline{\mathsf{M}} \mathsf{X}_i  \|_F^2,
\end{equation}
where $\overline{\mathsf{M}} \mathsf{X}_i=\mathsf{M} \bm{X}_i+\mathsf{M}^{\dagger} \bm{X}_i^{\dagger}$. Recall $\bm{X}_i$ is defined in (\ref{eq_xgenerate}),  $\bm{X}_i^\dagger$ is defined in a similar fashion by using the basis $\{b_j(\cdot)\}_{j>c}$ and $\mathsf{X}_i$ is defined accordingly. Ideally, $\mathsf{M}^{\dagger} \bm{X}_i^{\dagger}$ corresponds to the $o(1)$ part in (\ref{eq_finalmodel1}).
Therefore, the true value $\overline{\mathsf{M}}^*$ is defined as 
\begin{equation*}
\overline{\mathsf{M}}^*=\underset{\overline{\mathsf{M}}}{\text{arg min}} \ \mathbb{E} \mathcal{L}(\overline{\mathsf{M}}; (\mathsf{X}_i, \bm{Y}_i), i=1,2,\cdots,n).
\end{equation*}

As discussed in Section \ref{sec:modelassumption},  we want to estimate $\mathsf{M}$ and hence we can treat $\mathsf{M}^\dagger$ as nuisance parameters. In this sense, the true value of $\mathsf{M}$ is defined as  
\begin{equation}\label{eq_mstar1}
\mathsf{M}^*=\underset{\mathsf{M}}{\text{arg min}} \ \mathbb{E} \mathcal{L}(\mathsf{M}; (\mathsf{X}_i, \bm{Y}_i), i=1,2,\cdots,n| (\mathsf{M}^{\dagger})^*),
\end{equation}
where
\begin{equation*}
\mathcal{L}(\mathsf{M}; (\mathsf{X}_i, \bm{Y}_i), i=1,2,\cdots,n| (\mathsf{M}^{\dagger})^*)=\frac{1}{nT} \sum_{i=1}^n \| \bm{Y}_i- \mathsf{M} \bm{X}_i-(\mathsf{M}^\dagger)^* \bm{X}_i^\dagger \|_F^2.
\end{equation*}
To ease the notation, we introduce a diagonal block matrix $\bm{M} \in \mathbb{R}^{np \times nsc}$ with $n$ blocks and each diagonal block is $\mathsf{M}$ and $\mathcal{X} \in \mathbb{R}^{ns \infty \times T}, \mathcal{Y} \in \mathbb{R}^{np \times T}$ contain the sequences of $\mathsf{X}_i$ and $\bm{Y}_i, i=1,2,\cdots,n$ respectively. Similarly, we can define $\bm{M}^\dagger$ and $\overline{\bm{M}}.$ As a consequence, we can rewrite  (\ref{eq_mstar1}) as 
\begin{equation}\label{eq_mstar}
\mathsf{M}^*=\underset{\mathsf{M}}{\text{arg min}} \ \mathbb{E} \mathcal{L}(\mathsf{M}; \mathcal{X}, \mathcal{Y}| (\mathsf{M}^{\dagger})^*),
\end{equation}
where
\begin{equation*}
\mathcal{L}(\mathsf{M}; \mathcal{X}, \mathcal{Y}| (\mathsf{M}^{\dagger})^*)=\frac{1}{nT}\|\mathcal{Y}-\bm{M} \mathcal{X}_c-\bm{M}^\dagger \mathcal{X}_c^\dagger \|_F^2,
\end{equation*}
with $\mathcal{X}_c \in \mathbb{R}^{nsc \times T}$ contains the matrices $\bm{X}_i, i=1,2,\cdots, n$ and $\mathcal{X}_c$ is defined similarly.

Next, we provide an estimate of $\mathsf{M}^*$ denoted by $\widehat{\mathsf{M}}.$ For any given regularizer $\mathcal{R}$ and regularization penalty $\lambda_n$, let
\begin{equation}\label{eq_mhat}
\widehat{\mathsf{M}}=\underset{\mathsf{M}}{\text{arg min}} \left[\mathcal{L}_1(\mathsf{M};\mathcal{X}, \mathcal{Y})+\lambda_n \mathcal{R}(\mathsf{M}) \right],
\end{equation}  
where $\mathcal{L}_1$ is an approximate loss function for $\mathcal{L}$ and defined as 
\begin{equation}\label{eq_definel1}
\mathcal{L}_1(\mathsf{M}; \mathcal{X}, \mathcal{Y})=\frac{1}{nT}\|\mathcal{Y}-\bm{M} \mathcal{X}_c \|_F^2,
\end{equation}
$\mathcal{R}$ is the nuclear norm for the rectangular matrix.
Our goal is to derive a bound for $\| \mathsf{M}-\widehat{\mathsf{M}} \|_F.$ We state such results in the following subsection. 

We start by decomposing the loss function $\mathcal{L}(\mathsf{M};\mathcal{X}, \mathcal{Y}|(\mathsf{M}^\dagger)^*).$ Note that 
\begin{align}\label{eq_decompose}
\mathcal{L}(\mathsf{M};\mathcal{X}, \mathcal{Y}|(\mathsf{M}^\dagger)^*)=\mathcal{L}_1(\mathsf{M}; \mathcal{X}, \mathcal{Y})&-2\frac{1}{nT}\operatorname{tr} \left((\mathcal{Y}-\bm{M} \mathcal{X}_c)[(\bm{M}^\dagger)^* \mathcal{X}_c^{\dagger}]^{\T} \right) \nonumber \\
&+\frac{1}{nT}\| (\bm{M}^{\dagger})^* \mathcal{X}_c^\dagger \|_F^2.
\end{align}
First of all, since $(\bm{M}^\dagger)^*$ is a nuisance parameter , the third term on the right-hand side of (\ref{eq_decompose}) can be regarded as a constant term with respect to $\mathsf{M}$. Hence it suffices to minimize $\mathcal{L}_1(\mathsf{M}; \mathcal{X}, \mathcal{Y})-\frac{2}{nT}\operatorname{tr} \left((\mathcal{Y}-\bm{M} \mathcal{X}_c)[(\bm{M}^\dagger)^* \mathcal{X}_c^\dagger]^{\T} \right).$ Denote
\begin{equation}\label{eq_defnmathcale}
\mathcal{E}(\mathsf{M})=\frac{2}{nT}\operatorname{tr} \left((\mathcal{Y}-\bm{M} \mathcal{X}_c)\left[(\bm{M}^\dagger)^*\mathcal{X}_c^\dagger \right]^{\T} \right).
\end{equation}
In order to state our results and make it clear how our work differs from \cite{negahban2012unified}, we follow the notation of \cite{negahban2012unified} and let $\theta \equiv \mathsf{M}.$  Based on the above discussion, in view of the definition of $\widehat{\bm{M}}$ in (\ref{eq_mhat}), it suffices to consider the following optimization problem
\begin{equation}\label{eq_finalmodel}
\widehat{\theta}=\underset{\theta}{\text{arg min}} \left[\mathcal{L}(\theta;\mathcal{X}, \mathcal{Y}|(\mathsf{M}^\dagger)^*)+\lambda_n \mathcal{R}(\theta)+\mathcal{E}(\theta)  \right].
\end{equation}


We introduce some notation and assumptions, which are also used in \cite{negahban2012unified}. Let $\mathcal{M}$ be the model subspace to capture the constraints; for instance,  the subspace of low-rank matrices under Assumption \ref{assum_lowrank} in our problem. Let $\overline{\mathcal{M}}$ be the completion of $\mathcal{M}$ and $\overline{\mathcal{M}}^\perp$ be the orthogonal complement of $\overline{\mathcal{M}}.$ It is remarkable that $\overline{\mathcal{M}}^{\perp}$ is referred to as the \emph{perturbation subspace}, representing deviations away from the model subspace $\mathcal{M}.$     

We need the following definition, taken from Definition 1 in \cite{negahban2012unified}. 

\begin{definition}[Decomposability of $\mathcal{R}$] Given a pair of subspaces $\mathcal{M} \subseteq \overline{\mathcal{M}},$ a norm-based regularizer $\mathcal{R}$ is decomposable with respect to $(\mathcal{M},\overline{\mathcal{M}}^\perp)$ if 
\begin{equation*}
\mathcal{R}(\theta+\gamma)=\mathcal{R}(\theta)+\mathcal{R}(\gamma),
\end{equation*}
for all $\theta \in \mathcal{M}$ and $\gamma \in \overline{\mathcal{M}}^\perp$.  
\end{definition}


It has been shown in Example 3 of \cite{negahban2012unified} that the nuclear norm is decomposable with respect to appropriately chosen subspaces (see equations (13a) and (13b) of \cite{negahban2012unified}).

We then introduce the \emph{restricted strong convexity} (RSC) condition, which is taken from Definition 2 of \cite{negahban2012unified}. Denote the error of Taylor series of $\mathcal{L}$ at $\theta^*$ as 
\begin{equation*}
\delta \mathcal{L}(\Delta, \theta^*):=\mathcal{L}(\theta^*+\Delta)-\mathcal{L}(\theta^*)-\langle \nabla \mathcal{L}(\theta^*), \Delta \rangle.
\end{equation*}
\begin{definition}[Restricted Strong Convexity]\label{defn_rsc} The loss function satisfies a RSC condition with curvature $\kappa_{\mathcal{L}}>0$ and tolerance function $\tau_{\mathcal{L}}$ if 
\begin{equation*}
\delta \mathcal{L}(\Delta, \theta^*) \geq \kappa_{\mathcal{L}} \| \Delta \|^2-\tau^2_{\mathcal{L}}(\theta^*),
\end{equation*}
for all $\Delta \in \mathbb{C}$ defined in (\ref{set_one}) or (\ref{set_two}). 
\end{definition} 

Finally, we introduce the \emph{subspace compatibility constant} to control $\mathcal{R}(\cdot)$ (Definition 3 in \cite{negahban2012unified}). 
\begin{definition} [Subspace compatibility constant]\label{defn_subspacecompati} For any subspace $\mathcal{M}$ of $\mathbb{R}^p,$ the subspace compatibility constant with respect to the pair $(\mathcal{R}, \| \cdot \|)$ is given by 
\begin{equation*}
\Psi(\mathcal{M}):=\sup_{\bm{u} \in \mathcal{M} \backslash \{ \bm{0}\}} \frac{\mathcal{R}(\bm{u})}{\| \bm{u}\|}.
\end{equation*}
\end{definition}


We are ready to state our main results. The following result deals with general $M$-estimator of the form  (\ref{eq_finalmodel}). We define the projection operator as follow
\begin{equation*}
\Pi_{\mathcal{M}}(\bm{u}):=\underset{\bm{v} \in \mathcal{M}}{\text{arg min}} \ \| \bm{u}-\bm{v} \|,
\end{equation*}
with the projection $\Pi_{\mathcal{M}^\perp}$ defined in an analogous way.  For simplicity, we use the following shorthand notation $\bm{u}_{\mathcal{M}}=\Pi_{\mathcal{M}}(\bm{u})$ and $\bm{u}_{\mathcal{M}^\perp}=\Pi_{\mathcal{M}^\perp}(\bm{u}).$
\begin{theorem}
\label{thm_main} Suppose that the loss function $\mathcal{L}$ is convex and differentiable, and satisfies the RSC condition in Definition \ref{defn_rsc} with curvature $\kappa_{\mathcal{L}}$ and tolerance $\tau_{\mathcal{L}}.$ We also assume that the regularizer $\mathcal{R}$ is a norm and is decomposable with respect to the subspace pair $(\mathcal{M}, \overline{\mathcal{M}}^\perp),$ where $\mathcal{M} \subset \overline{\mathcal{M}}.$ Denote 
\begin{eqnarray}\label{eq_defn_bound}
&&\varsigma (\kappa,\tau, \upsilon) \equiv \varsigma(\lambda_n, \Psi, \kappa, \tau, \theta^*, \upsilon) \nonumber\\
&:=& 9 \frac{\lambda_n^2}{\kappa^2} \Psi^2(\overline{\mathcal{M}})+\frac{\lambda_n}{\kappa} \{2 \tau^2(\theta^*)+4 \mathcal{R}(\theta^*_{\mathcal{M}^\perp})+2\upsilon \}. 
\end{eqnarray}
 Suppose that there exists some linear function $\mathcal{D}$ in $\Delta$ and independent of $\theta^*$ such that 
\begin{equation*}
\mathcal{D}(\Delta)=\mathcal{E}(\theta^*+\Delta)-\mathcal{E}(\theta^*).
\end{equation*}
Furthermore, let $\epsilon>0$ such that
\begin{equation}\label{eq_epsilon}
\sup_{\Delta} | \mathcal{D}(\Delta) | \leq \epsilon.
\end{equation}
If the strictly positive regularization constant satisfies $\lambda_n \geq 2 \mathcal{R}^* (\nabla \mathcal{L}(\theta^*)),$ when conditional on the observation $(\mathcal{X}, \mathcal{Y})$ and $n$ is large enough, we have that 
\begin{equation}\label{eq_thm_1}
\| \widehat{\theta}_{\lambda_n}-\theta^* \|^2 \leq C\varsigma(\kappa_{\mathcal{L}}, \tau_{\mathcal{L}},2\epsilon),
\end{equation}
where the error norm is the same as $\mathcal{L}(\cdot)$ and $C>0$ is some universal constant.
\end{theorem}

We remark that the error bound in \cite[Theorem 1]{negahban2012unified} is $\varsigma(\kappa_{\mathcal{L}}, \tau_{\mathcal{L}}, 0)$ since they do not have the error term $\mathcal{E}.$ Moreover, in our setup for the sieve regression, $\mathcal{E}(\cdot)$ satisfies (i) of Theorem \ref{thm_main} since 
\begin{eqnarray}\label{eq_ourlinear}
&&\operatorname{tr} \left((\mathcal{Y}-(\theta^*+\Delta) \mathcal{X}_c)[(\theta^\dagger)^* \mathcal{X}_c^\dagger]^{\T} \right)-\operatorname{tr} \left((\mathcal{Y}-\theta^* \mathcal{X}_c)((\theta^\dagger)^* \mathcal{X}_c^\dagger)^{\T} \right) \nonumber\\
&=&\operatorname{tr}\left( \Delta \mathcal{X}_c [(\theta^\dagger)^* \mathcal{X}_c^\dagger]^{\T} \right).
\end{eqnarray}

\section{Technical proofs}
To make it convenient for the readers, we use Table \ref{table_symbols} to list the matrices and their associated dimensions.

\begin{table}[H]
\centering
\begin{tabular}{@{}l*2{>{}l}%
  l<{}l@{}}
\hline
  & \multicolumn{1}{c}{\head{Matrix}}
  & \multicolumn{1}{c}{\head{Dimension}}\\
  %
  \hline
  \multirow{8}{*}
	&  $\bm{Y}_i, \bm{E}_i$ &  $p \times T$ & \\
    &  $\bm{X}_i$ &  $sc \times T$ & \\
 	& $\mathsf{M}$ & $p \times sc$ & \\
	& $\bm{X}_i^\dagger$  & $s\infty \times T$  & \\
	& $\mathsf{M}^\dagger$  & $p \times s\infty$  & \\
	& $\mathcal{X}_c$ & $nsc \times T$ &   \\
	& $\bm{M}$ & $np \times nsc$  & \\
	& $\mathcal{X}_c^\dagger$ & $ns\infty \times T$ &   \\
	& $\bm{M}^\dagger$ & $np \times ns\infty$  & \\
	& $\mathcal{Y}, \mathbf{E}$ & $np \times T$& \\ 
\hline
  
\end{tabular}
\caption{Matrices and their dimensions. Here we use the short-hands that $sc=s \times c.$ For instance, $\bm{X}_i$ contains $T$ matrices of dimension $s \times c$ as a stack.}\label{table_symbols}
\end{table}

We use the following abbreviations for our proof
\begin{equation*}
\mathcal{L}(\cdot)=\mathcal{L}(\theta;\mathcal{X}, \mathcal{Y}|(\theta^\dagger)^*).
\end{equation*} 
We will make use of the function $\mathcal{F}$ given by 
\begin{equation}\label{eq_defnmathcalf}
\mathcal{F}(\Delta):=\mathcal{L}(\theta^*+\Delta)-\mathcal{L}(\theta^*)+\lambda_n(\mathcal{R}(\theta^*+\Delta)-\mathcal{R}(\theta^*))+\mathcal{E}(\theta^*+\Delta)-\mathcal{E}(\theta^*).
\end{equation}
We denote the optimal error by 
\begin{equation*}
\widehat{\Delta}=\widehat{\theta}-\theta^*.
\end{equation*}
We notice that $\mathcal{F}(0)=0$ and $\mathcal{F}(\widehat{\Delta}) \leq 0.$  

\subsection{Some auxiliary lemmas}\label{sec: auxilaritylemma}
In this section, we provide some auxiliary lemmas which will be used in the proof of Theorem \ref{thm_main} and Theorem \ref{cor_main}.  We first provide some preliminary results.

\begin{lemma}[Deviation inequalities]\label{lem_deviation}  For any decomposable regularizer and $\theta^*$ and $\Delta,$ we have 
\begin{equation*}
\mathcal{R}(\theta^*+\Delta)-\mathcal{R}(\theta^*) \geq \mathcal{R}(\Delta_{\overline{\mathcal{M}}^\perp})-\mathcal{R}(\Delta_{\overline{\mathcal{M}}})-2 \mathcal{R}(\theta^*_{\mathcal{M}^\perp}).
\end{equation*}
Moreover, as long as $\lambda_n \geq 2 \mathcal{R}^*( \nabla \mathcal{L}(\theta^*))$ and $\mathcal{L}$ is convex, we have 
\begin{equation*}
\mathcal{L}(\theta^*+\Delta)-\mathcal{L}(\theta^*) \geq \frac{\lambda_n}{2} [\mathcal{R}(\Delta_{\overline{\mathcal{M}}})+\mathcal{R}(\Delta_{\overline{\mathcal{M}}^\perp})].
\end{equation*}
\end{lemma}
\begin{proof}
See Lemma 3 of \cite{negahban2012unified}.
\end{proof}

\begin{lemma}\label{lem_mainone} Suppose $\mathcal{L}$ is a convex and differentiable function, and consider any optimal solution $\widehat{\theta}$ to the optimization problem (\ref{eq_finalmodel}) with a strictly positive regularization parameter satisfying 
\begin{equation*}
\lambda_n \geq 2 \mathcal{R}^*( \nabla \mathcal{L}(\theta^*)).
\end{equation*}
Then for any pair $(\mathcal{M}, \overline{\mathcal{M}}^\perp)$ over which $\mathcal{R}$ is decomposable, the error $\widehat{\Delta}=\widehat{\theta}_{\lambda_n}-\theta^*$ belongs to the set 
\begin{align}\label{set_one}
\mathbb{C}(\mathcal{M}, \overline{\mathcal{M}}^\perp;\theta^*)
:=\left\{ \Delta\vert \mathcal{R}(\Delta_{\overline{\mathcal{M}}^\perp} ) \leq \frac{2}{\lambda_n} \left(\mathcal{E}(\theta^*)-\mathcal{E}(\theta^*+\Delta) \right) \right.\nonumber\\
\left. +3\mathcal{R}(\Delta_{\overline{\mathcal{M}}})+4 \mathcal{R}(\theta^*_{\mathcal{M}^\perp})  \right\}.
\end{align}
Moreover, if $\mathcal{E}(\cdot)$ is some convex and differentiable norm on the metric space of the parameter, we have that if 
\begin{equation*}
\lambda_n \geq 2 \mathcal{R}^* (\nabla \mathcal{L}(\theta^*)+\nabla \mathcal{E}(\theta^*)),
\end{equation*}
then for any pair $(\mathcal{M}, \overline{\mathcal{M}}^\perp)$ over which $\mathcal{R}$ is decomposable, the error $\widehat{\Delta}=\widehat{\theta}_{\lambda_n}-\theta^*$ belongs to the set 
\begin{align}\label{set_two}
\mathbb{C}(\mathcal{M}, \overline{\mathcal{M}}^\perp;\theta^*):=\left\{ \Delta\vert \mathcal{R}(\Delta_{\overline{\mathcal{M}}^\perp} ) \leq 3\mathcal{R}(\Delta_{\overline{\mathcal{M}}})+4 \mathcal{R}(\theta^*_{\mathcal{M}^\perp})  \right\}.
\end{align}
\end{lemma}
\begin{proof}
By the expansion (\ref{eq_defnmathcalf}), the fact $\mathcal{F}(\widehat{\Delta}) \leq 0$ and Lemma \ref{lem_deviation}, we readily obtain that 
\begin{align*}
0 \geq \mathcal{F}(\widehat{\Delta}) & \geq \lambda_n \left\{ \mathcal{R}(\Delta_{\overline{\mathcal{M}}^\perp})-\mathcal{R}(\Delta_{\overline{\mathcal{M}}})-2 \mathcal{R}(\theta^*_{\mathcal{M}^\perp}) \right\}-\frac{\lambda_n}{2}\left[ \mathcal{R}(\Delta_{\overline{\mathcal{M}}})+\mathcal{R}(\Delta_{\overline{\mathcal{M}}^\perp})\right] \\
& +\mathcal{E}(\theta^*+\widehat{\Delta})-\mathcal{E}(\theta^*) \\
& =\frac{\lambda_n}{2}\left\{ \mathcal{R}(\Delta_{\overline{\mathcal{M}}^\perp})-3\mathcal{R}(\Delta_{\overline{\mathcal{M}}})-4 \mathcal{R}(\theta^*_{\mathcal{M}^\perp})  \right\}+\mathcal{E}(\theta^*+\widehat{\Delta})-\mathcal{E}(\theta^*) .
\end{align*}
This concludes the proof of (\ref{set_one}). For the proof of (\ref{set_two}), we can apply Lemma \ref{lem_deviation} to the convex and differentiable function  $\mathcal{L}_1:=\mathcal{L}+\mathcal{E}.$ Since $\nabla$ is a linear operator, when $\lambda_n \geq 2 \mathcal{R}^* (\nabla \mathcal{L}(\theta^*)+\nabla \mathcal{E}(\theta^*)),$ we have 
\begin{align*}
0 \geq \mathcal{F}(\widehat{\Delta}) & \geq \lambda_n \left\{ \mathcal{R}(\Delta_{\overline{\mathcal{M}}^\perp})-\mathcal{R}(\Delta_{\overline{\mathcal{M}}})-2 \mathcal{R}(\theta^*_{\mathcal{M}^\perp}) \right\}-\frac{\lambda_n}{2}\left[ \mathcal{R}(\Delta_{\overline{\mathcal{M}}})+\mathcal{R}(\Delta_{\overline{\mathcal{M}}^\perp})\right] \\
& =\frac{\lambda_n}{2}\left\{ \mathcal{R}(\Delta_{\overline{\mathcal{M}}^\perp})-3\mathcal{R}(\Delta_{\overline{\mathcal{M}}})-4 \mathcal{R}(\theta^*_{\mathcal{M}^\perp})  \right\},
\end{align*}
from which the proof of (\ref{set_two}) follows.

\end{proof}

\begin{remark}
The counterpart of the above lemma \cite[Lemma 1]{negahban2012unified} does not have the approximation error term $\mathcal{E}(\cdot)$.  In our setup, $\mathcal{E}(\cdot)$ is not a properly defined norm. Hence, we need to apply (\ref{set_one}) whenever it is needed. For an interpretation of the Lemma \ref{lem_mainone}, we refer to Figure 1 of \cite{negahban2012unified}. 
\end{remark}

Recall the sets $\mathbb{C}$ defined in (\ref{set_one}) or (\ref{set_two}). Let $\delta$ be a given error radius and denote $\mathbb{K}(\delta):=\mathbb{C} \cap \{ \| \Delta \|=\delta\}.$ We have the following lemma, the counterpart of which is \cite[Lemma 4]{negahban2012unified}. 
\begin{lemma}\label{lem_sign}
Suppose $\mathcal{R}(\cdot)$ is decomposable and convex, and $\mathcal{L}$ is differentiable and convex. Then \\ 
(i). When $\mathbb{C}$ is defined in (\ref{set_one}), suppose that there exists some linear function $\mathcal{D}$ in $\Delta$ and independent of $\theta^*$ such that 
\begin{equation*}
\mathcal{D}(\Delta)=\mathcal{E}(\theta^*+\Delta)-\mathcal{E}(\theta^*).
\end{equation*}
If $\mathcal{F}(\Delta)>0$ for all vectors $\Delta \in \mathbb{K}(\delta),$ then $\| \widehat{\Delta} \| \leq \delta.$ \\
(ii). When $\mathbb{C}$ is defined in (\ref{set_two}), i.e. $\mathcal{E}(\cdot)$ is some convex differentiable norm,  if $\mathcal{F}(\Delta)>0$ for all vectors $\Delta \in \mathbb{K}(\delta),$ then $\| \widehat{\Delta} \| \leq \delta.$
\end{lemma}
\begin{proof}
We start with the proof of (i). We prove the contrapositive statement: in particular, if for some optimal solution $\widehat{\theta}$ such that $\| \widehat{\Delta} \|>\delta,$  there must be some vector $\widetilde{\Delta} \in \mathbb{K}(\delta)$ such that $\mathcal{F}(\widetilde{\Delta}) \leq 0. $ To achieve this goal, it suffices to prove the following claim:
\begin{claim}\label{claim_shape}
If $\widehat{\Delta} \in \mathbb{C},$ then the entire line $\{t \widehat{\Delta}| t \in (0,1)\}$ connecting $\widehat{\Delta}$ with all-zeros vector is contained with $\mathbb{C}.$ 
\end{claim}
We first show how we can construct such a $\widetilde{\Delta}$ using the above claim. If $\| \widehat{\Delta}\|>\delta,$ then the line joining $\widehat{\Delta}$ and $0$ must intersect the set $\mathbb{K}(\delta)$ at some intermediate point $t^* \widehat{\Delta},$ for some $t^* \in (0,1)$ (i.e. after some proper scaling). By Claim \ref{claim_shape},  we have that $t^* \widehat{\Delta} \in \mathbb{C}.$  Since $\mathcal{D}(\Delta)$ is linear and both $\mathcal{L}$ and $\mathcal{R}$ are convex, by Jensen's inequality, we have 
\begin{equation*}
\mathcal{F}(t^* \widehat{\Delta})=\mathcal{F}(t^* \Delta+(1-t^*)0) \leq t^* \mathcal{F}(\widehat{\Delta})+(1-t^*) \mathcal{F}(0)=t^* \mathcal{F}(\widehat{\Delta}),
\end{equation*} 
where in the last equality we use the fact that $\mathcal{F}(0)=0.$ Since $\widehat{\Delta}$ is optimal, we have that $\mathcal{F}(t^* \widehat{\Delta}) \leq 0.$ Hence, we can choose $\widetilde{\Delta}=t^* \widehat{\Delta}$ and  conclude the proof of (i). 

Finally, we prove Claim \ref{claim_shape}. First, when $\theta^* \in \mathcal{M},$ we have that $\mathcal{R}(\theta^*_{\mathcal{M}^\perp})=0$ and the proof follows immediate. Second, when $\theta^* \notin \mathcal{M}$, it is easy to see that for any $t \in (0,1),$ we have 
\begin{equation*}
\Pi_{\overline{\mathcal{M}}}(t \Delta)=\underset{\gamma \in \overline{\mathcal{M}}}{\text{arg min}} \| t\Delta-\gamma \|=t \ \underset{\gamma \in \overline{\mathcal{M}}}{\text{arg min}} \| \Delta-\frac{\gamma}{t} \|=t \ \Pi_{\overline{\mathcal{M}}}(\Delta), 
\end{equation*}    
where we use the fact $\gamma/t$ also belongs to the subspace $\overline{\mathcal{M}}.$ Similarly, we can show that 
\begin{equation*}
\Pi_{\overline{\mathcal{M}}^\perp}(t \Delta)=t \Pi_{\overline{\mathcal{M}}^\perp}(\Delta).
\end{equation*}
Hence, we have that for all $\Delta \in \mathbb{C}$, 
\begin{align*}
\mathcal{R}(\Pi_{\overline{\mathcal{M}}^\perp}(t \Delta))=\mathcal{R}(t \Pi_{\overline{\mathcal{M}}^\perp}(\Delta))=t \mathcal{R}(\Pi_{\overline{\mathcal{M}}^\perp}) \\
\leq t \left\{ 3 \mathcal{R}(\Pi_{\overline{\mathcal{M}}}(\Delta)) +4\mathcal{R}(\Pi_{\mathcal{M}^\perp}(\theta^*))+\mathcal{D}(\Delta)\right\},
\end{align*}
where we use the fact that $\mathcal{R}(\cdot)$ is a norm and  the definition of $\mathbb{C}$ in (\ref{set_one}). We observe that $3t \mathcal{R}(\Pi_{\overline{\mathcal{M}}}(\Delta))=3 \mathcal{R}(\Pi_{\overline{\mathcal{M}}}(t \Delta))$ and $4t \mathcal{R}(\Pi_{\mathcal{M}^\perp}(\theta^*)) \leq 4 \mathcal{R}(\Pi_{\mathcal{M}^\perp}(\theta^*)), t \in (0,1).$ Moreover, since $\mathcal{D}(\cdot)$ is linear in $\Delta,$ we have $\mathcal{D}(t \Delta)=t \mathcal{D}(\Delta).$ Putting all these together, we find that
\begin{align*}
\mathcal{R}(\Pi_{\overline{\mathcal{M}}^\perp}(t \Delta)) & \leq 3 \mathcal{R}(\Pi_{\overline{\mathcal{M}}}(t \Delta))+4t \mathcal{R}(\Pi_{\mathcal{M}^\perp}(\theta^*))+\mathcal{D}(t\Delta) \\
& \leq 3 \mathcal{R}(\Pi_{\overline{\mathcal{M}}}(t \Delta))+4 \mathcal{R}(\Pi_{\mathcal{M}^\perp}(\theta^*))+\mathcal{D}(t\Delta).
\end{align*} 
This concludes the proof of Claim \ref{claim_shape}. 

For (ii), we can apply \cite[Lemma 4]{negahban2012unified} for the loss function $\mathcal{L}_1:=\mathcal{L}+\mathcal{E}$ to finish the proof.  
\end{proof}

Next, we provide some matrix inequalities that will be used in the proof of Theorem \ref{cor_main} to bound $\mathcal{E}(\cdot).$ 
\begin{lemma}\label{lem_circlerectangular} Suppose $A \in \mathbb{R}^{m \times n}.$ Denote 
\begin{equation*}
r_i:=\sum_{1 \leq j \neq i \leq n} |a_{ij}|, \  c_i:=\sum_{ 1 \leq j \neq i \leq m}|a_{ji}|, 
\end{equation*}
and 
\begin{equation*}
s_i:=\max\{r_i, c_i\}, \ a_i:=|a_{ii}|,
\end{equation*}
for $i=1,2,\cdots, \min \{m,n\}.$ Moreover, for $m \neq n,$ we define 
\begin{equation*}
s:=
\begin{cases}
\underset{n+1 \leq i \leq m}{\max} \left\{ \sum_{j=1}^n |a_{ij}|\right\}, & m>n \\
\underset{m+1 \leq i \leq n}{\max} \left\{ \sum_{j=1}^n |a_{ji}|\right\}, & m<n.
\end{cases}
\end{equation*}
For $m \geq n,$ we have that for each singular value of $A$ lies in one of the real intervals defined as 
\begin{equation*}
B_i=[\max\{a_i-s_i,0\}, a_i+s_i], \ i=1,2,\cdots, n; \ B_{n+1}=[0,s]. 
\end{equation*}
If $m=n$ or if $m>n$ and $a_i \geq s_i+s, \ i=1, \cdots, n,$ then $B_{n+1}$ is not needed in the above statement. Similarly results hold when $m \leq n.$
\end{lemma}
\begin{proof}
See Theorem 2 of \cite{Qi1984}.
\end{proof}

\begin{lemma}\label{lem_tracebound} Suppose $A$ and $B$ are positive definite square matrices. Then we have 
\begin{equation*}
\lambda_{\min}(A) \operatorname{tr}(B) \leq \operatorname{tr}(AB) \leq \lambda_{\max}(A) \operatorname{tr}(B) ,
\end{equation*}
where $\lambda_{\max}(A)$ is the largest eigenvalue of $A$ and $\lambda_{\min}(A)$ is the smallest eigenvalues of $A$. 
\end{lemma}
\begin{proof}
See equation (1) of \cite{FLF}.
\end{proof}

\begin{lemma}\label{lem_covconverge} Consider a sequence of i.i.d. subgaussian vectors $\bm{z}_1, \cdots, \bm{z}_n$ in $\mathbb{R}^s$ with covariance matrix $\Sigma_s.$ Let $\epsilon \in (0,1),  t \geq 1,$ then for some constant $C>0,$ with probability at least $1-2\exp(-t^2n)$ we have if $n \geq C (t/\epsilon)^2 s$
\begin{equation*}
\| \Sigma_s^n-\Sigma_s \|\leq \epsilon,
\end{equation*}
where $\Sigma_s^n$ is the sample covariance matrix of $(\bm{z}_i).$
\end{lemma}
\begin{proof}
See Corollary 5.50 of \cite{RMTNON}.
\end{proof}

\begin{lemma}\label{lem_riemansum}
Let $[a,b]$ be a bounded closed interval. We take an $n$-division $\Delta$ of $[a,b]$ as 
\begin{equation*}
\Delta: \ a=s_0 \leq s_1 \leq \cdots s_{n-1} \leq s_n=b,
\end{equation*}
and hence $s_i=a+i(b-a)/n.$ If $f$ is a twice differentiable and $f^{''}$is bound and almost everywhere continuous on $[a,b]$ then 
\begin{equation}
\lim_{n \rightarrow \infty} n^2 \left\{ \int_a^b f(x)dx-\sum_{i=1}^n(s_i-s_{i-1})f(\frac{s_{i-1}+s_i}{2}) \right\}=\frac{(b-a)^2}{24}(f'(b)-f'(a)).
\end{equation}

\end{lemma}

\begin{proof}
See Theorem 1.1 of \cite{HT}.
\end{proof}
\begin{lemma}\label{lem_circletheorem} 

Let  $A=(a_{ij})$ be a real $ n\times n$ matrix. For  $1 \leq i \leq n,$ let  $R_{i}=\sum _{{j\neq {i}}}\left|a_{{ij}}\right| $ be the sum of the absolute values of the non-diagonal entries in the  $i$-th row. Let  $ D(a_{ii},R_{i})\subseteq \mathbb {R} $ be a closed disc centered at $a_{ii}$ with radius  $R_{i}$. Such a disc is called a \emph{Gershgorin disc}. Every eigenvalue of $ A=(a_{ij})$ lies within at least one of the Gershgorin discs  $D(a_{ii},R_{i})$, where $R_i=\sum_{j\ne i}|a_{ij}|$.
\end{lemma}

\begin{proof}
See Lemma D.1 of \cite{ding2019}. 
\end{proof}

Finally, we provide a concentration inequality for the locally stationary time series $\{\epsilon_{ik}(t)\}, 1 \leq i \leq n, \ 1 \leq k \leq p.$

\begin{lemma}\label{lem_localconcen}
Let $x_i=G_i(\mathcal{F}_i),$ where $G_i(\cdot)$ is a measurable function and  $\mathcal{F}_i=(\ldots, \eta_{i-1}, \eta_i)$ and $\eta_i, \ i  \in \mathbb{Z}$ are i.i.d.  random variables. Suppose that $\mathbb{E}x_i=0$ and $\max_i \mathbb{E}|x_i|^q<\infty$ for some $q>1.$ For some $k>0,$ let $\delta_x(k):=\max_{ 1 \leq i \leq n} \left\|G_i(\mathcal{F}_i)-G_i(\mathcal{F}_{i,i-k}) \right\|_q.$ We further let $\delta_x(k)=0$ if $k<0.$ Write $\gamma_k=\sum_{i=0}^k \delta_x(i).$ Let $S_i=\sum_{j=1}^i x_j.$ \\
(i). For $q'=\min(2,q),$
\begin{equation*}
\left\| S_n \right\|_q^{q'} \leq C_q \sum_{i=-n}^{\infty} (\gamma_{i+n}-\gamma_i)^{q'}.
\end{equation*}
(ii). If $\Delta:=\sum_{j=0}^{\infty} \delta_x(j) <\infty,$ we then have 
\begin{equation*}
\left\| \max_{1 \leq i \leq n}|S_i| \right\|_q \leq C_q n^{1/q'} \Delta. 
\end{equation*}
In (i) and (ii), $C_q$ are generic finite constants which only depend on $q$ and can vary from place to place. 
\end{lemma}
\begin{proof}
See Lemma D.6 of \cite{ding2019}. 
\end{proof}

\subsection{Proof of Theorem \ref{thm_main}}

With the preparation in Section \ref{sec: auxilaritylemma}, especially Lemma \ref{lem_sign}, we now proceed to finish the proof of Theorem \ref{thm_main}.

\begin{proof}
The proof is essentially similar to that of \cite[Theorem 1]{negahban2012unified}, for the self-completeness, we also provide the complete proof. In light of Lemma \ref{lem_sign}, it suffices to establish a lower bound on $\mathcal{F}(\Delta)$ over $\mathbb{K}(\delta)$ for an appropriately chosen radius $\delta>0.$ Indeed, for an arbitrary $\Delta \in \mathbb{K}(\delta),$ using the definition of $\mathcal{F}$ in (\ref{eq_defnmathcalf}),  we have 
\begin{align*}
\mathcal{F}(\Delta) & \geq \langle \nabla \mathcal{L}(\theta^*), \Delta \rangle+\kappa_{\mathcal{L}} \| \Delta \|^2-\tau_{\mathcal{L}}^2(\theta^*)+\lambda_n\{\mathcal{R}(\theta^*+\Delta)-\mathcal{R}(\theta^*)\}+\mathcal{D}(\Delta) , \\
& \geq \langle \nabla \mathcal{L}(\theta^*), \Delta \rangle+\kappa_{\mathcal{L}} \| \Delta \|^2-\tau_{\mathcal{L}}^2(\theta^*) + \\
& \quad +\lambda_n\{\mathcal{R}(\Delta_{\overline{\mathcal{M}}^\perp})-\mathcal{R}(\Delta_{\overline{\mathcal{M}}})-2 \mathcal{R}(\theta^*_{\mathcal{M}^\perp})\}+ \mathcal{D}(\Delta),
\end{align*}  
where the first inequality follows from RSC and the second inequality follows from Lemma \ref{lem_deviation}. Moreover, by the Cauchy-Schwarz inequality and the definition of dual norm, we readily obtain that 
\begin{equation*}
\left|\langle \nabla \mathcal{L}(\theta^*), \Delta \rangle \right| \leq \mathcal{R}^*(\nabla \mathcal{L}(\theta^*)) \mathcal{R}(\Delta).
\end{equation*}
Since $\lambda_n \geq 2 \mathcal{R}^*(\nabla \mathcal{L}(\theta^*))$ by assumption, we find that $\left|\langle \nabla \mathcal{L}(\theta^*), \Delta \rangle \right| \leq \frac{\lambda_n}{2} \mathcal{R}(\Delta)$ and hence we have
\begin{eqnarray*}
\mathcal{F}(\Delta)  &\geq& 
\kappa_{\mathcal{L}} \| \Delta \|^2-\tau_{\mathcal{L}}^2(\theta^*)+\\
&& + \lambda_n\{\mathcal{R}(\Delta_{\overline{\mathcal{M}}^\perp})-\mathcal{R}(\Delta_{\overline{\mathcal{M}}})-2 \mathcal{R}(\theta^*_{\mathcal{M}^\perp})\}+\mathcal{D}(\Delta)-\frac{\lambda_n}{2} \mathcal{R}(\Delta).
\end{eqnarray*}
Together with $\mathcal{R}(\Delta) \leq \mathcal{R}(\Delta_{\overline{\mathcal{M}}^\perp})+\mathcal{R}(\Delta_{\overline{\mathcal{M}}}),$ we find that 
\begin{equation*}
\mathcal{F}(\Delta) \geq \kappa_{\mathcal{L}} \| \Delta \|^2-\tau_{\mathcal{L}}^2(\theta^*)-\frac{\lambda_n}{2} \{ 3 \mathcal{R}(\Delta_{\overline{\mathcal{M}}})+4 \mathcal{R}(\theta^*_{\mathcal{M}^\perp}) \}+\mathcal{D}(\Delta).
\end{equation*} 
Since $0 \in \overline{\mathcal{M}},$ it is easy to see (the equation below eq. (55) of \cite{negahban2012unified}) $\| \Delta_{\overline{\mathcal{M}}} \| \leq \| \Delta \|.$ Moreover, by Definition \ref{defn_subspacecompati}, we find that $ \mathcal{R}(\Delta_{\overline{\mathcal{M}}}) \leq \Psi(\overline{\mathcal{M}}) \| \Delta \|. $ Since $\sup_{\Delta}|\mathcal{D}(\Delta)| \leq \epsilon,$ this leads to 
\begin{equation*}
\mathcal{F}(\Delta) \geq \kappa_{\mathcal{L}} \| \Delta \|^2-\tau_{\mathcal{L}}^2(\theta^*)-\frac{\lambda_n}{2} \{ 3 \Psi(\overline{\mathcal{M}})\| \Delta \|+4 \mathcal{R}(\theta^*_{\mathcal{M}^\perp}) \}-\epsilon.
\end{equation*} 
Note that the right-hand side of the above inequality is a strictly defined quadratic form in $\| \Delta \|$ and hence will be positive for $\| \Delta \|$ large. The proof then follows from some elementary computation on the quadratic equation.

\end{proof}

\subsection{Proof of Theorem \ref{cor_main}}
We will employ Theorem \ref{thm_main} to prove Theorem \ref{cor_main}. The key ingredient is to provide a bound for $\mathcal{E}(\cdot)$ using (\ref{eq_betaexpan}).  We will need the following facts on the matrix differentiation. For details, it can be found in \cite{matrixbook}. For any two $m \times n $ rectangular matrices $A,B$ and any matrix function $f,$ we have 
\begin{equation}\label{eq_matrixfacts}
\nabla_{A} \operatorname{tr} (AB^{\T})=B, \ \nabla_{A^{\T}} f(A)=(\nabla_{A} f(A))^{\T}. 
\end{equation}
We prepare some computation on the derivatives using (\ref{eq_matrixfacts}). Recall (\ref{eq_definel1}), as 
\begin{equation*}
\mathcal{L}_1=\frac{1}{nT} \operatorname{tr} \left( \mathcal{Y} \mathcal{Y}^{\T}-\mathcal{Y} ( \bm{M} \mathcal{X}_c)^{\T}- \bm{M} \mathcal{X}_c \mathcal{Y}^{\T}+ \bm{M} \mathcal{X}_c \mathcal{X}_c^{\T} \bm{M}^{\T} \right),
\end{equation*}
by (\ref{eq_matrixfacts}), we readily obtain that 
\begin{equation}\label{eq_gradone}
\nabla_{\bm{M}} \mathcal{L}_1=\frac{1}{nT}\left(-\mathcal{Y} \mathcal{X}_c^{\T}- \mathcal{X}_c \mathcal{Y}^{\T}+\bm{M} \mathcal{X}_c \mathcal{X}_c^{\T}+\mathcal{X}_c \mathcal{X}_c^{\T} \bm{M}^{\T} \right), 
\end{equation}
and the Hessian matrix of $\mathcal{L}_1$ at $\bm{M}$ is 
\begin{equation}\label{eq_hessianone}
\mathbf{H}_{\bm{M}} \mathcal{L}_1=\frac{2}{nT} \mathcal{X}_c \mathcal{X}_c^{\T}.
\end{equation}
Recall (\ref{eq_defnmathcale}), we have that 
\begin{equation}\label{eq_gradtwo}
\nabla_{\bm{M}} \mathcal{E}=-\frac{2}{nT} (\bm{M}^\dagger)^*\mathcal{X}_c^\dagger   \mathcal{X}_c^{\T},
\end{equation} 
and the Hessian matrix of $\mathcal{E}$ at $\bm{M}$ is
\begin{equation}\label{eq_hessiantwo}
\mathbf{H}_{\bm{M}} \mathcal{E}=\bm{0}. 
\end{equation}

\begin{proof} In view of (\ref{eq_ourlinear}), we need to apply (i) of Theorem \ref{thm_main}. We prepare two important facts for our proof. First, since $\mathcal{R}$ is the nuclear norm, we have that (\cite[Section 2.3]{negahban2012unified})
\begin{equation}\label{eq_dualnorm}
\mathcal{R}^*(\bm{M})=\| \bm{M} \|,
\end{equation}
where $\| \bm{M} \|$ stands for the largest singular value of $\bm{M}.$ Second, from the proof of \cite[Corollary 5]{NIPS}, we know that under Assumption \ref{assum_lowrank},  $\Psi(\overline{\mathcal{M}})=2 \sqrt{r}$ (Recall Definition \ref{defn_subspacecompati}).   Armed with the above results, we now proceed to check the conditions of Theorem \ref{thm_main} and the computation of the inputs there. 

In what follows, we denote $\Delta=\mathsf{M}-\widehat{\mathsf{M}}$ and $\bm{\Delta}$ to be a diagonal matrix with $n$ blocks whose diagonals are $\Delta.$   \\

\noindent \underline{Checking of the decomposablity and differentiablity}: It is clear that $\mathcal{L}$ is differentiable with respect to $\bm{M}.$ $\mathcal{R}$ is the nuclear norm and it is decomposable with respect to suitable subspaces defined in \cite[equations (13a) and (13b)]{negahban2012unified}.     \\

\noindent \underline{Checking of the RSC condition}: By (\ref{eq_hessianone}) and (\ref{eq_hessiantwo}),  
we find that the first-order Taylor expansion from Definition \ref{defn_rsc} is exact such that 
\begin{equation*}
\delta \mathcal{L}=\frac{2}{nT}\| \mathcal{X}_c \bm{\Delta} \|_F^2.
\end{equation*}  
It suffices to provide a lower bound for $\frac{2}{nT}\| \mathcal{X}_c \bm{\Delta} \|_F^2.$ {Note that 
\begin{align*}
\delta \mathcal{L}=\frac{2}{nT} \sum_{i=1}^n \| \bm{X}_i \Delta \|_F^2 & \geq \frac{2}{nT} \| \sum_{i=1}^n \bm{X}_i \Delta \|_F^2 \\ &\geq 2 \| \Delta \|_F^2 \lambda_{\min} \left( (1/nT) \sum_{i=1}^n \bm{X}_i \bm{X}_i^\top \right). 
\end{align*}
where $X=(\bm{x}_1, \cdots, \bm{x}_n)$ and in the first inequality we use the property of matrix norm and the second inequality we use Lemma \ref{lem_tracebound}.   }

By Assumption \ref{assump_parameter} and Lemma \ref{lem_covconverge}, for some constant $C>0,$ with probability at least $1-2 \exp(-n),$ we have 
\begin{equation}\label{kroneker1}
\|\frac{1}{n}XX^{\T}-\Sigma_s \| \leq C \sqrt{\frac{s}{n}}.
\end{equation} 
Together with  Assumptions \ref{assump_parameter} and \ref{ass_pdc}, when $n$ is sufficiently large, with probability at least $1-2\exp(-n),$ we have
{ 
\begin{equation*}
\lambda_{\min}(\frac{1}{n}XX^{\T}) \geq \frac{\sqrt{2}}{2} \lambda_{\min} (\Sigma_s).
\end{equation*} 
}

Moreover, by Lemmas \ref{lem_riemansum} and \ref{lem_circletheorem}, we find that for some constant $C>0,$
\begin{equation}\label{kroneker2}
\| \frac{1}{T} \bm{B}\bm{B}^{\T}-I_c \| \leq  \frac{C}{T^2},
\end{equation}
where we use the smoothness of the basis functions and the facts that $\int_0^1 b_i(t) b_j(t)dt=\delta_{ij}$ and  the $ij$th entry of $\frac{1}{T} \bm{B} \bm{B}^\top$ is $\frac{1}{T} \sum_{k=1}^T b_{i}(t_k) b_j(t_k).$ 

Hence, when $T$ is large enough, we have
{
\begin{equation*}
\lambda_{\min}(\frac{1}{T} \bm{B} \bm{B}^{\T}) \geq \frac{\sqrt{2}}{2}.
\end{equation*}
}
%
%
This shows that with probability at least $1-\exp(-n),$ we have 
{
\begin{equation*}
\delta \mathcal{L} \geq c \lambda_{\min}(\Sigma_s) \| \Delta \|_F^2,
\end{equation*}
where we use the fact that the eigenvalues of $A \otimes B$ are the products of the eigenvalues of $A$ and $B.$ 
}
Hence, we can take $\kappa_{\mathcal{L}}=C \lambda_{\min} (\Sigma_s)>0.$ 

\noindent \underline{Computation of $\lambda_n$}: Let $\mathbf{E}$ be the matrix contains $\bm{E}_i, i=1,2,\cdots, n.$ By (\ref{eq_gradone}) and (\ref{eq_gradtwo}), we have that 
\begin{equation}\label{eq_1213}
\nabla_{\bm{M}} \mathcal{L}=-\frac{1}{nT} \left( \mathbf{E} \mathcal{X}_c^{\T}+\mathcal{X}_c \mathbf{E}^{\T}\right)-\frac{3}{nT}(\bm{M}^\dagger)^* \mathcal{X}_c^\dagger \mathcal{X}_c^{\T}-\frac{1}{nT} \mathcal{X}_c \left[ (\bm{M}^\dagger)^* \mathcal{X}_c^\dagger\right]^{\T},
\end{equation} 
where we use the fact that $\mathcal{Y}-\bm{M} \mathcal{X}_c=(\bm{M}^\dagger)^* \mathcal{X}_c^\dagger+\mathbf{E}.$ By (\ref{eq_dualnorm}), it suffices to provide an upper bound for the largest singular value of the right-hand side of (\ref{eq_1213}) using Lemmas \ref{lem_circlerectangular} and \ref{lem_riemansum}. First, by Lemmas \ref{lem_circlerectangular} and \ref{lem_localconcen}, under Assumption \ref{assum_varaible}, for some constant $C_q>0,$ with probability at least $1-C_q n^{-q \tau}$, for some constant $C>0,$
\begin{align*}
\frac{1}{nT}\| \mathbf{E} \mathcal{X}_c^{\T}\| \leq \frac{p \xi n^{\tau}}{\sqrt{T}}.
\end{align*} 
Similarly, we have 
\begin{align*}
\frac{1}{nT}\| \mathcal{X}_c \mathbf{E}^{\T}\| \leq \frac{p \xi n^{\tau}}{\sqrt{T}}.
\end{align*} 
Second, by (\ref{eq_betaexpan}), (\ref{eq_scalarmodel}) and  Assumption \ref{ass_pdc} that $\bm{x}_i, i=1,2,\cdots, n,$ are subgaussian, for some small constant $\tau>0,$ we have that with probability at least $1-2\exp(-n^{2\tau}),$ 
\begin{equation}\label{eq_c-d1}
\frac{1}{nT} \| (\bm{M}^\dagger)^* \mathcal{X}_c^\dagger \mathcal{X}_c^{\T} \| \leq  \frac{p \xi n^{\tau} c^{-d}}{\sqrt{T}},
\end{equation} 
where we use Lemma \ref{lem_circlerectangular}. Similarly, we have 
\begin{equation}\label{eq_c-d2}
\frac{1}{nT} \| \mathcal{X}_c \left[ (\bm{M}^\dagger)^* \mathcal{X}_c^\dagger\right]^{\T} \| \leq  \frac{p \xi n^{\tau} c^{-d}}{\sqrt{T}}.
\end{equation}
This implies that with probability at least $1-C_q n^{-q \tau}$, we can choose
\begin{equation*}
\lambda_n \geq C \frac{p \xi n^{\tau}}{\sqrt{T}}.
\end{equation*}

\noindent \underline{Computation of $\epsilon$ in (\ref{eq_epsilon})}: We simply show that $\epsilon$ can be chosen as a bounded constant value with high probability, which is sufficient for our proof. This is done by using Cauchy-Schwartz inequality.  Note that 
\begin{equation*}
\mathcal{D}(\Delta)=\frac{1}{nT}  \operatorname{tr}\left( \bm{\Delta} \mathcal{X}_c ((\theta^\dagger)^* \mathcal{X}_c^\dagger)^{\T} \right) \leq  \frac{1}{nT} \left( \| \bm{\Delta} \mathcal{X}_c \|_F^2+\| (\theta^\dagger \mathcal{X}_c^\dagger)^*  \|_F^2 \right),
\end{equation*}
where we use Cauchy-Schwartz inequality.  On one hand,  by (\ref{kroneker1}) and  Assumption \ref{ass_pdc} that $\Sigma_s$ is bounded, when $n$ is large enough, we have that with probability at least $1-2\exp(-n)$
\begin{equation}\label{eq_epsilon1}
\frac{1}{nT} \| \Delta \mathcal{X}_c \|_F^2 \leq C_1 \lambda_1(\Sigma_s) \| \Delta \|_F^2 \leq C_2,
\end{equation}
where in the first inequality we use Lemma \ref{lem_tracebound} and second inequality we use Assumption \ref{assum_lowrank},  and $C_1, C_2>0$ are some constants. On the other hand, by (\ref{eq_betaexpan}), (\ref{eq_scalarmodel}) and  Assumption \ref{ass_pdc} that $\bm{x}_i, i=1,2,\cdots, n,$ are subgaussian, for some small constant $\tau>0,$ we have that with probability at least $1-2\exp(-n^{2\tau}),$ 
\begin{equation}\label{eq_epsilon2}
\frac{1}{nT} \| (\theta^\dagger)^* \mathcal{X}_c^\dagger \|_F^2 \leq  \xi p s^2 n^{2\tau} c^{-d}=o(1),
\end{equation}
where in the first inequality we use the fact that for subgaussian random variable $x,$ we have $\mathbb{P}(|x| \geq n^{\tau})\leq 1- 2\exp(-n^{2\tau})$ and the definition of $\| \cdot \|_F^2,$ and for the second equality we use Assumption \ref{assump_parameter}. By (\ref{eq_epsilon1}) and (\ref{eq_epsilon2}), we have shown that with probability at least $1-2\exp(-n^{2\tau})$ 
\begin{equation*}
\mathcal{D}(\Delta) \leq 2C_2. 
\end{equation*} 

After checking the conditions of Theorem \ref{thm_main} and the computation of $\lambda_n$ and $\epsilon,$ together with the fact that $\mathcal{R}(\theta^*_{\mathcal{M}^\perp})=0,$ we complete the proof. 
 
\end{proof}

\subsection{Proof of Corollary \ref{coro_coro}}
In this section, we prove Corollary \ref{coro_coro}.

\begin{proof}
Note that $\widehat{\bm{\beta}}_j(t)-\bm{\beta}_j(t) =(\widehat{\bm{M}}_j-\bm{M}_j)b(t). $ Since 
\begin{align*}
\| \widehat{\bm{\beta}}_j(t)-\bm{\beta}_j(t) \|^2 &=\operatorname{tr}(( \widehat{\bm{\beta}}_j(t)-\bm{\beta}_j(t))( \widehat{\bm{\beta}}_j(t)-\bm{\beta}_j(t))^\top )  \\
&=\operatorname{tr}((\widehat{\bm{M}}_j-\bm{M}_j)(\widehat{\bm{M}}_j-\bm{M}_j)^\top(\bm{b}(t) \bm{b}^\top(t))),
\end{align*}
by Lemma \ref{lem_tracebound} and Theorem \ref{cor_main}, with probability at least $1-C_q n^{-q\tau},$ we have
\begin{equation*}
\| \widehat{\bm{\beta}}_j(t)-\bm{\beta}_j(t) \|^2 \leq C  \left( \frac{\sqrt{p \xi } n^{\tau/2}}{ \sqrt{c \operatorname{tr}(\Sigma_s)} T^{1/4}}+ \frac{\sqrt{r} p \xi n^{\tau}}{c \operatorname{tr}(\Sigma_s) \sqrt{T}} \right)^2 \|\bm{b}(t) \bm{b}^\top(t) \|.
\end{equation*}
Moreover, for any $t \in [0,1],$ we have 
\begin{equation}
\|\bm{b}(t) \bm{b}^\top(t) \|=\sum_{h=1}^c b_h^2(t) \leq c.  
\end{equation}
We can therefore conclude our proof. 
\end{proof}

\section{Numerical simulation and real data application results}

\begin{figure}[H] \centering \subfigure[]{\includegraphics[scale = 0.15]{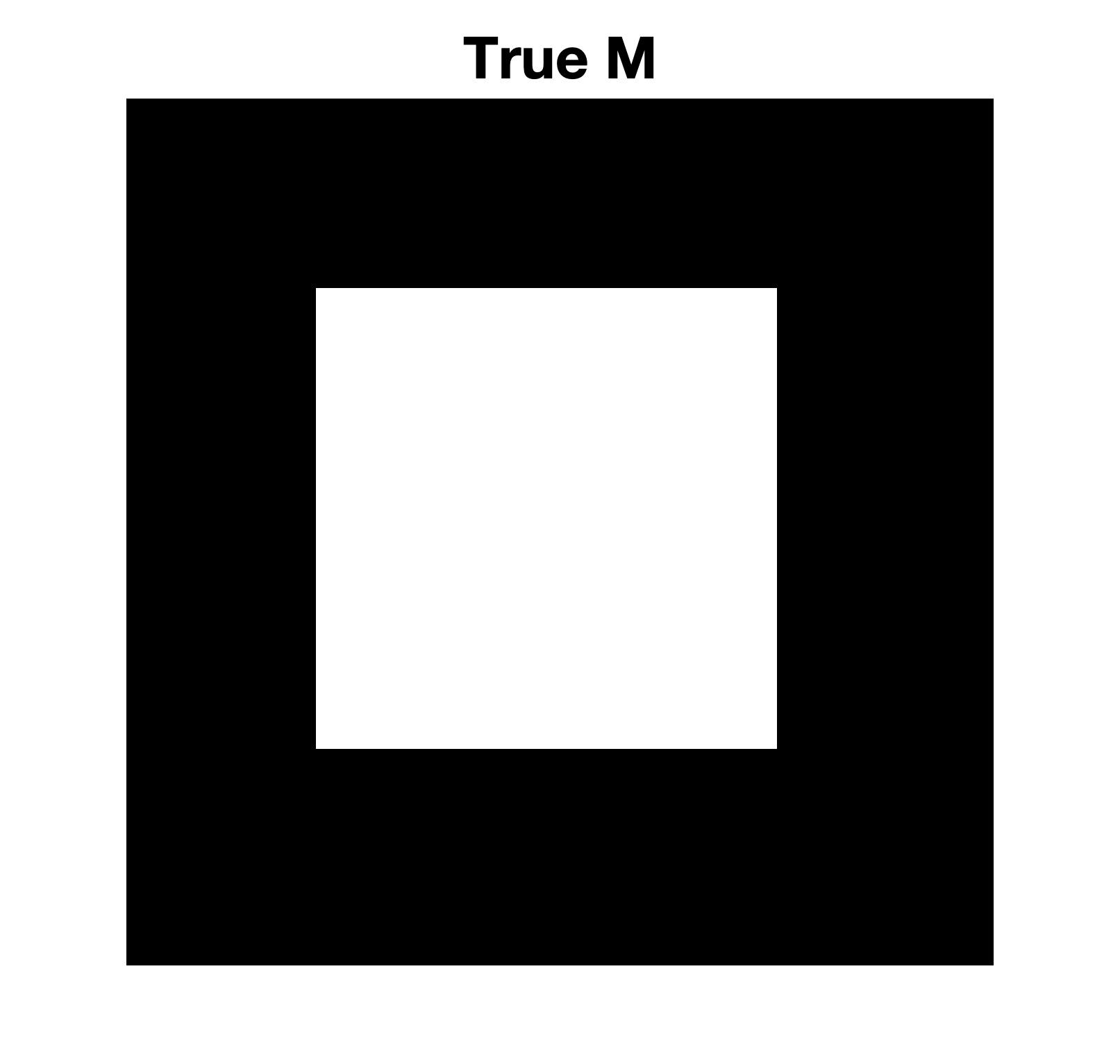}} \subfigure[]{\includegraphics[scale = 0.15]{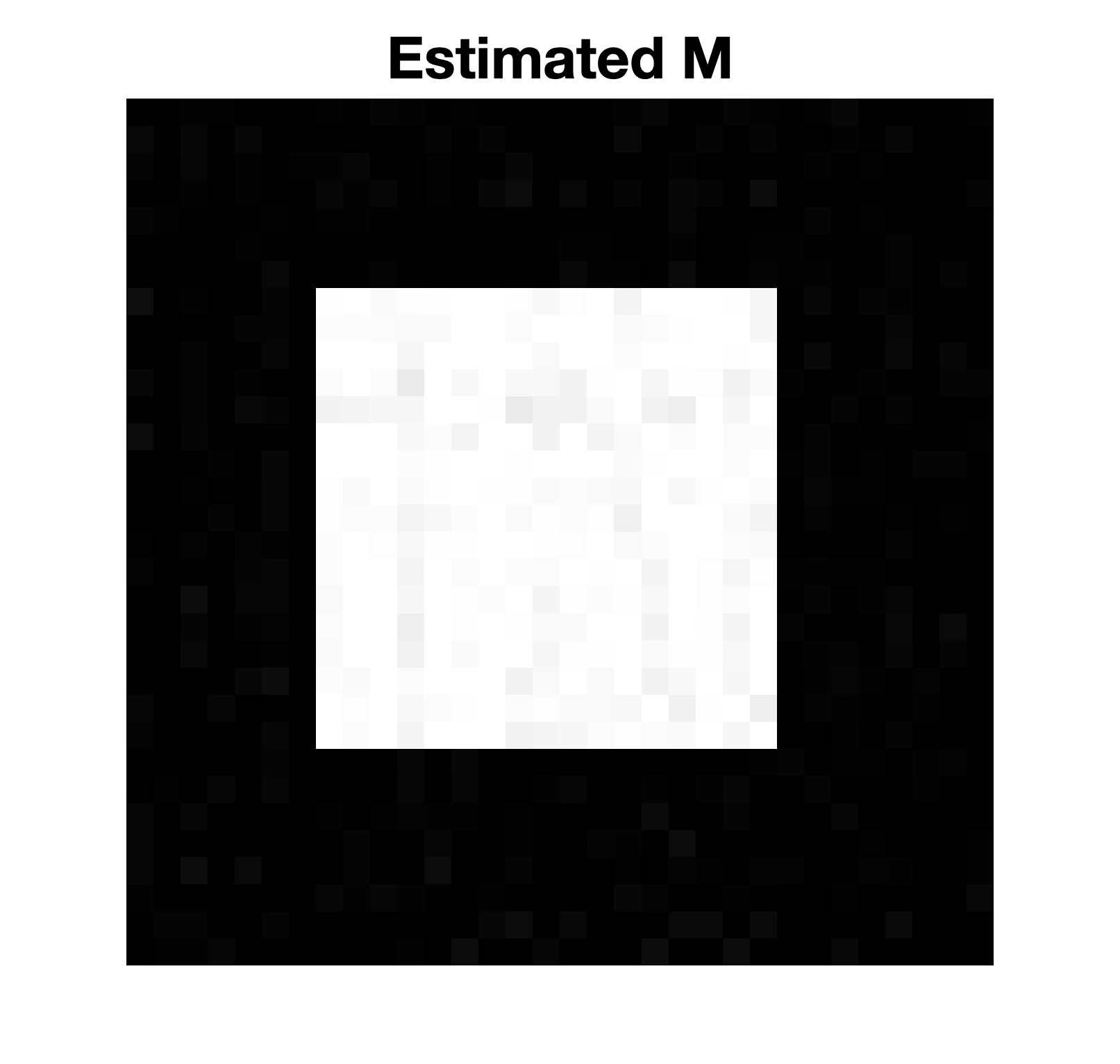}} \subfigure[]{\includegraphics[scale = 0.15]{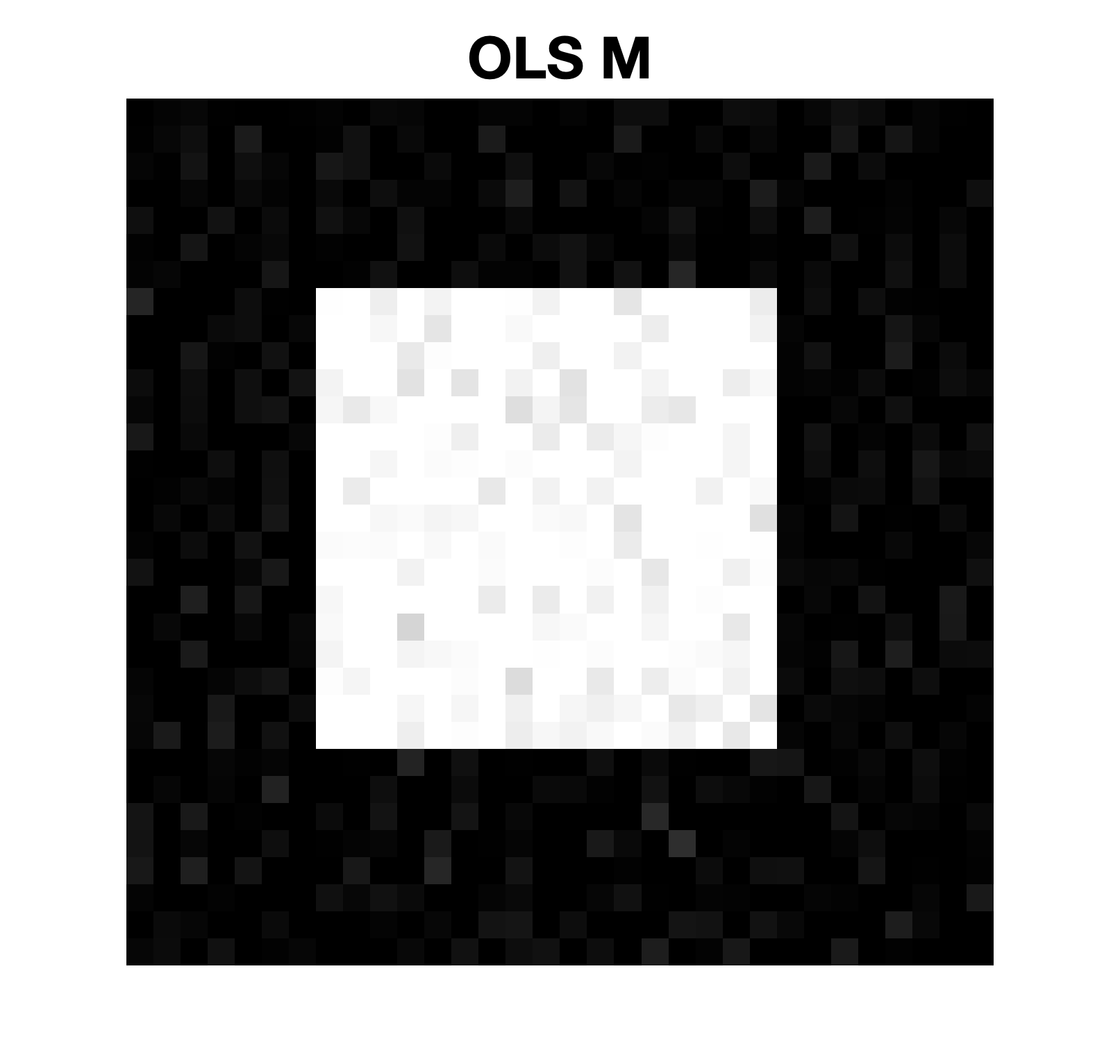}} 
\subfigure[]{\includegraphics[scale = 0.15]{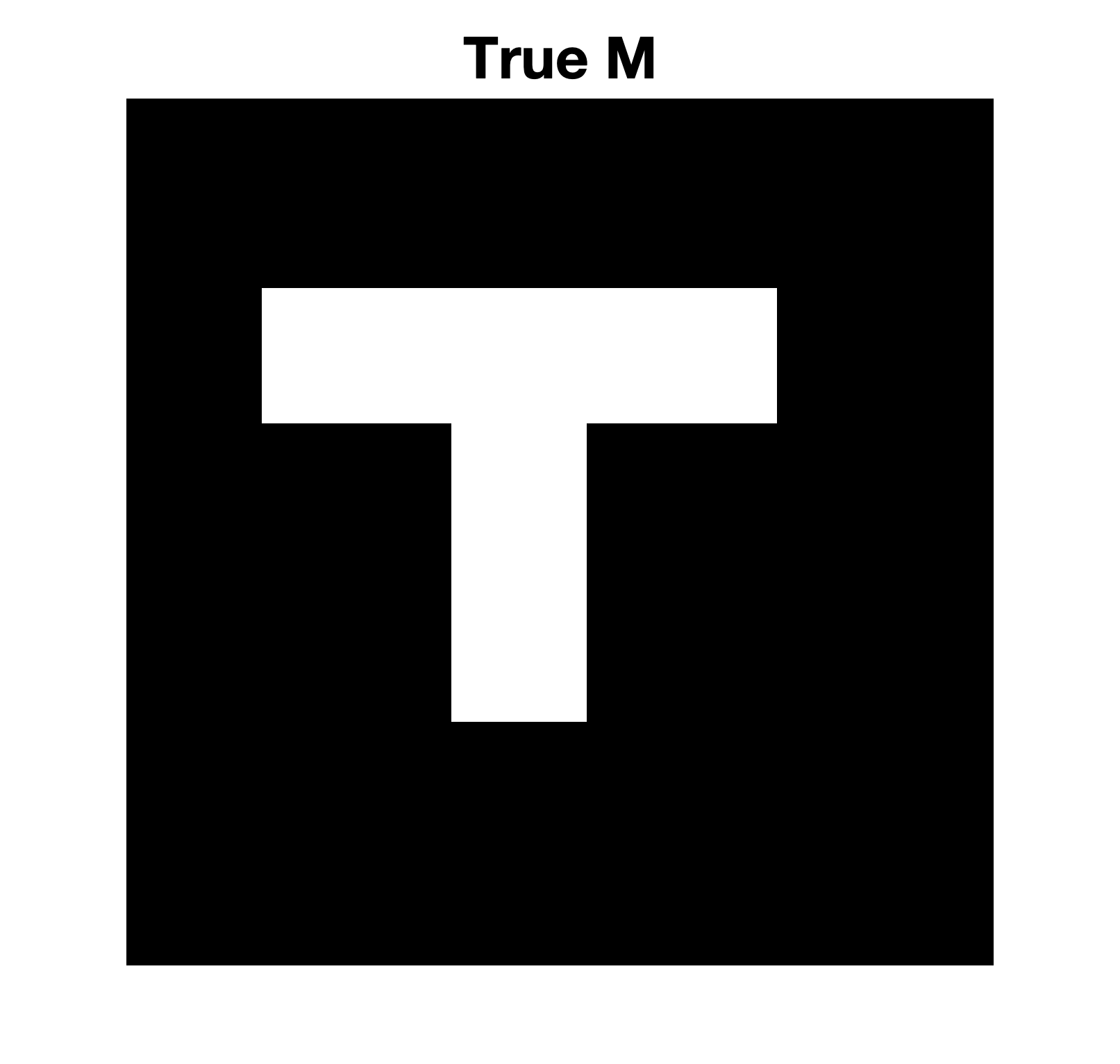}} \subfigure[]{\includegraphics[scale = 0.15]{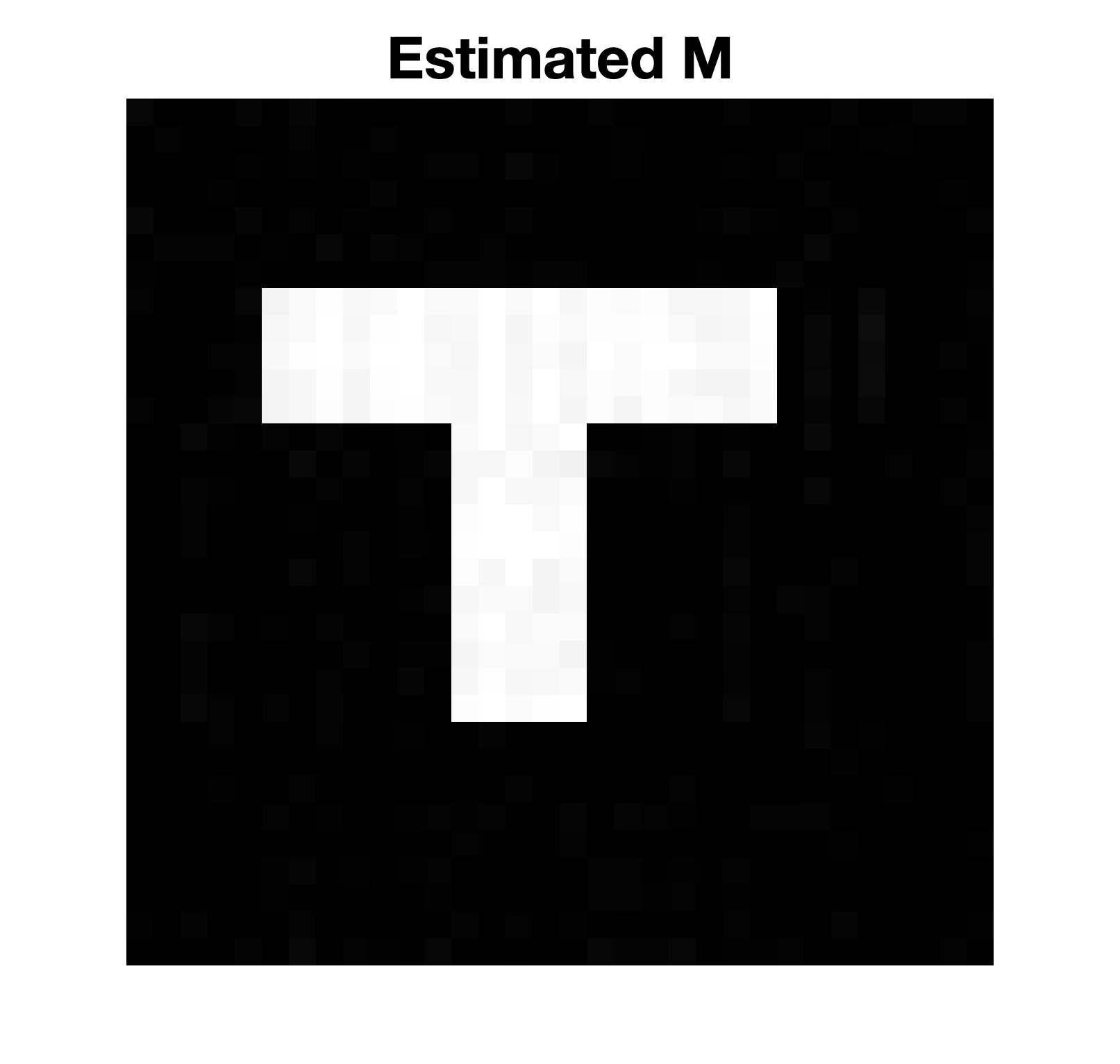}} \subfigure[]{\includegraphics[scale = 0.15]{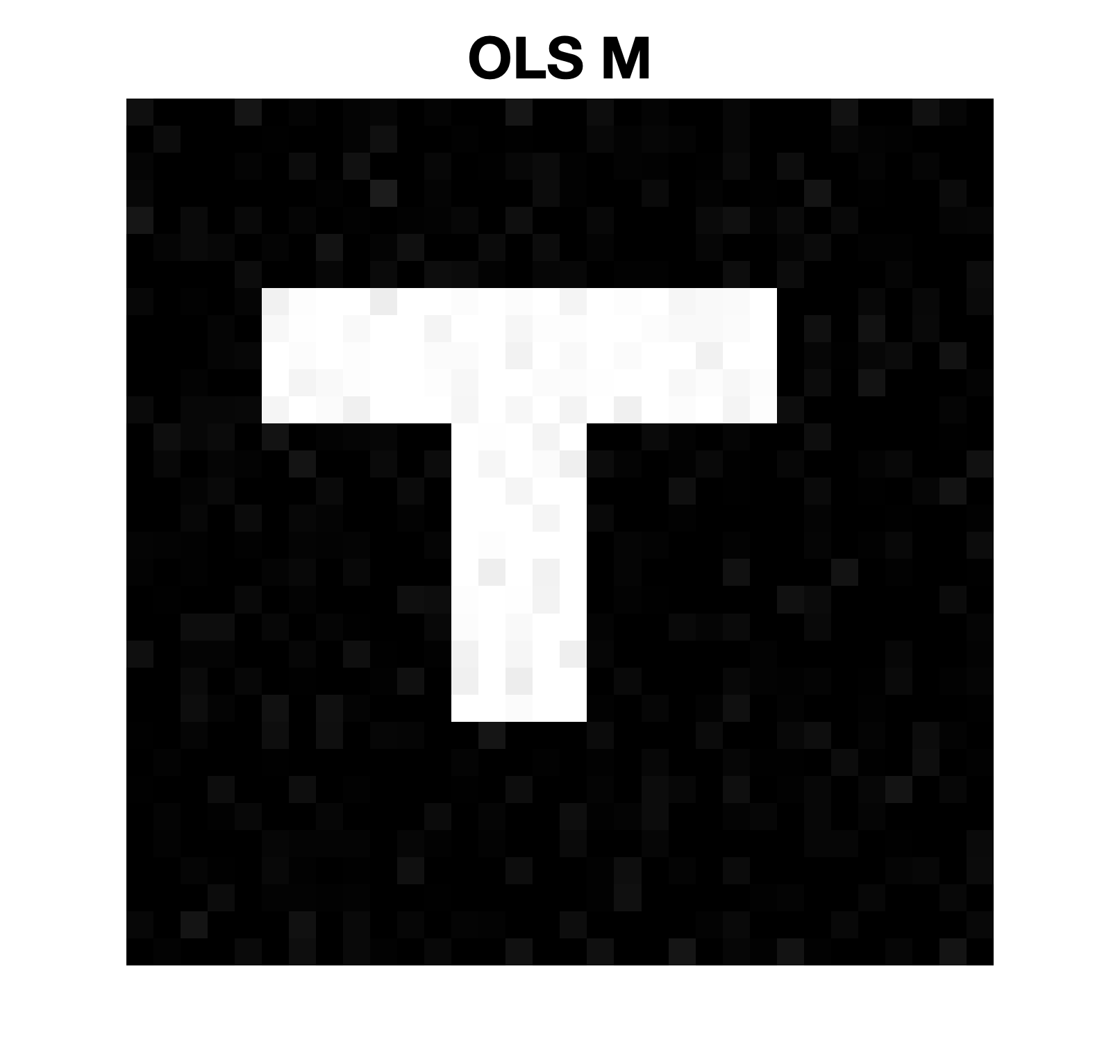}} 
\subfigure[]{\includegraphics[scale = 0.15]{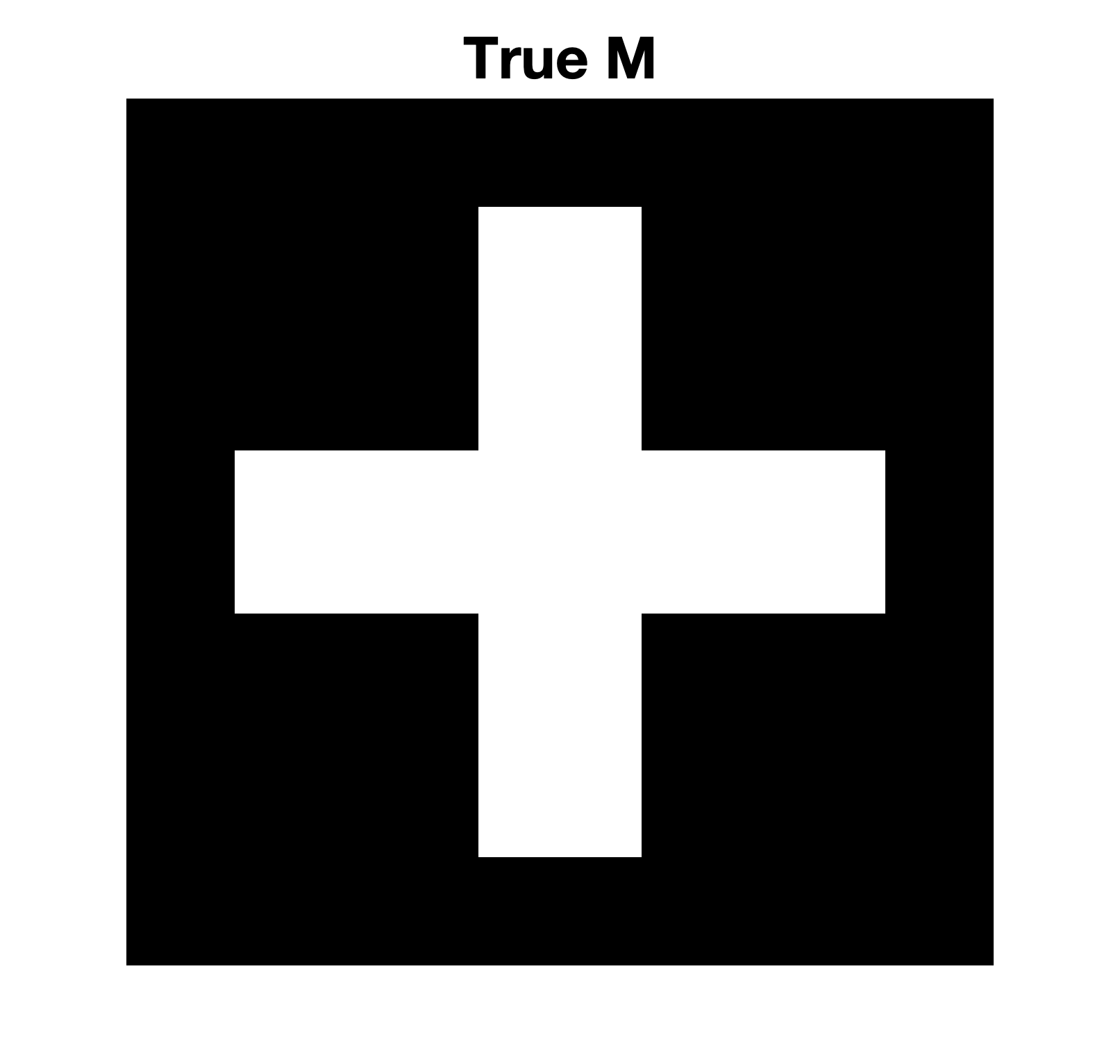}} \subfigure[]{\includegraphics[scale = 0.15]{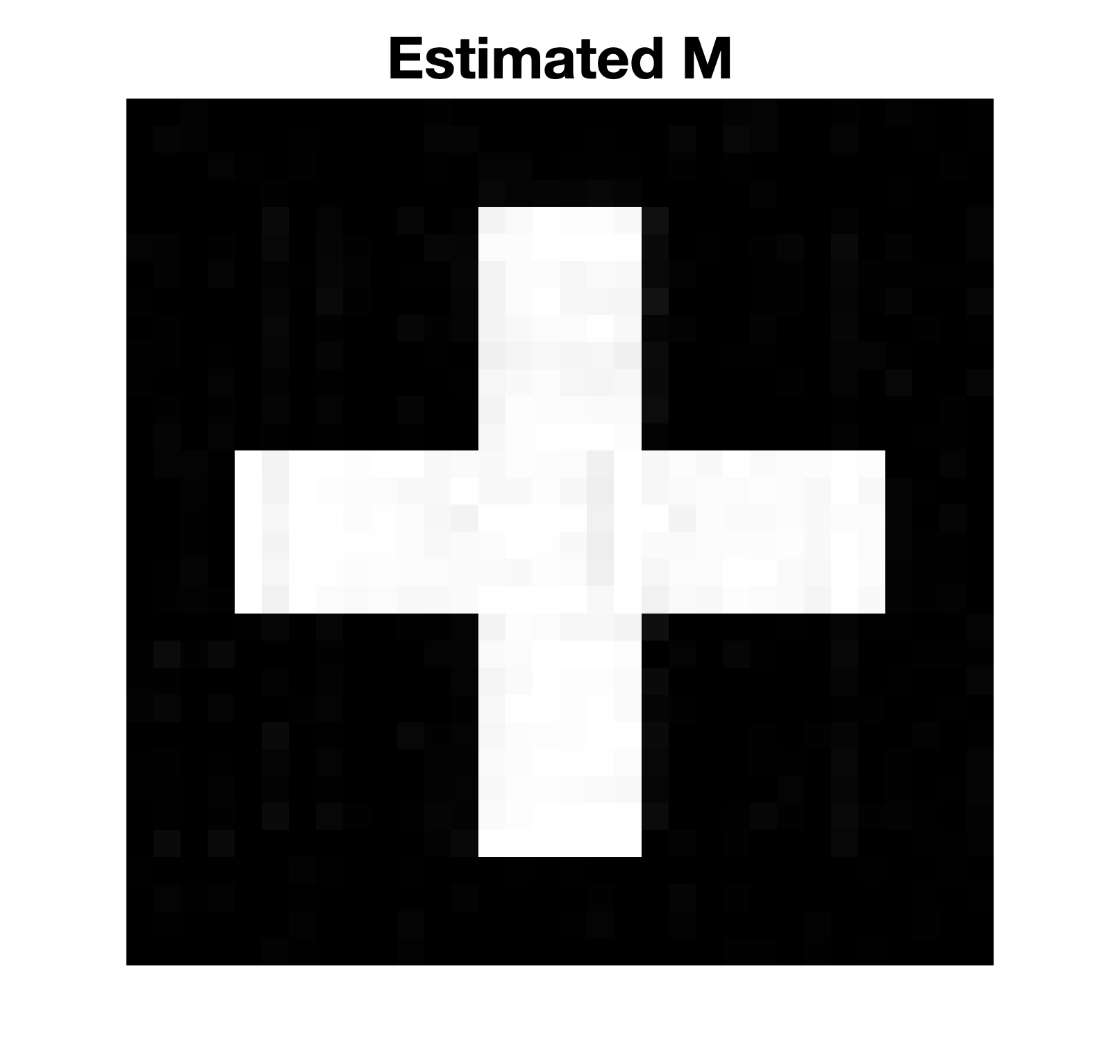}} \subfigure[]{\includegraphics[scale = 0.15]{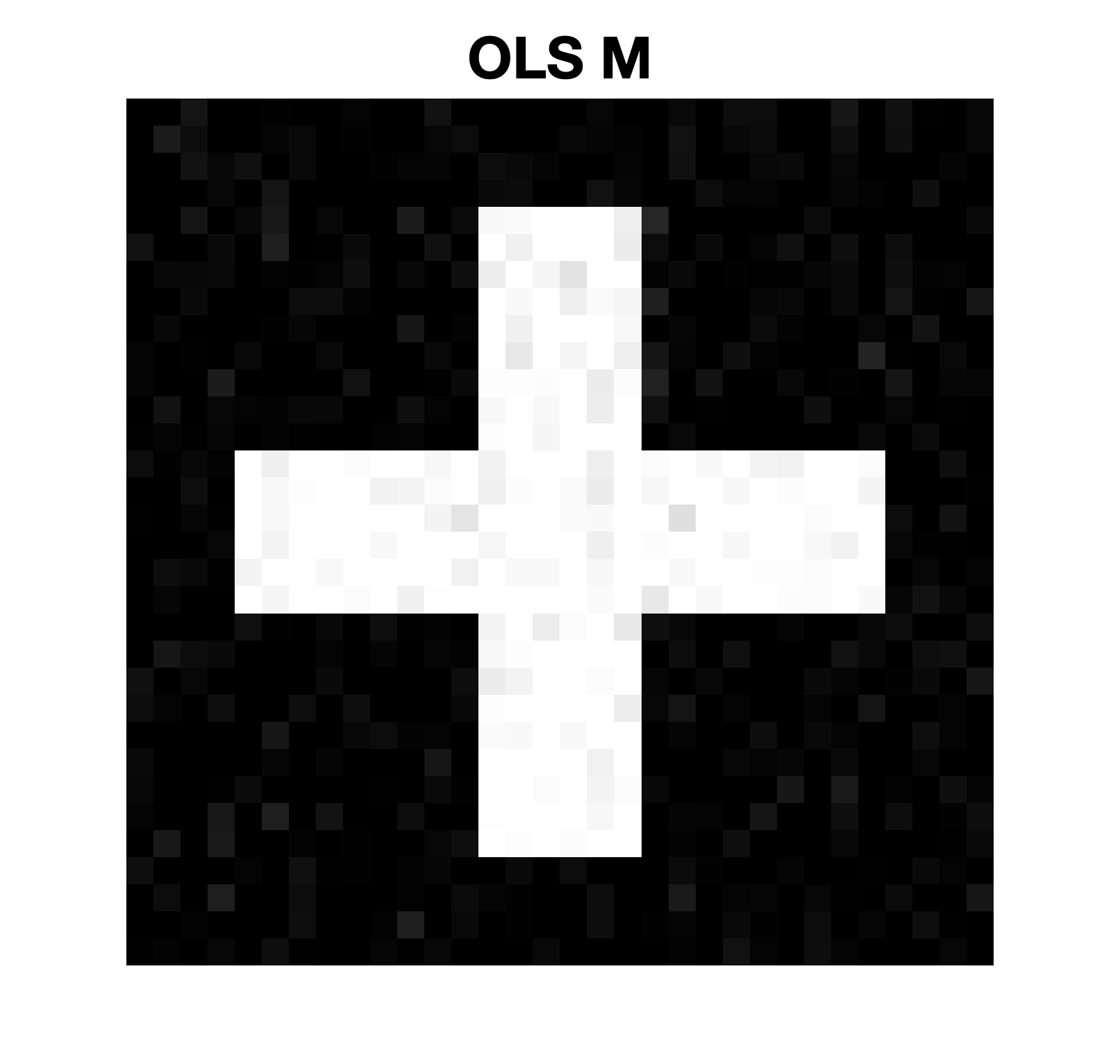}} 
\caption{Simulation results for the case $n=100$, $p=32$, $T=256$, $c=4$ and SNR = 1 from a randomly selected Monte Carlo run.
First column plots the shapes of true $\mathsf{M}$, second column plots estimated $\widehat{\mathsf{M}}$ based on the Fourier basis, and third column plots estimated $\widehat{\mathsf{M}}_{OLS}$ using OLS. The true basis is Fourier. The fitted basis is Chebyshev2.} 
\label{sim1snr1figure} 
\end{figure}

\begin{table}[H]
\centering
\begin{tabular}{r r r r}
\hline
MISE(sieve)	& Square	&T	&Cross\\
\hline
$\beta_1(t)$&1.944(0.049)&1.201(0.003)&2.984(0.018)\\
$\beta_2(t)$&1.946(0.045)&0.416(0.002)&2.987(0.018)\\
$\beta_3(t)$&1.958(0.048)&0.416(0.002)&2.984(0.019)\\
$\beta_4(t)$&1.952(0.048)&0.414(0.002)&1.425(0.008)\\
$\beta_5(t)$&1.954(0.048)&0.535(0.003)&1.426(0.006)\\
$\beta_6(t)$&1.937(0.045)&1.170(0.004)&1.783(0.019)\\
$\beta_7(t)$&1.944(0.049)&1.169(0.004)&1.799(0.019)\\
$\beta_8(t)$&4.198(0.067)&1.169(0.004)&1.798(0.019)\\
\hline
MISE(OLS)	& Square	&T	&Cross\\
\hline
$\beta_1(t)$&24.410(1.402)&2.075(0.006)&8.028(0.020)\\
$\beta_2(t)$&24.598(1.379)&0.931(0.005)&8.120(0.022)\\
$\beta_3(t)$&24.704(1.403)&0.935(0.005)&8.108(0.023)\\
$\beta_4(t)$&24.622(1.391)&0.938(0.006)&4.235(0.015)\\
$\beta_5(t)$&24.680(1.410)&1.450(0.007)&4.257(0.017)\\
$\beta_6(t)$&24.554(1.382)&2.260(0.008)&6.949(0.021)\\
$\beta_7(t)$&24.589(1.401)&2.267(0.008)&6.957(0.024)\\
$\beta_8(t)$&37.464(2.070)&2.230(0.007)&6.881(0.021)\\
\hline
\end{tabular}
\caption{Simulation results for $n=100$, $p=32$, $T=256$, $c=4$ and SNR = 5 for 100 Monte Carlo runs. The mean values of MISEs ($10^{-2}$) for 8 functional slope estimates are listed for the proposed method as well as the OLS method. Their associated standard errors are reported in the parentheses . For each method, the $\mathsf{M} \in \mathbb{R}^{32 \times 32}$ are chosen to be a $32$ by $32$ pictures of Square, T and Cross correspondingly. The true basis is Fourier. The fitted basis is Chebyshev2.}
\label{sim1:snr5:mise:basis21:tab}
\end{table}

\begin{table}[H]
\centering
\begin{tabular}{r r r r}
\hline
MISE(sieve)	& Square	&T	&Cross\\
\hline
$\beta_1(t)$&0.037(0.004)&0.020(0.001)&0.028(0.001)\\
$\beta_2(t)$&0.038(0.004)&0.019(0.001)&0.029(0.001)\\
$\beta_3(t)$&0.039(0.005)&0.016(0.001)&0.032(0.001)\\
$\beta_4(t)$&0.040(0.005)&0.018(0.001)&0.032(0.001)\\
$\beta_5(t)$&0.038(0.005)&0.020(0.001)&0.029(0.001)\\
$\beta_6(t)$&0.040(0.005)&0.020(0.001)&0.032(0.001)\\
$\beta_7(t)$&0.038(0.005)&0.019(0.001)&0.030(0.001)\\
$\beta_8(t)$&0.051(0.007)&0.018(0.000)&0.029(0.001)\\
\hline
MISE(OLS)	& Square	&T	&Cross\\
\hline
$\beta_1(t)$&0.659(0.130)&0.089(0.002)&0.139(0.003)\\
$\beta_2(t)$&0.714(0.130)&0.109(0.002)&0.175(0.004)\\
$\beta_3(t)$&0.702(0.132)&0.109(0.002)&0.175(0.004)\\
$\beta_4(t)$&0.713(0.134)&0.114(0.002)&0.179(0.004)\\
$\beta_5(t)$&0.708(0.132)&0.109(0.002)&0.174(0.003)\\
$\beta_6(t)$&0.717(0.135)&0.107(0.002)&0.176(0.004)\\
$\beta_7(t)$&0.709(0.133)&0.104(0.002)&0.177(0.004)\\
$\beta_8(t)$&0.890(0.195)&0.084(0.002)&0.137(0.003)\\
\hline
\end{tabular}
\caption{Simulation results for $n=100$, $p=32$, $T=256$, $c=4$ and SNR = 5 for 100 Monte Carlo runs using Fourier basis. The mean values of MISEs ($10^{-2}$) for 8 functional slope estimates are listed for the proposed method as well as the OLS method. Their associated standard errors are reported in the parentheses . For each method, the $\mathsf{M} \in \mathbb{R}^{32 \times 32}$ are chosen to be a $32$ by $32$ pictures of Square, T and Cross correspondingly. The true basis type is Fourier. The fitted basis is Fourier.}\label{sim1:snr5:mise:basis22:tab}
\end{table}

\begin{table}[H]
\centering
\begin{tabular}{rrrrrr}
\hline
 SNR & Basis & Square	&T	&Cross & ORACLE\\
\hline
1 & Fourier & 5.190(0.049) &4.000(0.000) &4.000(0.000) &4\\
5 & Fourier & 4.220(0.052) &4.000(0.000) &4.000(0.000) &4\\
10 & Fourier & 4.100(0.030) &4.000(0.000) &4.000(0.000) &4\\
\hline
1 & Chebyshev2 & 5.210(0.050) &4.000(0.000) &4.000(0.000) &4\\
5 & Chebyshev2 & 4.270(0.057) &4.000(0.000) &4.000(0.000) &4\\
10 & Chebyshev2 & 4.120(0.036) &4.000(0.000) &4.000(0.000) &4\\
\hline
\end{tabular}
\caption{Simulation results for $n=100$, $p=32$, $T=256$, $c=4$ for 100 Monte Carlo runs. The average number of basis selected as well as the oracle number of basis is listed for each of basis type, Fourier or Chebyshev2, and SNR, 1, 5 or 10. Their associated standard errors are reported in the parentheses .}
\label{sim1:c:tab}
\end{table}

\begin{figure}[H]
\centering
\includegraphics[height=2in,width=3.5in]{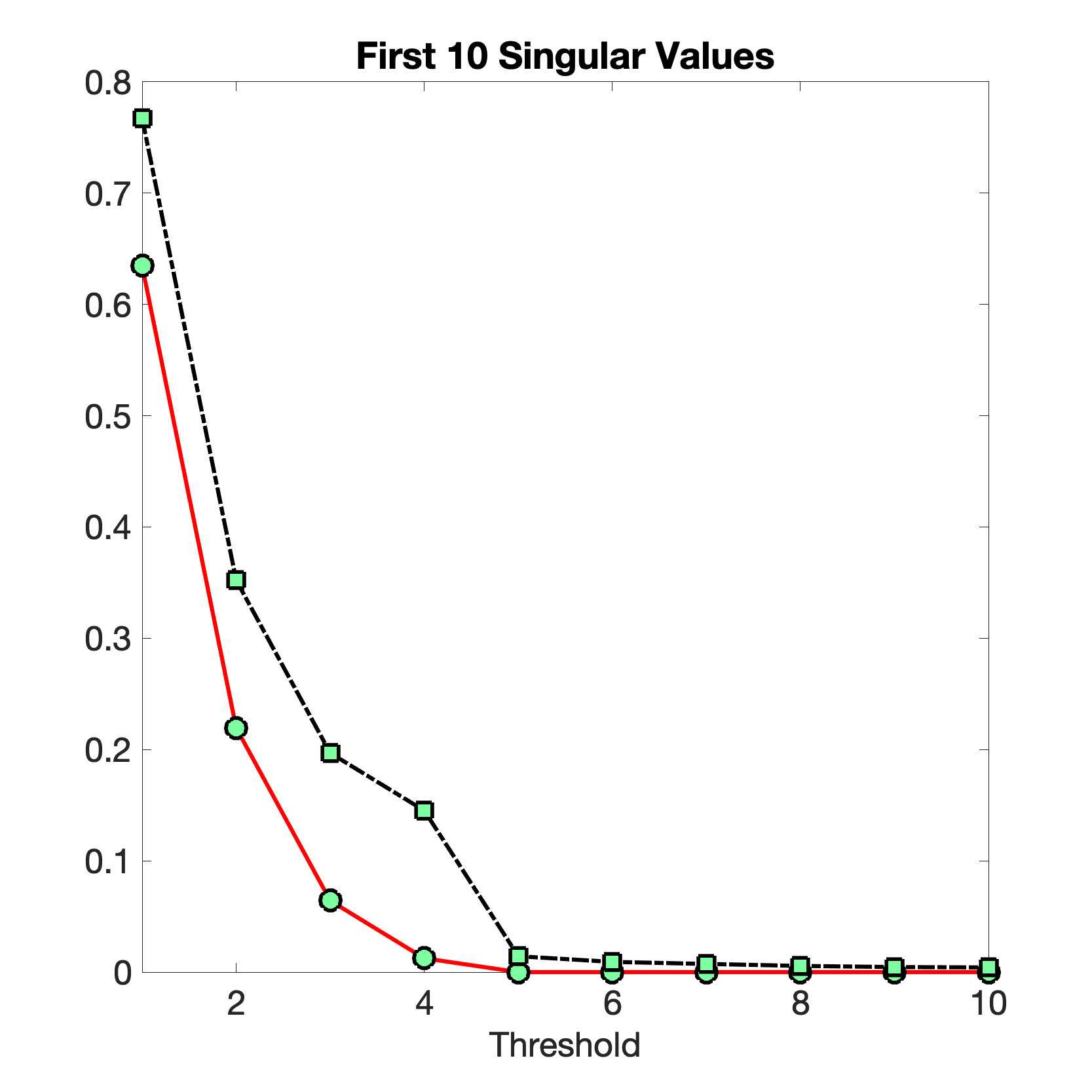}
\caption{Real Data Results: The scree plot for the first 10 singular values for the estimated  $\widehat{\mathsf{M}}$ (red solid) and $\widehat{\mathsf{M}}_{OLS}$ using OLS (black dashed). The fitted basis is Fourier. }
\label{dataPlotsScree}
\end{figure}
\begin{figure}[H]
\centering
  \subfigure[]{\includegraphics[height=1.5in,width=2in]{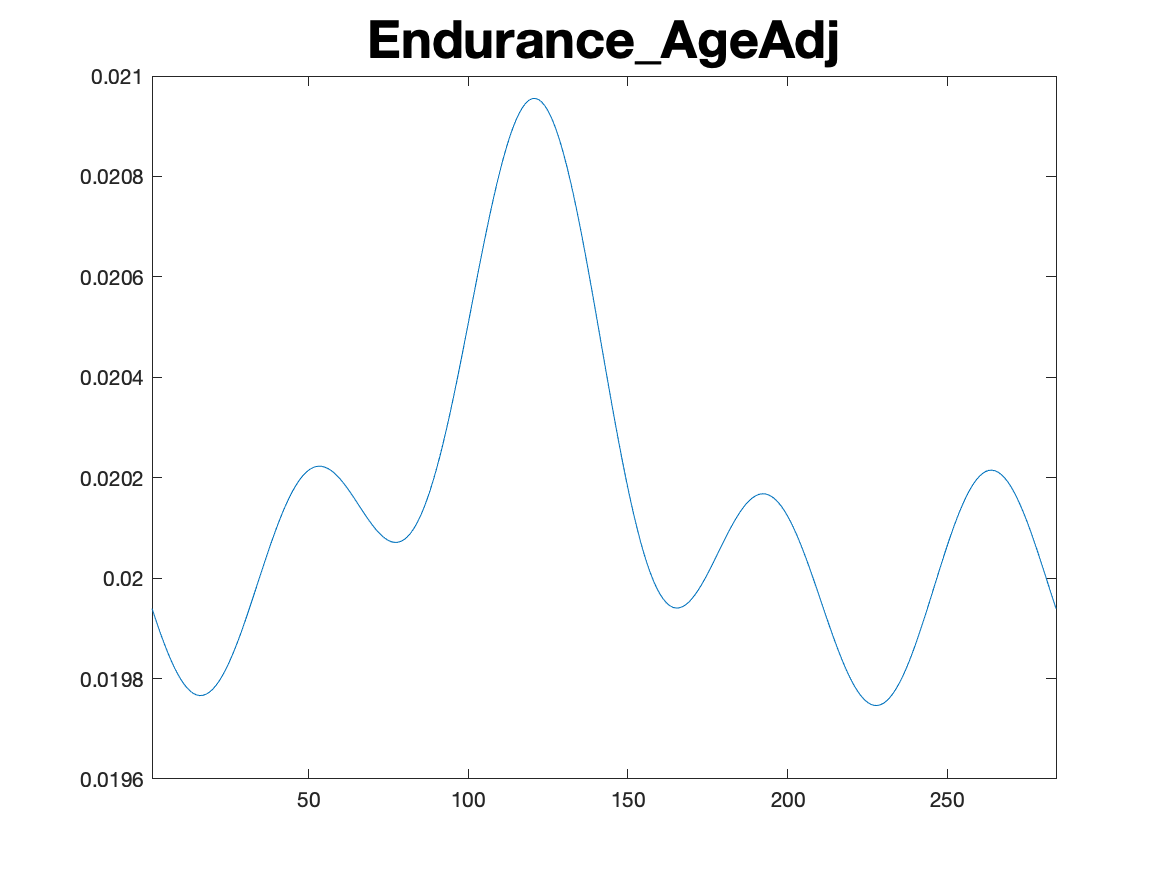}}
  \subfigure[]{\includegraphics[height=1.5in,width=2in]{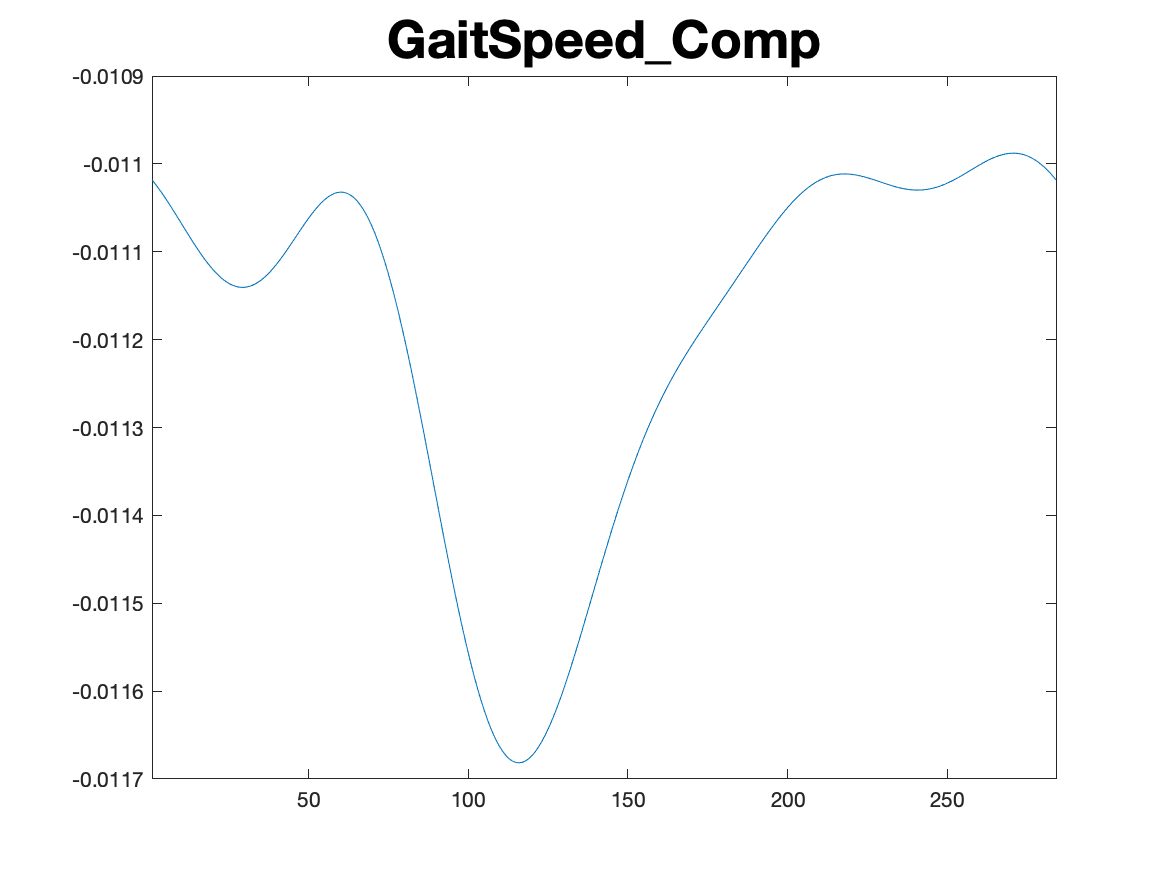}}
  \subfigure[]{\includegraphics[height=1.5in,width=2in]{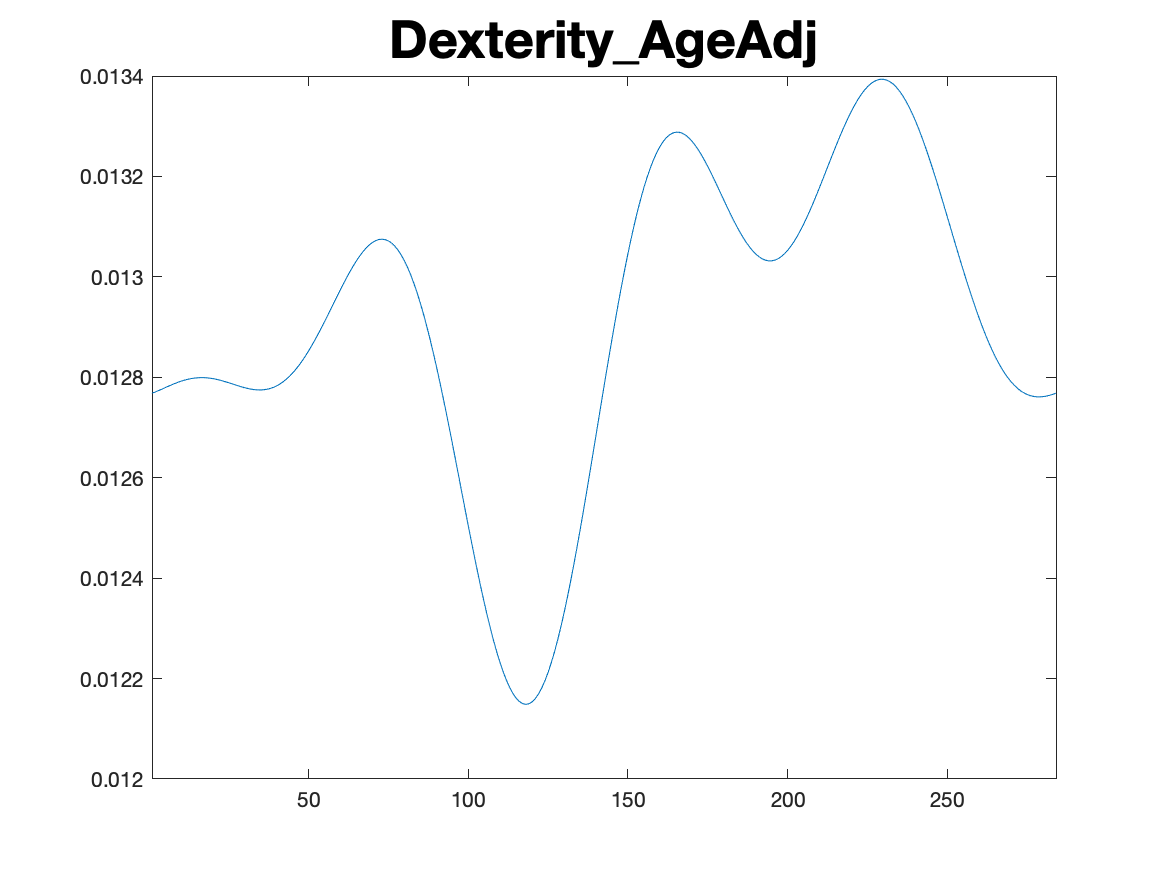}}
  \subfigure[]{\includegraphics[height=1.5in,width=2in]{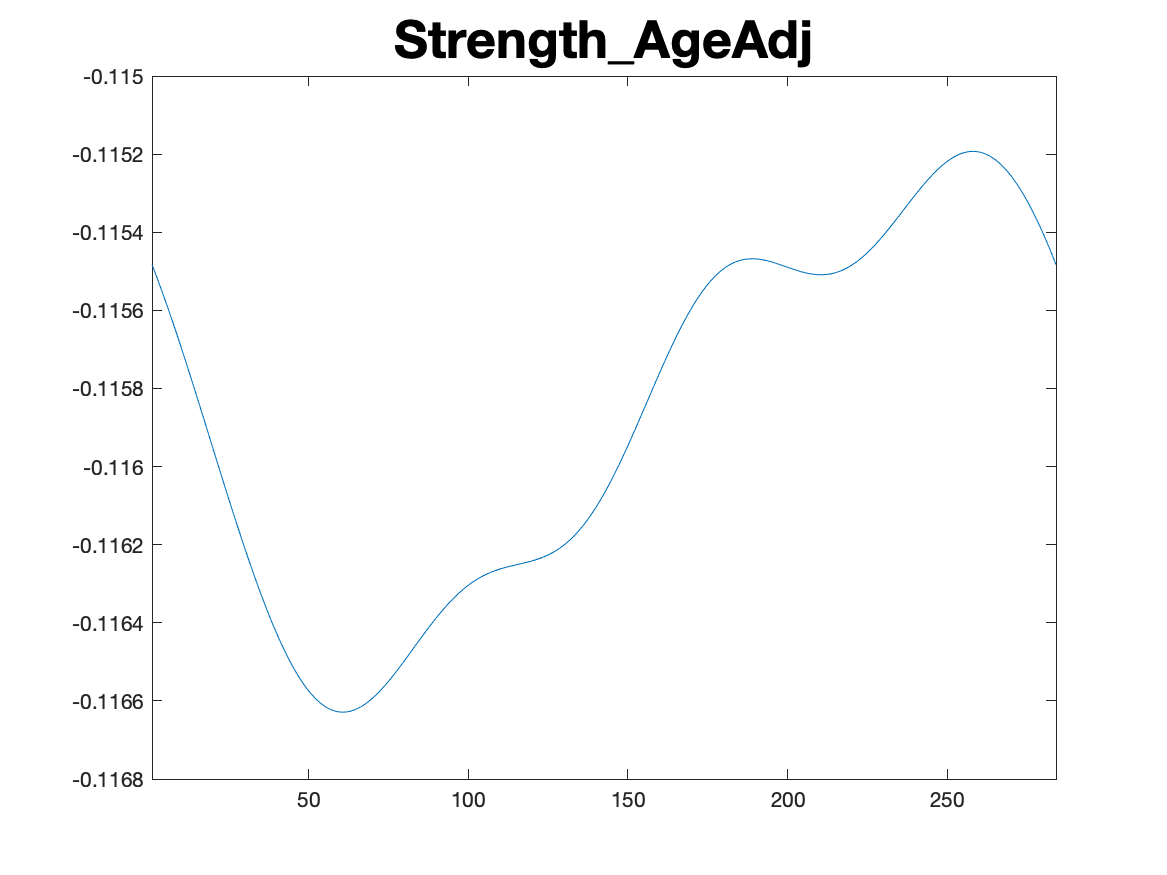}}
\caption{Real Data Results: Panel (a) - (d) plot, for left superior frontal region, the estimated $\hat{\beta}_j(t)$ for $s=4$ motor instrument covariates: ``Endurance-AgeAdj",    ``GaitSpeed-Comp",   ``Dexterity-AgeAdj",  and  ``Strength-AgeAdj". The fitted basis is Fourier. }
\label{dataPlots1}
\end{figure}
\begin{figure}[H]
\centering
  \subfigure[]{\includegraphics[height=1.5in,width=2in]{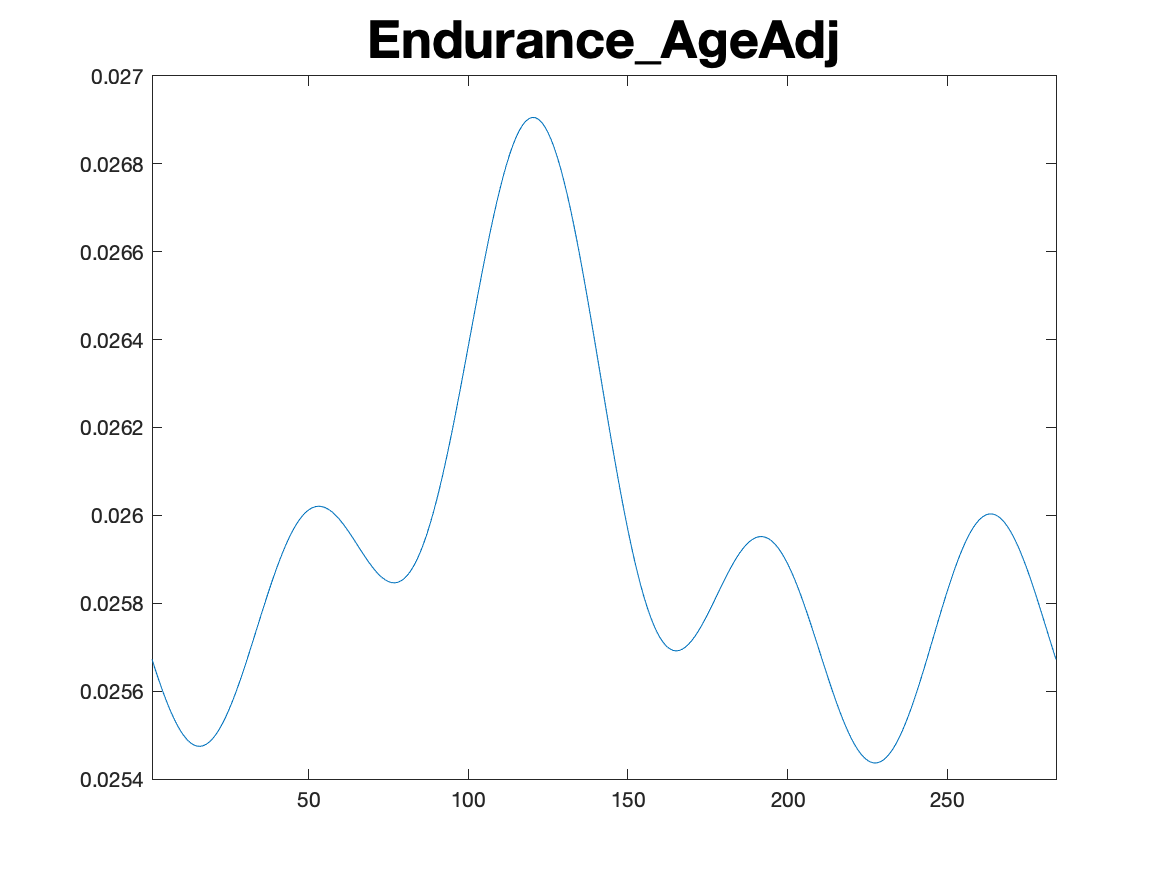}}
  \subfigure[]{\includegraphics[height=1.5in,width=2in]{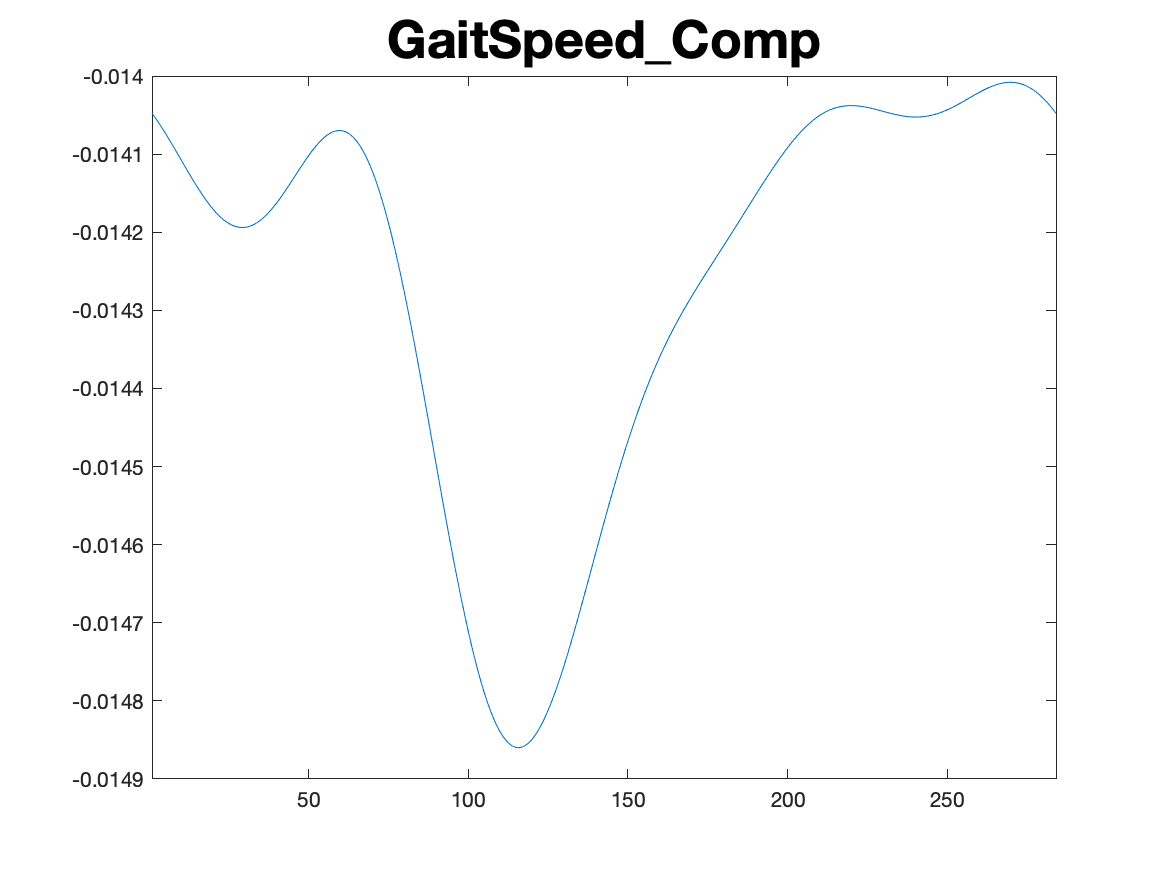}}
  \subfigure[]{\includegraphics[height=1.5in,width=2in]{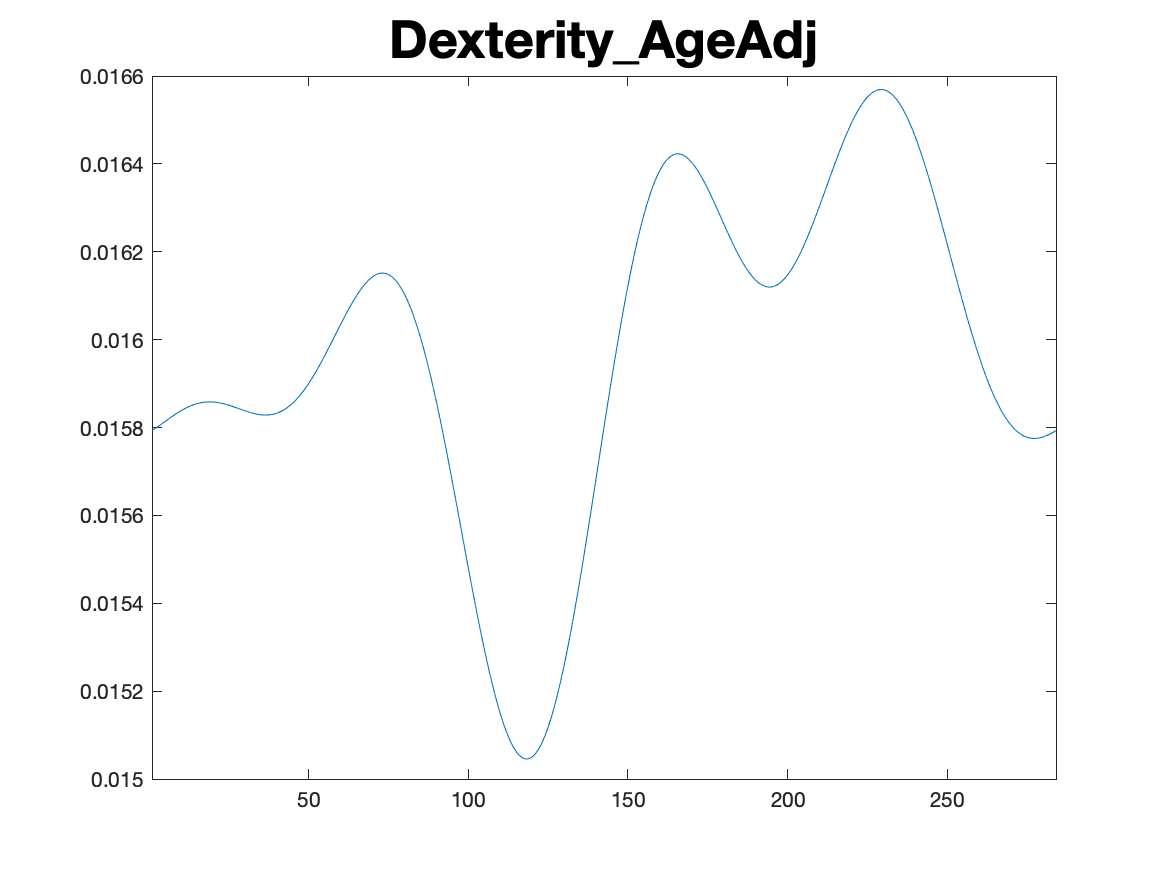}}
  \subfigure[]{\includegraphics[height=1.5in,width=2in]{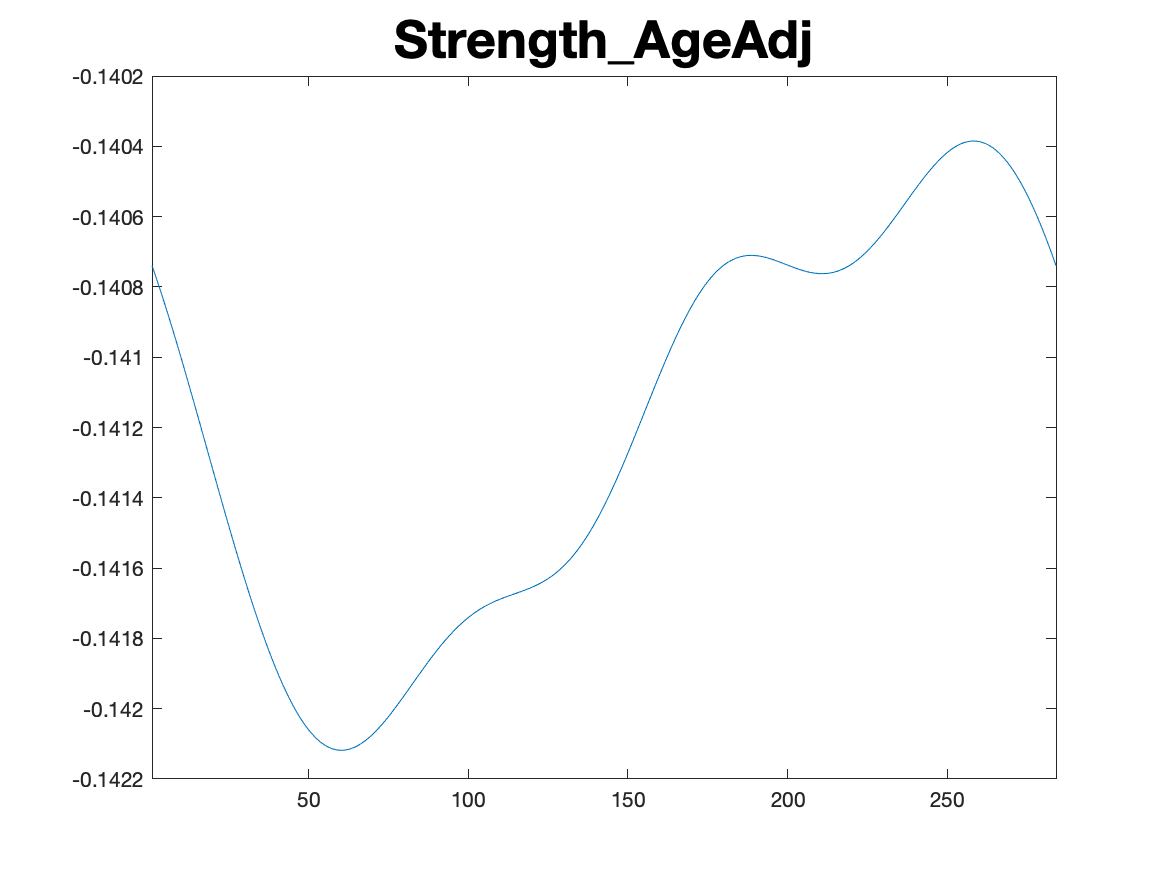}}
\caption{Real Data Results: Panel (a) - (d) plot, for right superior frontal region, the estimated $\{\hat{\beta}_j(t): 1\leq j\leq 4\}$ corresponding to four motor instrument covariates: ``Endurance-AgeAdj",    ``GaitSpeed-Comp",   ``Dexterity-AgeAdj",  and  ``Strength-AgeAdj", respectively. The fitted basis is Fourier. }
\label{dataPlots2}
\end{figure}

\begin{figure}[H]
\centering
  \subfigure[]{\includegraphics[height=1.5in,width=2in]{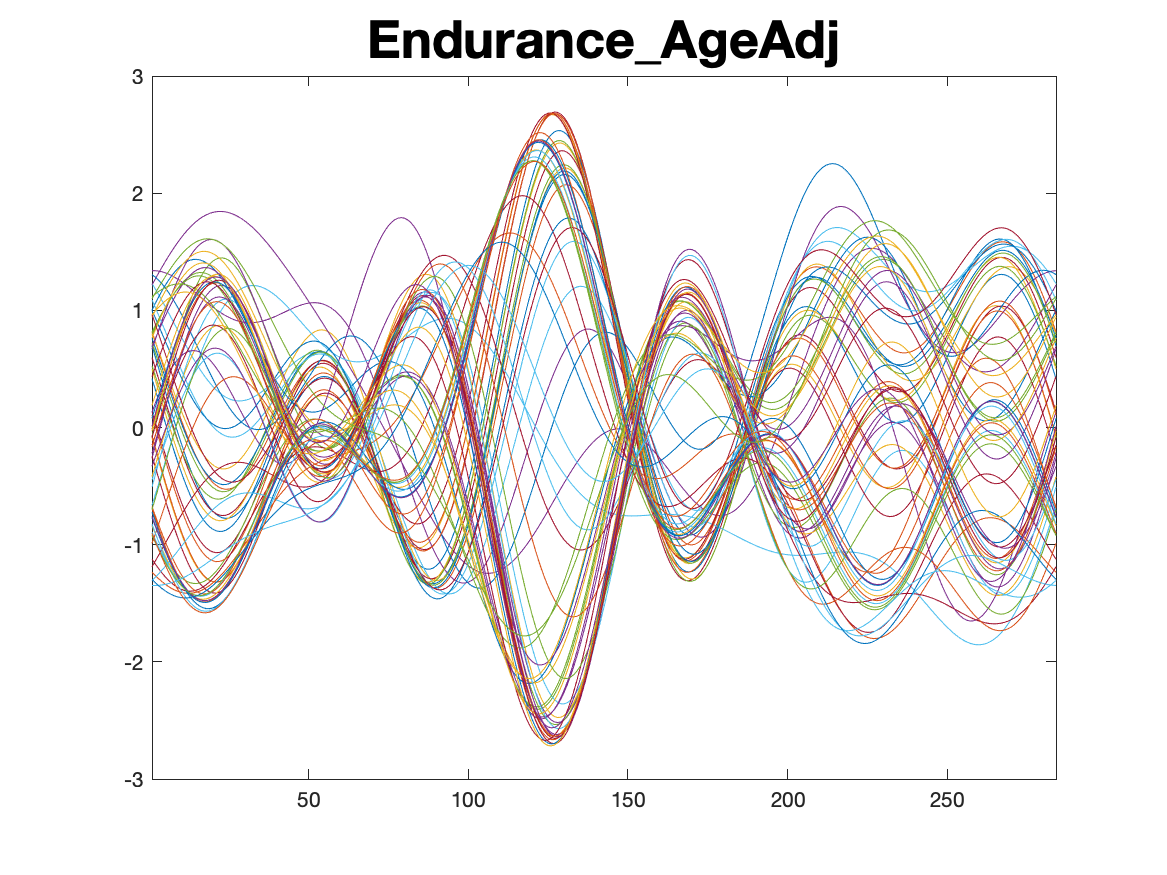}}
  \subfigure[]{\includegraphics[height=1.5in,width=2in]{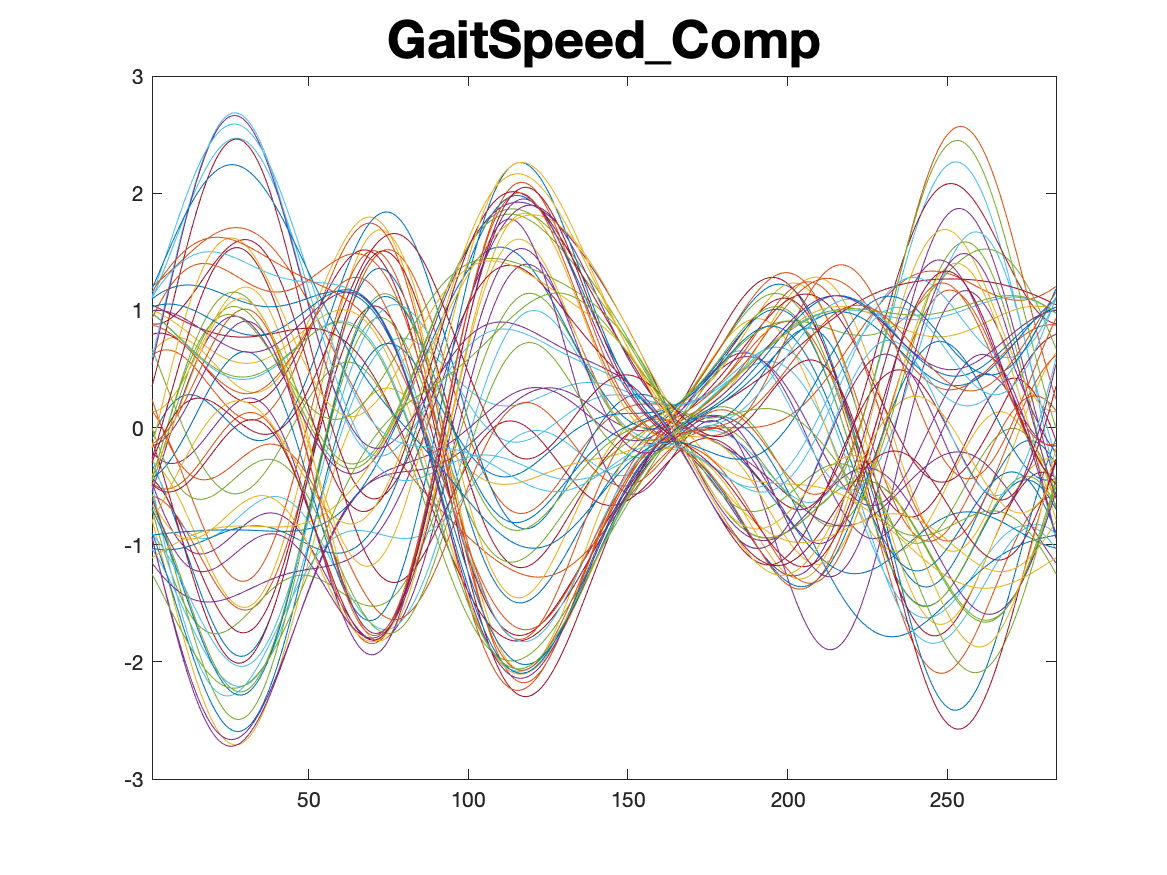}}
  \subfigure[]{\includegraphics[height=1.5in,width=2in]{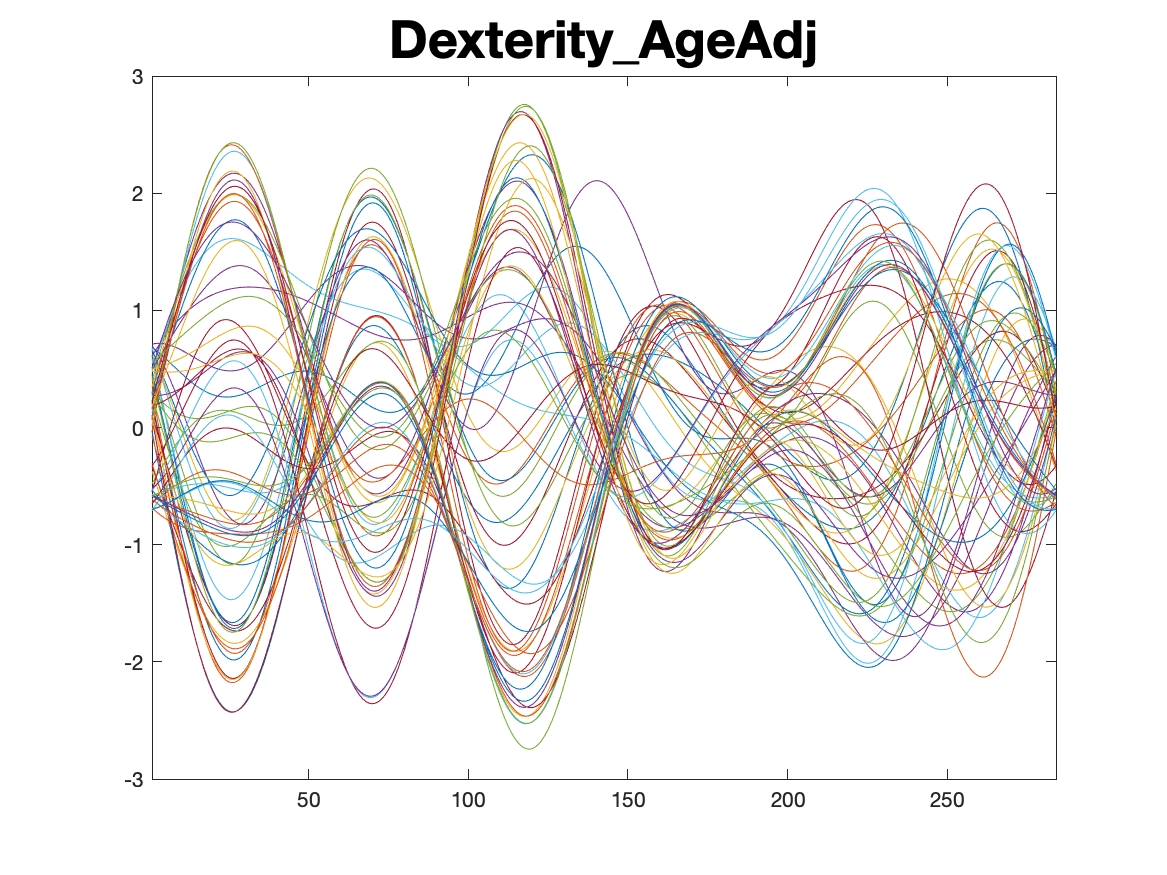}}
  \subfigure[]{\includegraphics[height=1.5in,width=2in]{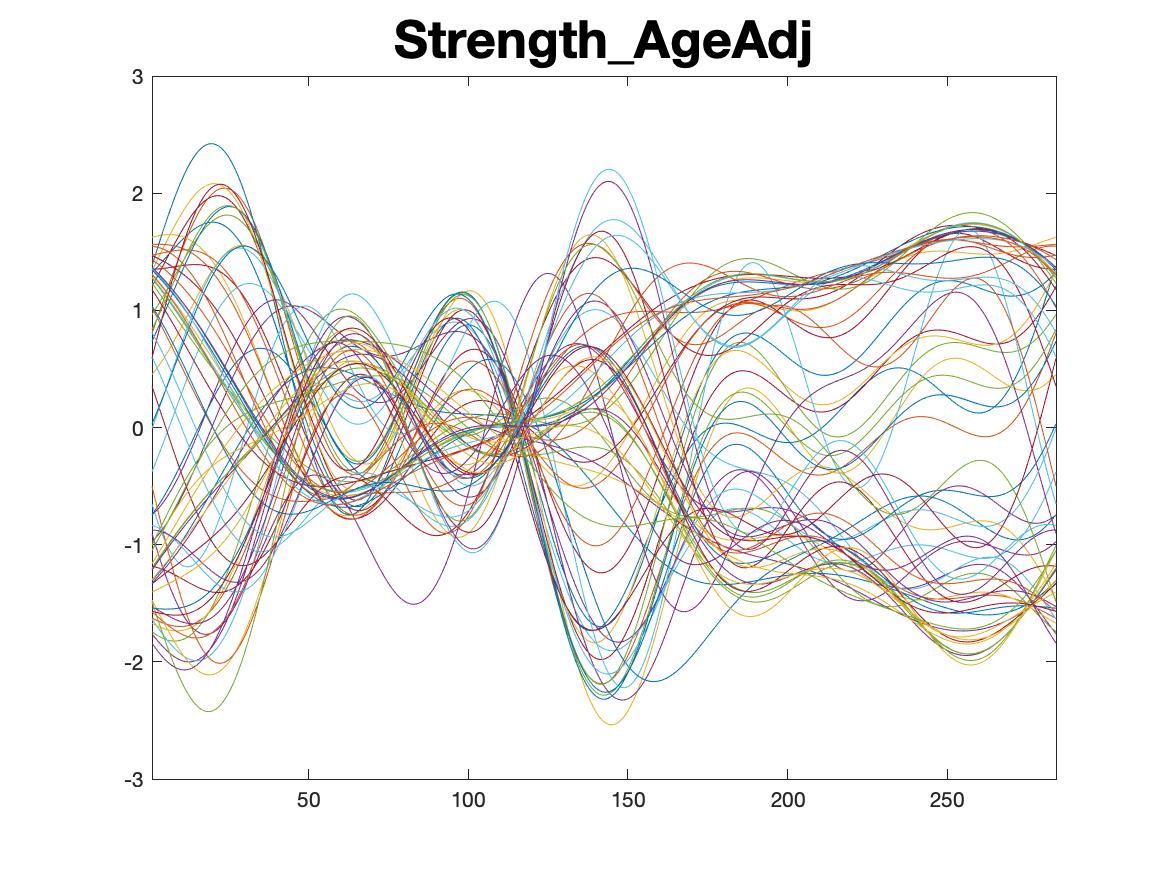}}
\caption{Real Data Results: Panel (a) - (d) plot the $\hat{\beta}_{j,\mathrm{stand}}(t)$'s of $68$ ROIs for $s=4$ motor instrument covariates: ``Endurance-AgeAdj",    ``GaitSpeed-Comp",   ``Dexterity-AgeAdj",  and  ``Strength-AgeAdj".  The fitted basis is Fourier. }
\label{dataPlots3}
\end{figure}

\begin{table}[htbp]
\centering
\begin{tabular}{r r r r}
\hline
MISE(sieve)	& Square	&T	&Cross\\
\hline
$\beta_1(t)$&0.166(0.008)&0.095(0.002)&0.144(0.005)\\
$\beta_2(t)$&0.175(0.009)&0.092(0.003)&0.145(0.006)\\
$\beta_3(t)$&0.181(0.009)&0.082(0.003)&0.152(0.005)\\
$\beta_4(t)$&0.194(0.010)&0.089(0.003)&0.147(0.006)\\
$\beta_5(t)$&0.190(0.011)&0.103(0.004)&0.140(0.005)\\
$\beta_6(t)$&0.178(0.009)&0.097(0.003)&0.158(0.006)\\
$\beta_7(t)$&0.190(0.009)&0.098(0.003)&0.142(0.005)\\
$\beta_8(t)$&0.243(0.012)&0.092(0.003)&0.132(0.004)\\
\hline
MISE(OLS)	& Square	&T	&Cross\\
\hline
$\beta_1(t)$&3.810(0.206)&0.429(0.008)&0.721(0.016)\\
$\beta_2(t)$&4.182(0.220)&0.555(0.012)&0.892(0.019)\\
$\beta_3(t)$&4.234(0.217)&0.535(0.010)&0.870(0.017)\\
$\beta_4(t)$&4.398(0.242)&0.542(0.011)&0.914(0.019)\\
$\beta_5(t)$&4.287(0.227)&0.533(0.012)&0.892(0.018)\\
$\beta_6(t)$&4.264(0.229)&0.532(0.011)&0.863(0.018)\\
$\beta_7(t)$&4.339(0.218)&0.536(0.011)&0.883(0.018)\\
$\beta_8(t)$&5.048(0.299)&0.425(0.007)&0.701(0.012)\\
\hline
\end{tabular}
\caption{Simulation results for $n=100$, $p=32$, $T=256$, $c=4$ and SNR = 1 for 100 Monte Carlo runs. The mean values of MISEs ($10^{-2}$) for 8 functional slope estimates are listed for the proposed method as well as the OLS method. Their associated standard errors are reported in the parentheses. For each method, the $\mathsf{M} \in \mathbb{R}^{32 \times 32}$ are chosen to be a $32$ by $32$ pictures of Square, T and Cross correspondingly. The true basis is Fourier. The fitted basis is Fourier.}
\label{sim1:snr1:mise:basis22:tab}
\end{table}

\begin{table}[htbp]
\centering
\begin{tabular}{r r r r}
\hline
MISE(sieve)	& Square	&T	&Cross\\
\hline
$\beta_1(t)$&0.016(0.001)&0.010(0.000)&0.014(0.000)\\
$\beta_2(t)$&0.017(0.001)&0.009(0.000)&0.016(0.001)\\
$\beta_3(t)$&0.017(0.002)&0.009(0.000)&0.017(0.001)\\
$\beta_4(t)$&0.017(0.001)&0.009(0.000)&0.015(0.001)\\
$\beta_5(t)$&0.017(0.001)&0.011(0.000)&0.015(0.000)\\
$\beta_6(t)$&0.016(0.001)&0.010(0.000)&0.017(0.001)\\
$\beta_7(t)$&0.017(0.001)&0.010(0.000)&0.016(0.001)\\
$\beta_8(t)$&0.021(0.002)&0.010(0.000)&0.014(0.000)\\
\hline
MISE(OLS)	& Square	&T	&Cross\\
\hline
$\beta_1(t)$&0.248(0.046)&0.044(0.001)&0.069(0.001)\\
$\beta_2(t)$&0.276(0.047)&0.056(0.001)&0.089(0.002)\\
$\beta_3(t)$&0.272(0.047)&0.056(0.001)&0.088(0.002)\\
$\beta_4(t)$&0.271(0.048)&0.055(0.001)&0.087(0.002)\\
$\beta_5(t)$&0.270(0.046)&0.054(0.001)&0.085(0.002)\\
$\beta_6(t)$&0.268(0.045)&0.055(0.001)&0.087(0.002)\\
$\beta_7(t)$&0.272(0.047)&0.053(0.001)&0.089(0.002)\\
$\beta_8(t)$&0.324(0.070)&0.043(0.001)&0.069(0.001)\\
\hline
\end{tabular}
\caption{Simulation results for $n=100$, $p=32$, $T=256$, $c=4$ and SNR = 10 for 100 Monte Carlo runs. The mean values of MISEs ($10^{-2}$) for 8 functional slope estimates are listed for the proposed method as well as the OLS method. Their associated standard errors are reported in the parentheses. For each method, the $\mathsf{M} \in \mathbb{R}^{32 \times 32}$ are chosen to be a $32$ by $32$ pictures of Square, T and Cross correspondingly. The true basis is Fourier. The fitted basis is Fourier.}
\label{sim1:snr10:mise:basis22:tab}
\end{table}

\begin{table}[htbp]
\centering
\begin{tabular}{r r r r}
\hline
MISE(sieve)	& Square	&T	&Cross\\
\hline
$\beta_1(t)$&1.946(0.048)&1.244(0.004)&3.007(0.023)\\
$\beta_2(t)$&1.974(0.050)&0.462(0.004)&3.030(0.019)\\
$\beta_3(t)$&1.960(0.053)&0.455(0.004)&3.034(0.021)\\
$\beta_4(t)$&1.953(0.051)&0.459(0.005)&1.505(0.010)\\
$\beta_5(t)$&1.962(0.050)&0.584(0.005)&1.494(0.012)\\
$\beta_6(t)$&1.946(0.055)&1.226(0.006)&1.826(0.021)\\
$\beta_7(t)$&1.969(0.050)&1.227(0.006)&1.825(0.022)\\
$\beta_8(t)$&4.152(0.069)&1.219(0.005)&1.829(0.020)\\
\hline
MISE(OLS)	& Square	&T	&Cross\\
\hline
$\beta_1(t)$&27.115(1.481)&2.596(0.016)&9.302(0.054)\\
$\beta_2(t)$&28.202(1.503)&1.606(0.019)&9.742(0.063)\\
$\beta_3(t)$&28.206(1.529)&1.612(0.020)&9.685(0.056)\\
$\beta_4(t)$&28.145(1.532)&1.631(0.023)&5.813(0.051)\\
$\beta_5(t)$&28.251(1.494)&2.097(0.021)&5.836(0.047)\\
$\beta_6(t)$&28.066(1.536)&2.922(0.023)&8.553(0.055)\\
$\beta_7(t)$&28.080(1.515)&2.890(0.025)&8.526(0.052)\\
$\beta_8(t)$&39.690(2.153)&2.750(0.019)&8.191(0.054)\\
\hline
\end{tabular}
\caption{Simulation results for $n=100$, $p=32$, $T=256$, $c=4$ and SNR = 1 for 100 Monte Carlo runs. The mean values of MISEs ($10^{-2}$) for 8 functional slope estimates are listed for the proposed method as well as the OLS method. Their associated standard errors are reported in the parentheses. For each method, the $\mathsf{M} \in \mathbb{R}^{32 \times 32}$ are chosen to be a $32$ by $32$ pictures of Square, T and Cross correspondingly. The true basis is Fourier. The fitted basis is Chebyshev2.}
\label{sim1:snr1:mise:basis21:tab}
\end{table}

\begin{table}[htbp]
\centering
\begin{tabular}{r r r r}
\hline
MISE(sieve)	& Square	&T	&Cross\\
\hline
$\beta_1(t)$&1.878(0.040)&1.201(0.003)&2.938(0.017)\\
$\beta_2(t)$&1.884(0.041)&0.412(0.001)&2.940(0.017)\\
$\beta_3(t)$&1.873(0.040)&0.413(0.001)&2.936(0.018)\\
$\beta_4(t)$&1.875(0.040)&0.412(0.001)&1.410(0.006)\\
$\beta_5(t)$&1.875(0.040)&0.533(0.003)&1.400(0.005)\\
$\beta_6(t)$&1.887(0.041)&1.170(0.003)&1.743(0.018)\\
$\beta_7(t)$&1.873(0.040)&1.168(0.003)&1.745(0.018)\\
$\beta_8(t)$&4.085(0.057)&1.169(0.003)&1.746(0.017)\\
\hline
MISE(OLS)	& Square	&T	&Cross\\
\hline
$\beta_1(t)$&22.550(1.093)&2.015(0.004)&7.893(0.013)\\
$\beta_2(t)$&22.674(1.096)&0.850(0.003)&7.919(0.016)\\
$\beta_3(t)$&22.596(1.088)&0.851(0.004)&7.913(0.017)\\
$\beta_4(t)$&22.628(1.096)&0.850(0.003)&4.031(0.012)\\
$\beta_5(t)$&22.601(1.092)&1.363(0.005)&4.028(0.012)\\
$\beta_6(t)$&22.700(1.106)&2.176(0.004)&6.771(0.016)\\
$\beta_7(t)$&22.600(1.090)&2.176(0.005)&6.763(0.016)\\
$\beta_8(t)$&34.711(1.618)&2.160(0.005)&6.734(0.013)\\
\hline
\end{tabular}
\caption{Simulation results for $n=100$, $p=32$, $T=256$, $c=4$ and SNR = 10 for 100 Monte Carlo runs. The mean values of MISEs ($10^{-2}$) for 8 functional slope estimates are listed for the proposed method as well as the OLS method. Their associated standard errors are reported in the parentheses. For each method, the $\mathsf{M} \in \mathbb{R}^{32 \times 32}$ are chosen to be a $32$ by $32$ pictures of Square, T and Cross correspondingly. The true basis is Fourier. The fitted basis is Chebyshev2.}
\label{sim1:snr10:mise:basis21:tab}
\end{table}

\begin{table}[htbp]
\centering
\begin{tabular}{r r r r}
\hline
MISE(sieve)	& Square	&T	&Cross\\
\hline
$\beta_1(t)$&0.159(0.010)&0.098(0.003)&0.144(0.005)\\
$\beta_2(t)$&0.163(0.009)&0.096(0.004)&0.137(0.005)\\
$\beta_3(t)$&0.174(0.010)&0.081(0.003)&0.164(0.006)\\
$\beta_4(t)$&0.172(0.010)&0.091(0.003)&0.140(0.005)\\
$\beta_5(t)$&0.156(0.010)&0.107(0.004)&0.143(0.006)\\
$\beta_6(t)$&0.165(0.010)&0.100(0.003)&0.156(0.005)\\
$\beta_7(t)$&0.174(0.011)&0.101(0.004)&0.139(0.006)\\
$\beta_8(t)$&0.210(0.013)&0.095(0.003)&0.135(0.004)\\
\hline
MISE(OLS)	& Square	&T	&Cross\\
\hline
$\beta_1(t)$&3.926(0.211)&0.432(0.009)&0.707(0.015)\\
$\beta_2(t)$&4.339(0.225)&0.539(0.012)&0.885(0.022)\\
$\beta_3(t)$&4.401(0.220)&0.519(0.012)&0.862(0.018)\\
$\beta_4(t)$&4.353(0.219)&0.539(0.013)&0.875(0.019)\\
$\beta_5(t)$&4.305(0.230)&0.544(0.012)&0.850(0.018)\\
$\beta_6(t)$&4.293(0.224)&0.532(0.013)&0.883(0.018)\\
$\beta_7(t)$&4.333(0.226)&0.547(0.015)&0.869(0.017)\\
$\beta_8(t)$&5.131(0.299)&0.433(0.010)&0.688(0.014)\\
\hline
\end{tabular}
\caption{Simulation results for $n=100$, $p=32$, $T=256$, $c=4$ and SNR = 1 for 100 Monte Carlo runs. The mean values of MISEs ($10^{-2}$) for 8 functional slope estimates are listed for the proposed method as well as the OLS method. Their associated standard errors are reported in the parentheses. For each method, the $\mathsf{M} \in \mathbb{R}^{32 \times 32}$ are chosen to be a $32$ by $32$ pictures of Square, T and Cross correspondingly. The true basis is Chebyshev2. The fitted basis is Chebyshev2.}
\label{sim1:snr1:mise:basis11:tab}
\end{table}

\begin{table}[htbp]
\centering
\begin{tabular}{r r r r}
\hline
MISE(sieve)	& Square	&T	&Cross\\
\hline
$\beta_1(t)$&0.044(0.006)&0.021(0.001)&0.029(0.001)\\
$\beta_2(t)$&0.044(0.005)&0.019(0.001)&0.031(0.001)\\
$\beta_3(t)$&0.050(0.006)&0.018(0.001)&0.034(0.001)\\
$\beta_4(t)$&0.046(0.005)&0.018(0.001)&0.032(0.001)\\
$\beta_5(t)$&0.046(0.006)&0.023(0.001)&0.030(0.001)\\
$\beta_6(t)$&0.045(0.005)&0.022(0.001)&0.034(0.001)\\
$\beta_7(t)$&0.045(0.005)&0.020(0.001)&0.033(0.001)\\
$\beta_8(t)$&0.060(0.009)&0.019(0.001)&0.028(0.001)\\
\hline
MISE(OLS)	& Square	&T	&Cross\\
\hline
$\beta_1(t)$&0.749(0.138)&0.088(0.002)&0.141(0.003)\\
$\beta_2(t)$&0.798(0.137)&0.110(0.003)&0.174(0.004)\\
$\beta_3(t)$&0.824(0.144)&0.106(0.003)&0.178(0.004)\\
$\beta_4(t)$&0.807(0.140)&0.108(0.002)&0.178(0.004)\\
$\beta_5(t)$&0.814(0.143)&0.111(0.003)&0.172(0.004)\\
$\beta_6(t)$&0.793(0.139)&0.109(0.003)&0.175(0.003)\\
$\beta_7(t)$&0.798(0.138)&0.106(0.002)&0.176(0.004)\\
$\beta_8(t)$&1.032(0.207)&0.086(0.002)&0.140(0.003)\\
\hline
\end{tabular}
\caption{Simulation results for $n=100$, $p=32$, $T=256$, $c=4$ and SNR = 5 for 100 Monte Carlo runs. The mean values of MISEs ($10^{-2}$) for 8 functional slope estimates are listed for the proposed method as well as the OLS method. Their associated standard errors are reported in the parentheses. For each method, the $\mathsf{M} \in \mathbb{R}^{32 \times 32}$ are chosen to be a $32$ by $32$ pictures of Square, T and Cross correspondingly. The true basis is Chebyshev2. The fitted basis is Chebyshev2.}
\label{sim1:snr5:mise:basis11:tab}
\end{table}
\begin{table}[htbp]
\centering
\begin{tabular}{r r r r}
\hline
MISE(sieve)	& Square	&T	&Cross\\
\hline
$\beta_1(t)$&0.018(0.002)&0.011(0.000)&0.014(0.000)\\
$\beta_2(t)$&0.018(0.002)&0.010(0.000)&0.016(0.000)\\
$\beta_3(t)$&0.019(0.003)&0.009(0.000)&0.016(0.001)\\
$\beta_4(t)$&0.018(0.002)&0.010(0.000)&0.015(0.000)\\
$\beta_5(t)$&0.018(0.002)&0.012(0.000)&0.015(0.001)\\
$\beta_6(t)$&0.019(0.003)&0.011(0.000)&0.016(0.001)\\
$\beta_7(t)$&0.020(0.003)&0.010(0.000)&0.016(0.001)\\
$\beta_8(t)$&0.023(0.004)&0.010(0.000)&0.014(0.000)\\
\hline
MISE(OLS)	& Square	&T	&Cross\\
\hline
$\beta_1(t)$&0.293(0.069)&0.043(0.001)&0.068(0.001)\\
$\beta_2(t)$&0.317(0.069)&0.053(0.001)&0.084(0.002)\\
$\beta_3(t)$&0.314(0.070)&0.053(0.001)&0.084(0.002)\\
$\beta_4(t)$&0.313(0.070)&0.054(0.001)&0.083(0.002)\\
$\beta_5(t)$&0.316(0.070)&0.054(0.001)&0.085(0.002)\\
$\beta_6(t)$&0.319(0.072)&0.053(0.001)&0.087(0.002)\\
$\beta_7(t)$&0.323(0.071)&0.053(0.001)&0.085(0.002)\\
$\beta_8(t)$&0.393(0.103)&0.042(0.001)&0.067(0.001)\\
\hline
\end{tabular}
\caption{Simulation results for $n=100$, $p=32$, $T=256$, $c=4$ and SNR = 10 for 100 Monte Carlo runs. The mean values of MISEs ($10^{-2}$) for 8 functional slope estimates are listed for the proposed method as well as the OLS method. Their associated standard errors are reported in the parentheses. For each method, the $\mathsf{M} \in \mathbb{R}^{32 \times 32}$ are chosen to be a $32$ by $32$ pictures of Square, T and Cross correspondingly. The true basis is Chebyshev2. The fitted basis is Chebyshev2.}
\label{sim1:snr10:mise:basis11:tab}
\end{table}
\begin{table}[htbp]
\centering
\begin{tabular}{r r r r}
\hline
MISE(sieve)	& Square	&T	&Cross\\
\hline
$\beta_1(t)$&5.515(0.104)&3.100(0.032)&11.250(0.062)\\
$\beta_2(t)$&5.521(0.103)&1.626(0.013)&11.260(0.072)\\
$\beta_3(t)$&5.598(0.096)&1.666(0.015)&11.338(0.071)\\
$\beta_4(t)$&5.515(0.109)&1.667(0.014)&7.910(0.065)\\
$\beta_5(t)$&5.572(0.101)&3.766(0.027)&3.116(0.039)\\
$\beta_6(t)$&5.599(0.105)&4.453(0.030)&6.996(0.075)\\
$\beta_7(t)$&5.570(0.106)&4.473(0.031)&6.963(0.076)\\
$\beta_8(t)$&8.087(0.149)&4.438(0.031)&6.914(0.078)\\
\hline
MISE(OLS)	& Square	&T	&Cross\\
\hline
$\beta_1(t)$&69.510(2.704)&13.131(0.331)&41.488(0.864)\\
$\beta_2(t)$&71.406(2.785)&7.504(0.183)&42.341(0.892)\\
$\beta_3(t)$&72.134(2.859)&7.610(0.176)&42.425(0.881)\\
$\beta_4(t)$&71.144(2.712)&7.748(0.187)&25.450(0.466)\\
$\beta_5(t)$&71.838(2.816)&14.171(0.318)&21.714(0.545)\\
$\beta_6(t)$&71.875(2.751)&16.817(0.382)&39.297(0.974)\\
$\beta_7(t)$&71.724(2.792)&16.868(0.388)&39.247(0.988)\\
$\beta_8(t)$&100.029(3.930)&16.492(0.378)&38.312(0.958)\\
\hline
\end{tabular}
\caption{Simulation results for $n=100$, $p=32$, $T=256$, $c=4$ and SNR = 1 for 100 Monte Carlo runs. The mean values of MISEs ($10^{-2}$) for 8 functional slope estimates are listed for the proposed method as well as the OLS method. Their associated standard errors are reported in the parentheses. For each method, the $\mathsf{M} \in \mathbb{R}^{32 \times 32}$ are chosen to be a $32$ by $32$ pictures of Square, T and Cross correspondingly. The true basis is Chebyshev2. The fitted basis is Fourier.}
\label{sim1:snr1:mise:basis12:tab}
\end{table}
\begin{table}[htbp]
\centering
\begin{tabular}{r r r r}
\hline
MISE(sieve)	& Square	&T	&Cross\\
\hline
$\beta_1(t)$&5.636(0.103)&3.024(0.035)&10.959(0.055)\\
$\beta_2(t)$&5.579(0.099)&1.608(0.010)&10.955(0.058)\\
$\beta_3(t)$&5.592(0.100)&1.603(0.011)&10.951(0.056)\\
$\beta_4(t)$&5.618(0.097)&1.605(0.010)&7.616(0.061)\\
$\beta_5(t)$&5.623(0.105)&3.781(0.023)&2.998(0.026)\\
$\beta_6(t)$&5.579(0.096)&4.451(0.027)&6.868(0.063)\\
$\beta_7(t)$&5.608(0.100)&4.435(0.027)&6.813(0.056)\\
$\beta_8(t)$&8.248(0.150)&4.440(0.027)&6.842(0.059)\\
\hline
MISE(OLS)	& Square	&T	&Cross\\
\hline
$\beta_1(t)$&64.673(2.387)&10.933(0.362)&41.135(0.846)\\
$\beta_2(t)$&64.822(2.420)&5.587(0.172)&41.347(0.847)\\
$\beta_3(t)$&64.864(2.416)&5.564(0.169)&41.255(0.852)\\
$\beta_4(t)$&65.097(2.429)&5.567(0.173)&22.879(0.384)\\
$\beta_5(t)$&65.035(2.427)&11.680(0.329)&19.530(0.478)\\
$\beta_6(t)$&64.810(2.417)&14.160(0.417)&38.493(0.934)\\
$\beta_7(t)$&64.921(2.413)&14.121(0.413)&38.313(0.911)\\
$\beta_8(t)$&95.809(3.591)&14.027(0.410)&38.240(0.923)\\
\hline
\end{tabular}
\caption{Simulation results for $n=100$, $p=32$, $T=256$, $c=4$ and SNR = 5 for 100 Monte Carlo runs. The mean values of MISEs ($10^{-2}$) for 8 functional slope estimates are listed for the proposed method as well as the OLS method. Their associated standard errors are reported in the parentheses. For each method, the $\mathsf{M} \in \mathbb{R}^{32 \times 32}$ are chosen to be a $32$ by $32$ pictures of Square, T and Cross correspondingly. The true basis is Chebyshev2. The fitted basis is Fourier.}
\label{sim1:snr5:mise:basis12:tab}
\end{table}
\begin{table}[htbp]
\centering
\begin{tabular}{r r r r}
\hline
MISE(sieve)	& Square	&T	&Cross\\
\hline
$\beta_1(t)$&5.603(0.085)&3.052(0.028)&11.043(0.060)\\
$\beta_2(t)$&5.604(0.083)&1.593(0.009)&11.028(0.058)\\
$\beta_3(t)$&5.608(0.084)&1.592(0.009)&11.048(0.061)\\
$\beta_4(t)$&5.612(0.085)&1.590(0.009)&7.596(0.055)\\
$\beta_5(t)$&5.617(0.084)&3.669(0.025)&3.034(0.027)\\
$\beta_6(t)$&5.584(0.084)&4.357(0.028)&6.913(0.063)\\
$\beta_7(t)$&5.638(0.084)&4.358(0.028)&6.919(0.067)\\
$\beta_8(t)$&8.245(0.128)&4.359(0.028)&6.923(0.064)\\
\hline
MISE(OLS)	& Square	&T	&Cross\\
\hline
$\beta_1(t)$&61.110(2.285)&12.198(0.257)&40.356(0.644)\\
$\beta_2(t)$&61.348(2.298)&6.066(0.121)&40.385(0.651)\\
$\beta_3(t)$&61.322(2.285)&6.046(0.119)&40.471(0.653)\\
$\beta_4(t)$&61.367(2.294)&6.077(0.120)&22.103(0.284)\\
$\beta_5(t)$&61.387(2.299)&12.748(0.233)&18.963(0.374)\\
$\beta_6(t)$&61.203(2.300)&15.589(0.295)&37.552(0.719)\\
$\beta_7(t)$&61.542(2.297)&15.616(0.293)&37.601(0.720)\\
$\beta_8(t)$&91.039(3.399)&15.554(0.294)&37.433(0.717)\\
\hline
\end{tabular}
\caption{Simulation results for $n=100$, $p=32$, $T=256$, $c=4$ and SNR = 10 for 100 Monte Carlo runs. The mean values of MISEs ($10^{-2}$) for 8 functional slope estimates are listed for the proposed method as well as the OLS method. Their associated standard errors are reported in the parentheses. For each method, the $\mathsf{M} \in \mathbb{R}^{32 \times 32}$ are chosen to be a $32$ by $32$ pictures of Square, T and Cross correspondingly. The true basis is Chebyshev2. The fitted basis is Fourier.}
\label{sim1:snr10:mise:basis12:tab}
\end{table}

\begin{table}[htbp]
\centering
\begin{tabular}{rrrr}
\hline
MISE(sieve)	& Square	&T	&Cross\\
\hline
$\beta_1(t)$&0.190(0.012)&0.097(0.003)&0.146(0.006)\\
$\beta_2(t)$&0.208(0.014)&0.094(0.004)&0.151(0.008)\\
$\beta_3(t)$&0.211(0.017)&0.089(0.004)&0.164(0.008)\\
$\beta_4(t)$&0.194(0.014)&0.094(0.004)&0.158(0.009)\\
$\beta_5(t)$&0.227(0.015)&0.103(0.003)&0.160(0.008)\\
$\beta_6(t)$&0.191(0.013)&0.103(0.004)&0.172(0.009)\\
$\beta_7(t)$&0.229(0.016)&0.098(0.003)&0.149(0.007)\\
$\beta_8(t)$&0.271(0.021)&0.096(0.003)&0.140(0.005)\\
\hline
Rank($\widehat{\mathsf{M}}_{\mathrm{sieve}}$)&1.440(0.219) &13.800(0.192) &13.010(0.260)\\
\hline
MISE($c=6$)	& Square	&T	&Cross\\
\hline
$\beta_1(t)$&0.322(0.028)&0.482(0.021)&0.810(0.038)\\
$\beta_2(t)$&0.357(0.034)&0.250(0.010)&0.854(0.041)\\
$\beta_3(t)$&0.357(0.033)&0.249(0.011)&0.864(0.038)\\
$\beta_4(t)$&0.334(0.032)&0.260(0.010)&0.433(0.018)\\
$\beta_5(t)$&0.377(0.035)&0.486(0.022)&0.424(0.016)\\
$\beta_6(t)$&0.328(0.032)&0.605(0.026)&0.865(0.038)\\
$\beta_7(t)$&0.381(0.036)&0.605(0.025)&0.851(0.041)\\
$\beta_8(t)$&0.485(0.045)&0.562(0.023)&0.831(0.036)\\
\hline
Rank($\widehat{\mathsf{M}}_{c=6}$)&1.000(0.000) &2.000(0.000) &2.000(0.000)\\
\hline
\end{tabular}
\caption{Simulation results for $n=100$, $p=32$, $T=256$, $c=4$ and SNR = 1 for 100 Monte Carlo runs. The mean values of MISEs ($10^{-2}$) for 8 functional slope estimates are listed for the proposed methods when  $c$ is chosen by cross validation and $c$ is chosen to be fixed as $6$. Their associated standard errors are reported in the parentheses. The average ranks of $\widehat{\mathsf{M}}$ are also reported with the associated standard errors in the parentheses. For each method, the $\mathsf{M} \in \mathbb{R}^{32 \times 32}$ are chosen to be a $32$ by $32$ pictures of Square, T and Cross correspondingly. The true basis type is Fourier. The fitted basis type is Fourier.}
\label{sim1v1c6:snr1:mise:basis22:tab}
\end{table}
\begin{table}[htbp]
\centering
\begin{tabular}{rrrr}
\hline
MISE(sieve)	& Square	&T	&Cross\\
\hline
$\beta_1(t)$&0.046(0.008)&0.020(0.001)&0.030(0.001)\\
$\beta_2(t)$&0.051(0.009)&0.018(0.001)&0.033(0.001)\\
$\beta_3(t)$&0.049(0.007)&0.017(0.001)&0.035(0.001)\\
$\beta_4(t)$&0.050(0.009)&0.018(0.001)&0.032(0.001)\\
$\beta_5(t)$&0.050(0.008)&0.022(0.001)&0.033(0.001)\\
$\beta_6(t)$&0.045(0.007)&0.021(0.001)&0.036(0.001)\\
$\beta_7(t)$&0.051(0.009)&0.021(0.001)&0.032(0.001)\\
$\beta_8(t)$&0.063(0.012)&0.019(0.001)&0.031(0.001)\\
\hline
Rank($\widehat{\mathsf{M}}_{\mathrm{sieve}}$)&9.790(0.448) &13.410(0.303) &13.560(0.311)\\
\hline
MISE($c=6$)	& Square	&T	&Cross\\
\hline
$\beta_1(t)$&0.244(0.023)&0.408(0.015)&0.743(0.033)\\
$\beta_2(t)$&0.259(0.025)&0.198(0.007)&0.753(0.033)\\
$\beta_3(t)$&0.250(0.025)&0.197(0.006)&0.751(0.034)\\
$\beta_4(t)$&0.259(0.024)&0.197(0.006)&0.346(0.013)\\
$\beta_5(t)$&0.249(0.024)&0.408(0.015)&0.356(0.012)\\
$\beta_6(t)$&0.243(0.023)&0.505(0.018)&0.741(0.033)\\
$\beta_7(t)$&0.244(0.024)&0.497(0.018)&0.751(0.033)\\
$\beta_8(t)$&0.369(0.036)&0.498(0.018)&0.736(0.033)\\
\hline
Rank($\widehat{\mathsf{M}}_{c=6}$)&1.000(0.000) &2.000(0.000) &2.000(0.000)\\
\hline
\end{tabular}
\caption{Simulation results for $n=100$, $p=32$, $T=256$, $c=4$ and SNR = 5 for 100 Monte Carlo runs. The mean values of MISEs ($10^{-2}$) for 8 functional slope estimates are listed for the proposed methods when  $c$ is chosen by cross validation and $c$ is chosen to be fixed as $6$. Their associated standard errors are reported in the parentheses. The average ranks of $\widehat{\mathsf{M}}$ are also reported with the associated standard errors in the parentheses. For each method, the $\mathsf{M} \in \mathbb{R}^{32 \times 32}$ are chosen to be a $32$ by $32$ pictures of Square, T and Cross correspondingly. The true basis type is Fourier. The fitted basis type is Fourier.}
\label{sim1v1c6:snr5:mise:basis22:tab}
\end{table}
\begin{table}[htbp]
\centering
\begin{tabular}{rrrr}
\hline
MISE(sieve)	& Square	&T	&Cross\\
\hline
$\beta_1(t)$&0.022(0.004)&0.011(0.000)&0.014(0.000)\\
$\beta_2(t)$&0.024(0.005)&0.010(0.001)&0.016(0.001)\\
$\beta_3(t)$&0.023(0.004)&0.009(0.000)&0.018(0.001)\\
$\beta_4(t)$&0.024(0.004)&0.009(0.000)&0.015(0.001)\\
$\beta_5(t)$&0.025(0.004)&0.011(0.000)&0.016(0.001)\\
$\beta_6(t)$&0.025(0.005)&0.011(0.001)&0.017(0.001)\\
$\beta_7(t)$&0.024(0.004)&0.011(0.000)&0.016(0.001)\\
$\beta_8(t)$&0.030(0.006)&0.009(0.000)&0.015(0.001)\\
\hline
Rank($\widehat{\mathsf{M}}_{\mathrm{sieve}}$)&10.540(0.369) &13.490(0.351) &13.950(0.250)\\
\hline
MISE($c=6$)	& Square	&T	&Cross\\
\hline
$\beta_1(t)$&0.296(0.031)&0.407(0.016)&0.681(0.027)\\
$\beta_2(t)$&0.292(0.031)&0.196(0.007)&0.683(0.028)\\
$\beta_3(t)$&0.298(0.032)&0.196(0.006)&0.689(0.027)\\
$\beta_4(t)$&0.295(0.031)&0.193(0.007)&0.313(0.009)\\
$\beta_5(t)$&0.293(0.031)&0.412(0.016)&0.315(0.010)\\
$\beta_6(t)$&0.290(0.031)&0.504(0.019)&0.679(0.027)\\
$\beta_7(t)$&0.292(0.031)&0.509(0.019)&0.687(0.028)\\
$\beta_8(t)$&0.439(0.046)&0.505(0.018)&0.678(0.027)\\
\hline
Rank($\widehat{\mathsf{M}}_{c=6}$)&1.000(0.000) &2.000(0.000) &2.000(0.000)\\
\hline
\end{tabular}
\caption{Simulation results for $n=100$, $p=32$, $T=256$, $c=4$ and SNR = 10 for 100 Monte Carlo runs. The mean values of MISEs ($10^{-2}$) for 8 functional slope estimates are listed for the proposed methods when  $c$ is chosen by cross validation and $c$ is chosen to be fixed as $6$. Their associated standard errors are reported in the parentheses. The average ranks of $\widehat{\mathsf{M}}$ are also reported with the associated standard errors in the parentheses. For each method, the $\mathsf{M} \in \mathbb{R}^{32 \times 32}$ are chosen to be a $32$ by $32$ pictures of Square, T and Cross correspondingly. The true basis type is Fourier. The fitted basis type is Fourier.}
\label{sim1v1c6:snr10:mise:basis22:tab}
\end{table}

\begin{table}[htbp]
\centering
\begin{tabular}{rrrr}
\hline
Rank(sieve)	& Square	&T	&Cross\\
\hline
SNR = 1 &1.000(0.000) &2.030(0.022) &2.000(0.000)\\
SNR = 5 &1.000(0.000) &2.000(0.000) &2.000(0.000)\\
SNR = 10 &1.000(0.000) &2.000(0.000) &2.000(0.000)\\
\hline
\end{tabular}
\caption{Simulation results for $n=100$, $p=32$, $T=256$, $c=4$ for 100 Monte Carlo runs. The average ranks of $\widehat{\mathsf{M}}$ are reported with the associated standard errors in the parentheses. The $\mathsf{M} \in \mathbb{R}^{32 \times 32}$ are chosen to be a $32$ by $32$ pictures of Square, T and Cross correspondingly. The true basis type is Fourier. The fitted basis type is Chebyshev2.}
\label{sim1v1c6:mise:basis21:tab}
\end{table}

\begin{table}[htbp]
\centering
\begin{tabular}{rrrrrr}
\hline
 SNR & Basis & Square	&T	&Cross & ORACLE\\
\hline
1 & Chebyshev2 & 4.280(0.045) &4.740(0.044) &5.000(0.000) &50\\
5 & Chebyshev2 & 4.280(0.045) &4.740(0.044) &5.000(0.000) &50\\
10 & Chebyshev2 & 4.280(0.045) &4.740(0.044) &5.000(0.000) &50\\
\hline
1 & Fourier & 4.220(0.042) &4.790(0.041) &5.000(0.000) &50\\
5 & Fourier & 4.200(0.040) &4.790(0.041) &5.000(0.000) &50\\
10 & Fourier & 4.200(0.040) &4.790(0.041) &5.000(0.000) &50\\
\hline
\end{tabular}
\caption{Simulation results for $n=100$, $p=32$, $T=256$, $c=50$ for 100 Monte Carlo runs. The mean number of basis selected as well as the oracle number of basis is listed for each of basis type, Fourier or Chebyshev2, and SNR, 1, 5 or 10. Their associated standard errors are reported in the parentheses.}
\label{sim1v7:c:tab}
\end{table}
\begin{table}[htbp]
\centering
\begin{tabular}{rrrr}
\hline
MISE(sieve)	& Square	&T	&Cross\\
\hline
$\beta_1(t)$&1.083(0.015)&0.732(0.053)&0.526(0.012)\\
$\beta_2(t)$&1.071(0.017)&0.939(0.055)&0.792(0.013)\\
$\beta_3(t)$&1.084(0.016)&0.955(0.057)&0.811(0.013)\\
$\beta_4(t)$&1.746(0.049)&0.944(0.056)&0.786(0.011)\\
\hline
MISE(OLS)	& Square	&T	&Cross\\
\hline
$\beta_1(t)$&5.389(0.172)&4.106(0.091)&6.001(0.021)\\
$\beta_2(t)$&5.411(0.176)&6.018(0.162)&9.337(0.034)\\
$\beta_3(t)$&5.454(0.171)&6.054(0.161)&9.438(0.031)\\
$\beta_4(t)$&10.244(0.369)&5.945(0.157)&9.205(0.027)\\
\hline
\end{tabular}
\caption{Simulation results for $n=100$, $p=32$, $T=256$, $c=50$ and SNR = 1 for 100 Monte Carlo runs. The mean values of MISEs ($10^{-2}$) for 4 functional slope estimates are listed for the proposed method as well as the OLS method. Their associated standard errors are reported in the parentheses. For each method, the $\mathsf{M} \in \mathbb{R}^{32 \times 32}$ are chosen to be a $32$ by $32$ pictures of Square, T and Cross correspondingly. The true basis type is Chebyshev2. The fitted basis type is Chebyshev2.}
\label{sim1v7:snr1:mise:basis11:tab}
\end{table}
\begin{table}[htbp]
\centering
\begin{tabular}{rrrr}
\hline
MISE(sieve)	& Square	&T	&Cross\\
\hline
$\beta_1(t)$&1.071(0.014)&0.722(0.052)&0.514(0.011)\\
$\beta_2(t)$&1.065(0.015)&0.934(0.055)&0.782(0.011)\\
$\beta_3(t)$&1.071(0.015)&0.941(0.056)&0.790(0.011)\\
$\beta_4(t)$&1.736(0.048)&0.937(0.056)&0.774(0.010)\\
\hline
MISE(OLS)	& Square	&T	&Cross\\
\hline
$\beta_1(t)$&5.177(0.169)&3.935(0.086)&5.763(0.009)\\
$\beta_2(t)$&5.169(0.170)&5.813(0.155)&9.060(0.015)\\
$\beta_3(t)$&5.189(0.168)&5.829(0.154)&9.106(0.013)\\
$\beta_4(t)$&10.050(0.364)&5.795(0.153)&8.951(0.012)\\
\hline
\end{tabular}
\caption{Simulation results for $n=100$, $p=32$, $T=256$, $c=50$ and SNR = 5 for 100 Monte Carlo runs. The mean values of MISEs ($10^{-2}$) for 4 functional slope estimates are listed for the proposed method as well as the OLS method. Their associated standard errors are reported in the parentheses. For each method, the $\mathsf{M} \in \mathbb{R}^{32 \times 32}$ are chosen to be a $32$ by $32$ pictures of Square, T and Cross correspondingly. The true basis type is Chebyshev2. The fitted basis type is Chebyshev2.}
\label{sim1v7:snr5:mise:basis11:tab}
\end{table}
\begin{table}[htbp]
\centering
\begin{tabular}{rrrr}
\hline
MISE(sieve)	& Square	&T	&Cross\\
\hline
$\beta_1(t)$&1.069(0.014)&0.721(0.052)&0.513(0.011)\\
$\beta_2(t)$&1.065(0.015)&0.934(0.055)&0.781(0.011)\\
$\beta_3(t)$&1.069(0.015)&0.939(0.056)&0.787(0.011)\\
$\beta_4(t)$&1.735(0.048)&0.936(0.056)&0.772(0.010)\\
\hline
MISE(OLS)	& Square	&T	&Cross\\
\hline
$\beta_1(t)$&5.149(0.168)&3.914(0.086)&5.734(0.006)\\
$\beta_2(t)$&5.140(0.169)&5.789(0.154)&9.030(0.010)\\
$\beta_3(t)$&5.154(0.168)&5.800(0.153)&9.062(0.009)\\
$\beta_4(t)$&10.026(0.364)&5.778(0.152)&8.918(0.008)\\
\hline
\end{tabular}
\caption{Simulation results for $n=100$, $p=32$, $T=256$, $c=50$ and SNR = 10 for 100 Monte Carlo runs. The mean values of MISEs ($10^{-2}$) for 4 functional slope estimates are listed for the proposed method as well as the OLS method. Their associated standard errors are reported in the parentheses. For each method, the $\mathsf{M} \in \mathbb{R}^{32 \times 32}$ are chosen to be a $32$ by $32$ pictures of Square, T and Cross correspondingly. The true basis type is Chebyshev2. The fitted basis type is Chebyshev2.}
\label{sim1v7:snr10:mise:basis11:tab}
\end{table}
\begin{table}[htbp]
\centering
\begin{tabular}{rrrr}
\hline
MISE(sieve)	& Square	&T	&Cross\\
\hline
$\beta_1(t)$&7.386(0.053)&3.605(0.020)&5.356(0.068)\\
$\beta_2(t)$&7.357(0.055)&5.619(0.029)&8.148(0.075)\\
$\beta_3(t)$&7.400(0.056)&5.646(0.029)&8.201(0.077)\\
$\beta_4(t)$&17.320(0.074)&5.620(0.028)&8.108(0.074)\\
\hline
MISE(OLS)	& Square	&T	&Cross\\
\hline
$\beta_1(t)$&19.985(0.507)&18.658(0.374)&29.668(0.239)\\
$\beta_2(t)$&20.059(0.513)&28.123(0.573)&46.284(0.368)\\
$\beta_3(t)$&20.177(0.509)&28.228(0.572)&46.579(0.372)\\
$\beta_4(t)$&41.213(0.952)&27.899(0.567)&45.524(0.369)\\
\hline
\end{tabular}
\caption{Simulation results for $n=100$, $p=32$, $T=256$, $c=50$ and SNR = 1 for 100 Monte Carlo runs. The mean values of MISEs ($10^{-2}$) for 4 functional slope estimates are listed for the proposed method as well as the OLS method. Their associated standard errors are reported in the parentheses. For each method, the $\mathsf{M} \in \mathbb{R}^{32 \times 32}$ are chosen to be a $32$ by $32$ pictures of Square, T and Cross correspondingly. The true basis type is Chebyshev2. The fitted basis type is Fourier.}
\label{sim1v7:snr1:mise:basis12:tab}
\end{table}
\begin{table}[htbp]
\centering
\begin{tabular}{rrrr}
\hline
MISE(sieve)	& Square	&T	&Cross\\
\hline
$\beta_1(t)$&7.379(0.052)&3.592(0.018)&5.327(0.066)\\
$\beta_2(t)$&7.365(0.053)&5.610(0.027)&8.130(0.072)\\
$\beta_3(t)$&7.384(0.053)&5.622(0.027)&8.153(0.073)\\
$\beta_4(t)$&17.325(0.073)&5.611(0.027)&8.079(0.072)\\
\hline
MISE(OLS)	& Square	&T	&Cross\\
\hline
$\beta_1(t)$&19.588(0.500)&18.146(0.359)&28.711(0.228)\\
$\beta_2(t)$&19.578(0.502)&27.511(0.545)&45.155(0.354)\\
$\beta_3(t)$&19.630(0.500)&27.563(0.547)&45.291(0.355)\\
$\beta_4(t)$&40.946(0.946)&27.461(0.545)&44.539(0.352)\\
\hline
\end{tabular}
\caption{Simulation results for $n=100$, $p=32$, $T=256$, $c=50$ and SNR = 5 for 100 Monte Carlo runs. The mean values of MISEs ($10^{-2}$) for 4 functional slope estimates are listed for the proposed method as well as the OLS method. Their associated standard errors are reported in the parentheses. For each method, the $\mathsf{M} \in \mathbb{R}^{32 \times 32}$ are chosen to be a $32$ by $32$ pictures of Square, T and Cross correspondingly. The true basis type is Chebyshev2. The fitted basis type is Fourier.}
\label{sim1v7:snr5:mise:basis12:tab}
\end{table}
\begin{table}[htbp]
\centering
\begin{tabular}{rrrr}
\hline
MISE(sieve)	& Square	&T	&Cross\\
\hline
$\beta_1(t)$&7.375(0.052)&3.590(0.018)&5.323(0.066)\\
$\beta_2(t)$&7.364(0.053)&5.609(0.026)&8.129(0.072)\\
$\beta_3(t)$&7.378(0.053)&5.618(0.026)&8.145(0.072)\\
$\beta_4(t)$&17.323(0.073)&5.610(0.026)&8.075(0.071)\\
\hline
MISE(OLS)	& Square	&T	&Cross\\
\hline
$\beta_1(t)$&19.521(0.499)&18.071(0.358)&28.592(0.227)\\
$\beta_2(t)$&19.506(0.500)&27.424(0.541)&45.026(0.352)\\
$\beta_3(t)$&19.543(0.498)&27.462(0.543)&45.122(0.353)\\
$\beta_4(t)$&40.889(0.945)&27.398(0.542)&44.414(0.350)\\
\hline
\end{tabular}
\caption{Simulation results for $n=100$, $p=32$, $T=256$, $c=50$ and SNR = 10 for 100 Monte Carlo runs. The mean values of MISEs ($10^{-2}$) for 4 functional slope estimates are listed for the proposed method as well as the OLS method. Their associated standard errors are reported in the parentheses. For each method, the $\mathsf{M} \in \mathbb{R}^{32 \times 32}$ are chosen to be a $32$ by $32$ pictures of Square, T and Cross correspondingly. The true basis type is Chebyshev2. The fitted basis type is Fourier.}
\label{sim1v7:snr10:mise:basis12:tab}
\end{table}
\begin{table}[htbp]
\centering
\begin{tabular}{rrrr}
\hline
MISE(sieve)	& Square	&T	&Cross\\
\hline
$\beta_1(t)$&3.645(0.050)&1.327(0.006)&1.350(0.023)\\
$\beta_2(t)$&3.628(0.047)&1.610(0.009)&2.390(0.026)\\
$\beta_3(t)$&3.648(0.052)&1.596(0.009)&2.405(0.026)\\
$\beta_4(t)$&6.511(0.100)&1.600(0.009)&2.260(0.026)\\
\hline
MISE(OLS)	& Square	&T	&Cross\\
\hline
$\beta_1(t)$&13.922(0.036)&8.003(0.022)&10.108(0.026)\\
$\beta_2(t)$&14.027(0.039)&11.873(0.034)&16.103(0.040)\\
$\beta_3(t)$&14.058(0.045)&11.768(0.031)&16.054(0.043)\\
$\beta_4(t)$&26.994(0.044)&11.735(0.025)&15.617(0.036)\\
\hline
\end{tabular}
\caption{Simulation results for $n=100$, $p=32$, $T=256$, $c=50$ and SNR = 1 for 100 Monte Carlo runs. The mean values of MISEs ($10^{-2}$) for 4 functional slope estimates are listed for the proposed method as well as the OLS method. Their associated standard errors are reported in the parentheses. For each method, the $\mathsf{M} \in \mathbb{R}^{32 \times 32}$ are chosen to be a $32$ by $32$ pictures of Square, T and Cross correspondingly. The true basis type is Fourier. The fitted basis type is Chebyshev2.}
\label{sim1v7:snr1:mise:basis21:tab}
\end{table}
\begin{table}[htbp]
\centering
\begin{tabular}{rrrr}
\hline
MISE(sieve)	& Square	&T	&Cross\\
\hline
$\beta_1(t)$&3.624(0.048)&1.321(0.005)&1.337(0.022)\\
$\beta_2(t)$&3.616(0.047)&1.595(0.008)&2.371(0.024)\\
$\beta_3(t)$&3.625(0.049)&1.589(0.008)&2.378(0.024)\\
$\beta_4(t)$&6.508(0.098)&1.592(0.008)&2.256(0.024)\\
\hline
MISE(OLS)	& Square	&T	&Cross\\
\hline
$\beta_1(t)$&13.511(0.016)&7.738(0.009)&9.751(0.011)\\
$\beta_2(t)$&13.522(0.018)&11.497(0.015)&15.622(0.018)\\
$\beta_3(t)$&13.538(0.019)&11.452(0.014)&15.602(0.019)\\
$\beta_4(t)$&26.645(0.019)&11.458(0.011)&15.287(0.016)\\
\hline
\end{tabular}
\caption{Simulation results for $n=100$, $p=32$, $T=256$, $c=50$ and SNR = 5 for 100 Monte Carlo runs. The mean values of MISEs ($10^{-2}$) for 4 functional slope estimates are listed for the proposed method as well as the OLS method. Their associated standard errors are reported in the parentheses. For each method, the $\mathsf{M} \in \mathbb{R}^{32 \times 32}$ are chosen to be a $32$ by $32$ pictures of Square, T and Cross correspondingly. The true basis type is Fourier. The fitted basis type is Chebyshev2.}
\label{sim1v7:snr5:mise:basis21:tab}
\end{table}
\begin{table}[htbp]
\centering
\begin{tabular}{rrrr}
\hline
MISE(sieve)	& Square	&T	&Cross\\
\hline
$\beta_1(t)$&3.620(0.048)&1.320(0.005)&1.336(0.022)\\
$\beta_2(t)$&3.614(0.047)&1.593(0.008)&2.369(0.024)\\
$\beta_3(t)$&3.621(0.048)&1.589(0.008)&2.374(0.024)\\
$\beta_4(t)$&6.510(0.098)&1.590(0.008)&2.257(0.024)\\
\hline
MISE(OLS)	& Square	&T	&Cross\\
\hline
$\beta_1(t)$&13.457(0.011)&7.704(0.007)&9.705(0.008)\\
$\beta_2(t)$&13.459(0.012)&11.446(0.011)&15.559(0.013)\\
$\beta_3(t)$&13.470(0.013)&11.414(0.010)&15.545(0.013)\\
$\beta_4(t)$&26.607(0.014)&11.423(0.008)&15.248(0.011)\\
\hline
\end{tabular}
\caption{Simulation results for $n=100$, $p=32$, $T=256$, $c=50$ and SNR = 10 for 100 Monte Carlo runs. The mean values of MISEs ($10^{-2}$) for 4 functional slope estimates are listed for the proposed method as well as the OLS method. Their associated standard errors are reported in the parentheses. For each method, the $\mathsf{M} \in \mathbb{R}^{32 \times 32}$ are chosen to be a $32$ by $32$ pictures of Square, T and Cross correspondingly. The true basis type is Fourier. The fitted basis type is Chebyshev2.}
\label{sim1v7:snr10:mise:basis21:tab}
\end{table}
\begin{table}[htbp]
\centering
\begin{tabular}{rrrr}
\hline
MISE(sieve)	& Square	&T	&Cross\\
\hline
$\beta_1(t)$&1.091(0.015)&0.674(0.049)&0.555(0.012)\\
$\beta_2(t)$&1.097(0.019)&0.905(0.054)&0.834(0.014)\\
$\beta_3(t)$&1.096(0.016)&0.891(0.052)&0.847(0.014)\\
$\beta_4(t)$&1.723(0.049)&0.896(0.053)&0.805(0.014)\\
\hline
MISE(OLS)	& Square	&T	&Cross\\
\hline
$\beta_1(t)$&5.219(0.159)&4.266(0.088)&6.106(0.018)\\
$\beta_2(t)$&5.300(0.173)&6.343(0.152)&9.572(0.029)\\
$\beta_3(t)$&5.295(0.157)&6.283(0.152)&9.551(0.031)\\
$\beta_4(t)$&9.836(0.343)&6.246(0.149)&9.282(0.026)\\
\hline
\end{tabular}
\caption{Simulation results for $n=100$, $p=32$, $T=256$, $c=50$ and SNR = 1 for 100 Monte Carlo runs. The mean values of MISEs ($10^{-2}$) for 4 functional slope estimates are listed for the proposed method as well as the OLS method. Their associated standard errors are reported in the parentheses. For each method, the $\mathsf{M} \in \mathbb{R}^{32 \times 32}$ are chosen to be a $32$ by $32$ pictures of Square, T and Cross correspondingly. The true basis type is Fourier. The fitted basis type is Fourier.}
\label{sim1v7:snr1:mise:basis22:tab}
\end{table}
\begin{table}[htbp]
\centering
\begin{tabular}{rrrr}
\hline
MISE(sieve)	& Square	&T	&Cross\\
\hline
$\beta_1(t)$&1.072(0.014)&0.669(0.049)&0.545(0.012)\\
$\beta_2(t)$&1.073(0.015)&0.890(0.053)&0.819(0.012)\\
$\beta_3(t)$&1.074(0.014)&0.884(0.052)&0.824(0.012)\\
$\beta_4(t)$&1.685(0.045)&0.886(0.053)&0.800(0.012)\\
\hline
MISE(OLS)	& Square	&T	&Cross\\
\hline
$\beta_1(t)$&4.943(0.152)&4.092(0.082)&5.858(0.008)\\
$\beta_2(t)$&4.959(0.158)&6.095(0.146)&9.233(0.013)\\
$\beta_3(t)$&4.960(0.152)&6.069(0.146)&9.225(0.014)\\
$\beta_4(t)$&9.500(0.329)&6.067(0.145)&9.053(0.011)\\
\hline
\end{tabular}
\caption{Simulation results for $n=100$, $p=32$, $T=256$, $c=50$ and SNR = 5 for 100 Monte Carlo runs. The mean values of MISEs ($10^{-2}$) for 4 functional slope estimates are listed for the proposed method as well as the OLS method. Their associated standard errors are reported in the parentheses. For each method, the $\mathsf{M} \in \mathbb{R}^{32 \times 32}$ are chosen to be a $32$ by $32$ pictures of Square, T and Cross correspondingly. The true basis type is Fourier. The fitted basis type is Fourier.}
\label{sim1v7:snr5:mise:basis22:tab}
\end{table}
\begin{table}[htbp]
\centering
\begin{tabular}{rrrr}
\hline
MISE(sieve)	& Square	&T	&Cross\\
\hline
$\beta_1(t)$&1.071(0.014)&0.668(0.049)&0.544(0.012)\\
$\beta_2(t)$&1.072(0.015)&0.887(0.053)&0.817(0.012)\\
$\beta_3(t)$&1.072(0.014)&0.883(0.052)&0.821(0.012)\\
$\beta_4(t)$&1.686(0.045)&0.885(0.052)&0.801(0.012)\\
\hline
MISE(OLS)	& Square	&T	&Cross\\
\hline
$\beta_1(t)$&4.918(0.152)&4.070(0.082)&5.826(0.006)\\
$\beta_2(t)$&4.926(0.156)&6.062(0.145)&9.189(0.009)\\
$\beta_3(t)$&4.927(0.152)&6.043(0.145)&9.183(0.010)\\
$\beta_4(t)$&9.481(0.329)&6.045(0.144)&9.026(0.008)\\
\hline
\end{tabular}
\caption{Simulation results for $n=100$, $p=32$, $T=256$, $c=50$ and SNR = 10 for 100 Monte Carlo runs. The mean values of MISEs ($10^{-2}$) for 4 functional slope estimates are listed for the proposed method as well as the OLS method. Their associated standard errors are reported in the parentheses. For each method, the $\mathsf{M} \in \mathbb{R}^{32 \times 32}$ are chosen to be a $32$ by $32$ pictures of Square, T and Cross correspondingly. The true basis type is Fourier. The fitted basis type is Fourier.}
\label{sim1v7:snr10:mise:basis22:tab}
\end{table}

\end{appendix}

\end{document}